\def \VersionAuthor {}
\ifdefined \VersionWithComments
	\def\DraftVersion{}
\fi

\ifdefined\VersionAuthor
	\newcommand{\AuthorVersion}[1]{#1}
	\newcommand{\FinalVersion}[1]{}
\else
	\newcommand{\AuthorVersion}[1]{}
	\newcommand{\FinalVersion}[1]{#1}
\fi

\ifdefined\LaTeXdiff
	\let\VersionWithComments\undefined
	\let\WithReply\undefined
\fi
\ifdefined \VersionLong
	\newcommand{\LongVersion}[1]{\ifdefined\VersionWithComments{\color{red!40!black}#1}\else#1\fi}
	
\else
	\newcommand{\LongVersion}[1]{\ifdefined\VersionWithComments{\color{black!40}#1}\fi}
	
\fi

\documentclass[]{svjour3hacked}                    %
\smartqed  %

\usepackage[utf8]{inputenc}
\usepackage{amsmath,amssymb } %
\usepackage{stmaryrd} %
\usepackage{listings}
\usepackage{scalefnt}
\usepackage{pgfplots}%

\usepackage{bigstrut}
\usepackage{tabularx}
\usepackage{caption}
\usepackage{subcaption}

\usepackage{csp}

\usepackage{multirow}
\usepackage[ruled,vlined,linesnumbered]{algorithm2e}

\usepackage{tikz}
\usetikzlibrary{
    shapes.geometric,
    positioning,
    fit,
    calc
}

\usetikzlibrary{shapes.geometric}
\usetikzlibrary{decorations.markings}
\usetikzlibrary{arrows,automata,positioning,shapes,patterns}
\usetikzlibrary{trees,decorations.pathmorphing}
\usetikzlibrary{positioning}
\usetikzlibrary{calc}
\usetikzlibrary{arrows}
\usetikzlibrary{decorations.markings}

\tikzstyle{intermediate}=[dotted]

\usepackage{graphicx}

\definecolor{darkblue}{rgb}{0, 0, 0.7}

\usepackage[
		pdfauthor={Étienne André,
	Tian Huat Tan,
	Manman Chen,
	Shuang Liu,
	Jun Sun,
	Yang Liu
	and
	Jin Song Dong},%
		pdftitle={Automated synthesis of local time requirement for service composition},
		colorlinks=true,
		citecolor   = blue!50!blue,
		linkcolor   = darkblue,
		urlcolor    = blue!50!green,
	\ifdefined\DraftVersion
		pagebackref=true,
	\fi
	]{hyperref}

\ifdefined\VersionAuthor
	\usepackage[backend=biber,backref=true,style=alphabetic,url=false,doi=true,defernumbers=true,sorting=anyt,maxnames=99]{biblatex} %
	\renewbibmacro*{doi+eprint+url}{%
		\iftoggle{bbx:doi}
			{\color{black!40}\footnotesize\printfield{doi}}
			{}%
		\newunit\newblock
		\iftoggle{bbx:eprint}
			{\usebibmacro{eprint}}
			{}%
		\newunit\newblock
		\iftoggle{bbx:url}
			{\usebibmacro{url+urldate}}
			{}%
	}

\fi
\ifdefined\DraftVersion
	\usepackage{draftwatermark}
	\SetWatermarkText{draft}
	\SetWatermarkScale{3}
	\SetWatermarkColor[gray]{0.9}
\fi
\ifdefined\LaTeXdiff
	\usepackage{draftwatermark}
	\SetWatermarkText{diff}
	\SetWatermarkScale{3.5}
	\SetWatermarkColor[gray]{0.9}
\fi

\newtheorem{assumption}{Assumption}

\usepackage[capitalise,english,nameinlink]{cleveref}
\crefname{line}{\text{line}}{\text{lines}} %
\crefname{assumption}{\text{Assumption}}{\text{Assumptions}} %

\newcommand*{\Shift}{0.6ex}

\ifdefined \VersionWithComments
 	\definecolor{colorok}{RGB}{80,80,150}
\else
	\definecolor{colorok}{RGB}{0,0,0}
\fi
\definecolor{colorNOK}{RGB}{1,0,0}

\ifdefined \VersionWithComments
	\usepackage{marginnote}
	\newcommand{\marginX}{\marginnote{\huge{\quad\quad\textbf{!}\quad\quad}}}
	\newcommand{\ea}[1]{\mbox{}{\color{blue}\marginX{}\textbf{[\'Etienne}: #1]}}
	\newcommand{\tth}[1]{\mbox{}{\color{red}\marginX{}\textbf{[Tian Huat}: #1]}}
	\newcommand{\ls}[1]{\mbox{}{\color{orange}\marginX{}\textbf{[\'Liu Shuang}: #1]}}
	\newcommand{\instructions}[1]{{\color{red}\marginX{}\textbf{[Instructions: ``#1'']}}}
	\newcommand{\reviewer}[2]{\mbox{}{\color{red}\marginX{}\textbf{[Reviewer #1}: ``#2'']}}
	\newcommand{\todo}[1]{\mbox{}{\color{red}{\marginX{}\textbf{TODO}\ifx#1\\\else:\ \fi #1}}} %
\else
	\newcommand{\instructions}[1]{}
	\newcommand{\ea}[1]{}
	\newcommand{\ls}[1]{}
	\newcommand{\tth}[1]{}
	\newcommand{\reviewer}[2]{}
	\newcommand{\todo}[1]{}
\fi

\newcommand{\dLTC}{{\color{colorok}sLTC}}
\newcommand{\rLTC}{{\color{colorok}rLTC}}

\newcommand{\algoDLTC}{{\color{colorok}\ensuremath{\mathsf{synthSLTC}}}}
\newcommand{\algoRLTC}{{\color{colorok}\ensuremath{\mathsf{synthRLTC}}}}
\newcommand{\algoCheckSat}{{\color{colorok}\ensuremath{\mathsf{checkSat}}}}

\newcommand{\SynthesizeConstraint}{{\color{colorok}\ensuremath{\mathsf{synthRec}}}}

\newcommand{\assign}{\leftarrow}
\newcommand{\globalDelay}{T_{G}}
\newcommand{\localDelay}{C_{L}}
\newcommand{\issat}{Is\_Sat}

\newcommand{\sequential}{\,{\textbf{;}}\,}

\newcommand{\conditionArrow}{\Rightarrow}
\newcommand{\Process}{{\color{colorok}\ensuremath{P}}}
\newcommand{\NProcessInit}{N_0}
\newcommand{\ProcessInit}{\Process_0}
\newcommand{\ConstraintInit}{{\color{colorok}\ensuremath{\Constraint_0}}}

\newcommand{\Constraint}{{\color{colorok}\ensuremath{C}}}
\newcommand{\Delay}{{\color{colorok}\ensuremath{D}}}
\newcommand{\mystate}{{\color{colorok}\ensuremath{s}}}

\newcommand{\Rules}{{\color{colorok}\ensuremath{\mathsf{Rules}}}}

\newcommand{\Sequences}[1]{{\color{colorok}\ensuremath{\mathsf{Sequences}(#1)}}}
\newcommand{\ruleseq}{{\color{colorok}\ensuremath{\mathit{seq}}}}

\newcommand{\Kresult}{{\color{colorok}\ensuremath{K}}}
\newcommand{\adaptMech}{AM}

\newcommand{\prune}[2]{\mathit{prune}_{#1}(#2)}

\newcommand{\minterleave}[2]{#1 ||| #2}
\newcommand{\msequence}[2]{#1 \sequential #2}
\newcommand{\mconditional}[3]{#1 \dres #2 \rres #3}
\newcommand{\lreceive}{rec}
\newcommand{\lsInvoke}{sInv}
\newcommand{\laInvoke}{aInv}
\newcommand{\lreply}{reply}

\newcommand{\lstop}{Stop}

\newcommand{\rSInv}{rSInv}
\newcommand{\rPickOne}{rPickM}
\newcommand{\rPickTwo}{rPickA}
\newcommand{\rReply}{rReply}
\newcommand{\rRec}{rRec}
\newcommand{\rAInv}{rAInv}

\newcommand{\rCondOne}{rCond1}
\newcommand{\rCondTwo}{rCond2}
\newcommand{\rCondThree}{rCond3}
\newcommand{\rCondFour}{rCond4}

\newcommand{\rSeqOne}{rSeq1}
\newcommand{\rSeqTwo}{rSeq2}
\newcommand{\rFlowOne}{rFlow1}
\newcommand{\rFlowTwo}{rFlow2}

\newcommand{\receive}[1]{\lreceive(#1)}
\newcommand{\sInvoke}[1]{\lsInvoke(#1)}
\newcommand{\aInvoke}[1]{\laInvoke(#1)}
\newcommand{\reply}[1]{\lreply(#1)}

\newcommand{\mpick}[4]{pick(#1 \conditionArrow #2, alrm(#3) \conditionArrow #4)}
\newcommand{\mpicknew}[4]{pick(\myuplus\limits_{i=1}^{n}#1 \conditionArrow #2,\myuplus\limits_{j=1}^{k} alrm(#3) \conditionArrow #4)}

\newcommand{\newsInfo}{Stock Market Indices Service}
\newcommand{\newsInfoShortI}{SMIS}
\newcommand{\newsInfoShort}{SMIS}

\newcommand{\free}{\textit{FS}}
\newcommand{\paid}{\textit{PS}}
\newcommand{\dataS}{\textit{DS}}
\newcommand{\paratitle}[1]{\noindent\textbf{#1}. }
\newcommand{\Steps}[0]{ {\delta} }
\newcommand{\timelaps}[1]{#1^\uparrow}
\newcommand{\sequence}[1]{\langle #1 \rangle}
\newcommand{\activation}{\textit{Act}} %
\newcommand{\LTS}[1]{LTS_{#1}}

\newcommand{\LTSPlumpI}[1]{{\color{colorNOK}LTS'}}

\newcommand{\action}{{\color{colorok}\ensuremath{a}}}
\newcommand{\Clock}{X} %
\newcommand{\ClockSeq}{\ensuremath{\mathit{ClkSeq}}} %
\newcommand{\clock}{x} %
\newcommand{\clockcard}{h}
\newcommand{\clockval}{w} %
\newcommand{\compOp}{\bowtie}
\newcommand{\EventsSet}{{\color{colorok}\Sigma}} %
\newcommand{\Param}{{\color{colorok}\Lambda}} %
\newcommand{\param}{{\color{colorok}\lambda}} %
\newcommand{\paramCard}{{\color{colorok}m}} %
\newcommand{\npProcesses}{\mathcal{P}_{np}}
\newcommand{\Processes}{{\color{colorok}\ensuremath{\mathcal{P}}}}
\newcommand{\CS}{{\color{colorok}\ensuremath{\mathsf{CS}}}}
\newcommand{\States}{{\color{colorok}\ensuremath{S}}}
\newcommand{\pval}{{\color{colorok}\pi}} %
\newcommand{\Service}{{\color{colorok}\ensuremath{\mathsf{S}}}}
\newcommand{\Servicei}[1]{\Service_{#1}}
\newcommand{\sinit}{{\color{colorok}\ensuremath{s_0}}} %
\newcommand{\Succ}{\textit{succ}}
\newcommand{\ActiveClocks}{\mathit{aclk}}
\newcommand{\funIdle}{\textit{idle}}
\newcommand{\varrun}{\ensuremath{\rho}} %
\newcommand{\Val}{\textit{Valuations}}
\newcommand{\Valuation}{{\color{colorok}\ensuremath{v}}}
\newcommand{\Var}{{\color{colorok}\ensuremath{\mathcal{V}}}}
\newcommand{\variable}{{\color{colorok}\ensuremath{y}}} %
\newcommand{\varUnitialized}{{\color{colorok}\ensuremath{\bot}}}
\newcommand{\Vinit}{\Valuation_0} %

\newcommand{\projectP}[1]{\ensuremath{{#1}{\downarrow_{\Param}}}}

\newcommand{\Csharp}{C$\sharp$}

\newcommand{\ToolBPEL}{\textsc{Selamat}}

\newcommand{\grandn}{{\mathbb N}}
\newcommand{\grandq}{{\mathbb Q}}
\newcommand{\grandqplus}{{\mathbb Q}_{\geq 0}}
\newcommand{\grandqplusNotZero}{\grandq_{> 0}}
\newcommand{\grandr}{{\mathbb R}}
\newcommand{\grandrplus}{\grandr_{\geq 0}}

\newcommand{\setX}{\mathcal{C}_{\Clock}}
\newcommand{\setP}{\mathcal{C}_{\Param}}
\newcommand{\setXP}{\mathcal{C}_{\Clock \cup \Param}}
\newcommand{\onmsg}[2]{#1 \conditionArrow #2}
\newcommand{\onalrm}[2]{alrm(#1) \conditionArrow #2}

\newcommand{\setLX}{\mathcal{L}_{\Clock}}
\newcommand{\setLP}{\mathcal{L}_{\Param}}
\newcommand{\setLXP}{\mathcal{L}_{\Clock \cup \Param}}
\newcommand{\setNNCCP}{\mathcal{NC}_{\Param}}

\newcommand{\xml}[1]{\texttt{<#1>}}
\newcommand{\code}[1]{$\mathtt{#1}$}

\newcommand{\compSS}{\textit{SS}}
\newcommand{\compLS}{\textit{LS}}
\newcommand{\compIS}{\textit{IS}}
\newcommand{\compMS}{\textit{MS}}
\newcommand{\compBS}{\textit{BS}}

\newcommand{\compFS}{\textit{\free}}
\newcommand{\compHS}{\textit{HS}}

\newcommand{\modSM}{\textit{SM}}
\newcommand{\modRE}{\textit{RE}}
\newcommand{\compTS}{\textit{TS}}
\newcommand{\compWS}{\textit{WS}}
\newcommand{\compDS}{\textit{\dataS}}
\newcommand{\compCom}{$DS_{com}$}
\newcommand{\compSea}{$DS_{sea}$}
\newcommand{\compSet}{\textit{E}}

\newcommand{\mopick}{\begin{tikzpicture}[
oplus/.style={draw,circle, text width=0.004em,
  postaction={path picture={%
    \draw[black]
      (path picture bounding box.south west) -- (path picture bounding box.north east)
      (path picture bounding box.north west) -- (path picture bounding box.south east);}}}
]
\node [oplus] (d) {};
\end{tikzpicture}}
\newcommand{\mdecision}{\begin{tikzpicture}[
decision/.style = {diamond, draw,
    text width=0.004em,inner sep=2pt}
]
\node [decision] (d) {};
\end{tikzpicture}}

\newcommand{\eg}{\textcolor{colorok}{\textit{e.\,g.},}\xspace}
\newcommand{\ie}{\textcolor{colorok}{\textit{i.\,e.},}\xspace}
\newcommand{\viz}{\textcolor{colorok}{\textit{viz.},}\xspace}

\DeclareMathOperator*{\myuplus}{\uplus}

\begin{document}

\title{Automated synthesis of local time requirement for service composition\thanks{%
\AuthorVersion{%
	This is a pre-print of an article published in the \href{http://www.sosym.org/}{International Journal on Software and Systems Modeling (SoSyM)}.
	The final authenticated version is available online at:
			\href{https://www.doi.org/10.1007/s10270-020-00787-5}{10.1007/s10270-020-00787-5}.
}
	Étienne André, Jin Song Dong and Yang Liu are partially supported by CNRS STIC-Asie project CATS (``Compositional Analysis of Timed Systems'').
	Étienne André is partially supported by the ANR national research program ANR-14-CE28-0002 PACS (``Parametric Analyses of Concurrent Systems'').
	Étienne André and Jun Sun are partially supported by the ANR-NRF French-Singaporean research program ProMiS (ANR-19-CE25-0015).
}
}

\author{%
	Étienne André
	\and
	Tian Huat Tan
	\and
	Manman Chen
	\and
	Shuang Liu
	\and
	Jun Sun
	\and
	Yang Liu
	\and
	Jin Song Dong
}

\authorrunning{É.\ André, T.H.\ Tan, M.\ Chen, S.\ Liu, J.\ Sun, Y.\ Liu, and J.S.\ Dong} %

\institute{%
	É. André \at
	Université de Lorraine, CNRS, Inria, LORIA, Nancy, France
	\and
	T. H. Tan \at
	IBM, Singapore %
	\and
	M. Chen \at
	Autodesk, Singapore %
	\and
	S. Liu \at
	School of Software, Tianjin University, China
	\and
	J. Sun \at
	Singapore Management University, Singapore
	\and
	Y. Liu \at
	Nanyang Technological University, Singapore
	\and
	J. S. Dong \at
	National University of Singapore\\
	Griffith University, Australia
}

\ifdefined\VersionAuthor
	\date{}
\else
	\date{Received: date / Accepted: date}
\fi

\maketitle

\ifdefined\DraftVersion
	\textcolor{red}{\textbf{Draft version (date: \today{})}}
\fi

\ifdefined\VersionWithComments
	\tableofcontents{}
\fi

\begin{abstract}
	Service composition aims at achieving a business goal by composing existing service-based applications or components.
	The response time of a service is crucial especially in time critical business environments, which is often stated as a clause in service level agreements %
	between service providers and service users.
	To meet the guaranteed response time requirement of a composite service, it is important to select a feasible set of component services such that their response time will collectively satisfy the response time requirement of the composite service.
	In this work, we use the BPEL modeling language, that aims at specifying Web services.
	We extend it with timing parameters, and equip it with a formal semantics.
	Then, we propose a fully automated approach to synthesize the response time requirement of component services modeled using BPEL, in the form of a constraint on the local response times.
	The synthesized requirement will guarantee the satisfaction of the global response time requirement, statically or dynamically.
	We implemented our work into a tool, \ToolBPEL{}, and performed several experiments to evaluate the validity of our approach.
\end{abstract}

\ifdefined\VersionWithComments
	\textcolor{red}{\textbf{This is the version with comments. To disable comments, comment out line~3 in the \LaTeX{} source.}}
\fi

\section{Introduction and motivation}\label{sec:intro}
Service-oriented architecture is a paradigm where building blocks are used as services for software applications. Services encapsulate their functionalities, information, and make them available through a set of operations accessible over a network infrastructure using standards like SOAP~\cite{soap12} and WSDL~\cite{wsdl20}.
To make use of a set of services to achieve a business goal, service composition languages such as BPEL (Business Process Execution Language)~\cite{WSBPEL20} have been proposed.
A service that is composed by other services is called a \emph{composite} service, and services that the composite service makes use of are called \emph{component} services.

The requirement on the service response time is often an important clause in service-level agreements (SLAs) especially in business where timing is critical.
An SLA is a contract between service consumers and service providers specifying the expected quality of service (QoS) level.
Henceforth, we refer to the response time requirement of composite services as \emph{global time requirement}, and to the set of constraints on the response times of the component services as \emph{local time requirement}.
The response time of a composite service is highly dependent on that of each component service.
It is therefore crucial to derive local time requirements (\ie{} requirements for the component services) from the global time requirement, so that it will help in the selection of component services when building a composite service while satisfying the response time requirement.

An additional motivation for our work is that of micro-services.
As pointed out by~\cite{TCSLADX16}, many big players in the market (\eg{} Netflix, Amazon, and Microsoft Azure) have adopted microservice architecture \cite{Microservices} by decomposing their existing monolithic applications into smaller, and highly decoupled services (also known as microservices).
These services are then composed for fulfilling their business requirements.
For example, Netflix decomposed their monolithic DVD rental application into services that work together, and that stream digital entertainment to millions of Netflix customers every day.
Services of Netflix are hosted in a cloud provided by Amazon EC2 \cite{AmazonEC2}, which offers about 40 instance types\LongVersion{ (\eg{} t2.micro, c4.4xlarge) for selection.
All instance types provide the same functionality while having different QoS}.
The problem of composition of Web services with a large set of microservices is more and more relevant now, as the micro-services are getting more popular than ever (see \eg{} \cite{middlewareblog17,ST19}).
This justifies the use of techniques for which different services can be compared to and eventually selected.
Service-oriented architecture and micro-service architectures are conceptually similar: service-oriented architecture is a term that is used earlier and also widely used in literature. Micro-service architecture is more of a newer term that is used and practise widely in current industry, for the purpose of agile development.
    (For detailed comparison, see \eg{} \cite{CDP17}.)
The methods developed here are applicable to both service-oriented architecture and micro-service architecture.
\label{newtext:microservice}

\tth{
We could use this as an argument too to say that indeed this is a practical/relevant problem/solution, and it is even more relevant now where the micro-services are getting more popular than ever  (possibly some references like: could be used)
}

Consider an example of a stock indices service, which has an SLA with the subscribed users requiring that the stock indices shall be returned within three seconds upon request. The stock indices service makes use of several component services, including a paid service, for requesting stock indices.
The stock indices service provider would be interested in knowing the local time requirement of the component services, while satisfying the global response time requirement.
To avoid discarding any service candidates that might be part of a feasible composition, the synthesized local time requirement needs to be as \emph{weak} as possible, \ie{} to maintain as many combinations of local time requirements as possible.
This is crucial as having a faster service might incur a higher cost.

\begin{figure}[tb]
	\centering
	\includegraphics[width=2.5in]{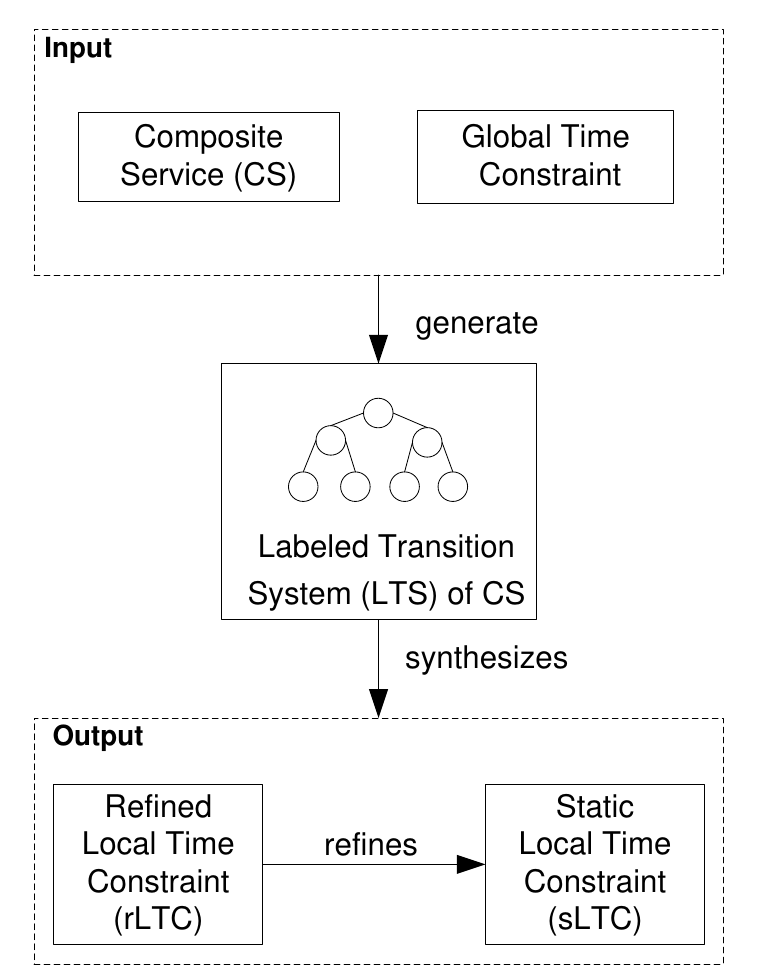}

	\caption{General approach}
	\label{fig:workflow}
\end{figure}

\subsection{Contribution}
In this paper, we present a fully automated technique to perform a rigorous model-based analysis of Web services, in order to synthesize the local time requirement in composite services.
Our approach performs an analysis of the composite service model behavior, using techniques inspired by parameter synthesis for timed systems.
Our synthesis approach does not only avoid bad scenarios in the service composition, but also guarantees the fulfillment of the global time requirement.

We use as a formalism %
BPEL,
which is a \emph{de-facto} standard language for specifying service composition.
BPEL supports control flow structures that involve complex timing constructs (\eg{} \xml{pick} control structure) and concurrent execution of activities (\eg{} \xml{flow} control structure).
Due to the non-determinism in both time and control flow, it is unknown which execution path will be executed at runtime.
Such a combination of timing constructs, concurrent calls to external services, and complex control structures, makes it a challenge to synthesize the local time requirement.
More precisely, response times of component services can be dependant; therefore, constraint between services may be of the form, \eg{} $t_\mathit{fs} > t_\mathit{hs}$ (for two parametric component service times), rendering the problem quite delicate.
In addition, there may be multiple possibilities of component services' response times that are satisfiable.
This can be particularly delicate to find out with only manual human inspection.\label{newtext:minorrevision}

\cref{fig:workflow} illustrates the main steps of our approach for synthesizing local time requirements.
The required inputs are the specification of the composite service, and its global time requirement.
The output will be local time requirements (at design time, and then at runtime) given in the form of a linear constraint.

We first propose a formal semantics for BPEL composite services augmented with timing parameters, \ie{} constants, the value of which is not known at design time;
this symbolic semantics is given the form of a labeled transition system (LTS).

Based on the LTS resulting from the input composite service, we then propose an approach to synthesize local time requirements of component services, represented as a (linear) constraint, which we refer to as the \emph{local time constraint}.
During the design phase of a composite service, the local time constraint is synthesized based on \emph{all} possible execution paths of the model, since it is unknown which execution path will be executed at runtime (this will depend on the dynamic behavior of the system%
).
The local time constraint of a composite service that is synthesized during the design time is called the \emph{static local time constraint} (hereafter \dLTC{}).

The synthesized \dLTC{} has several advantages.
Firstly, when creating a new composite service, it allows the selection of feasible services from a large pool of services with similar functionalities but different local response times.
Secondly, service designers can use the synthesized result to avoid over-approximations on the local response times, which may lead the service provider to purchase a service at a higher cost, while a service at a lower cost with a slower response time might have been sufficient to guarantee the global time requirement.
Thirdly, the local time requirements serve as a safe guideline when component services need to be substituted or new services need to be introduced.

\label{citetext:highly-evolving}Due to the highly evolving and dynamic environment which the composite service is running in, the design time assumptions for Web service composition, even if they are initially accurate, may later change at runtime.
For example, the execution time of a component service could violate the \dLTC{} due to reasons such as network congestion.
Nevertheless, this does not necessarily imply that the composite service will not satisfy the global time requirement.
	Indeed, the \dLTC{} is synthesized based on all possible execution paths at design time, whereas only one path will be executed at runtime.
	At runtime, some of the execution paths can be eliminated.
Therefore, we can use the runtime information to refine the \dLTC{} to make it weaker---which results in a more relaxed constraint\LongVersion{ on the response times of the component services}.
We refer to the \dLTC{} refined at runtime as the \emph{refined local time constraint} (hereafter \rLTC{}).
The \rLTC{} is used to decide whether the current composite service can still satisfy the global time requirement, despite some unplanned issues such as network congestion.

Our contributions are summarized as follows.
\begin{enumerate}
	\item We augment the BPEL modeling language with timing parameters, and we equip it with a formal semantics in the form of a labeled transition system.
	\item Given a composite service modeled using BPEL, we develop a sound method for synthesizing the local time requirement in the form of a set of constraints, which can be applied at the design stage of service composition.

    \item We introduce a refinement procedure on the \dLTC{} of a composite service based on the runtime information, which results in a more relaxed \rLTC{}. The \rLTC{} can be used to verify whether the composite service could still eventually satisfy the global time requirement at runtime.

	\item  We implement our algorithms into a tool \ToolBPEL{}.
	We then conduct experiments on several examples.
	The results show that the \rLTC{} can indeed help to improve the accuracy of the \dLTC{}.
	In addition, we show that the runtime adaptation does not incur much overhead in practice.
\end{enumerate}

\subsection{About this manuscript}

This manuscript is an extended version of~\cite{TanICSE13}.
We in fact rewrote most of the manuscript for a better readability.
The most notable differences between this manuscript and~\cite{TanICSE13} are:
\begin{enumerate}
	\item we replaced the formerly defined ``AOLTS'' with what we believe to be a simpler and more elegant presentation of labeled transition systems (LTS);
	\item we added details on our implementation and used more service composition examples; and,
	\item most importantly, we added a refinement procedure that attempts to meet the global time requirement at runtime even when the constraint computed statically is violated (\cref{sec:rrefine}).
\end{enumerate}

\subsection{Outline}
The rest of this paper is structured as follows.
\cref{sec:timeBpelExample} introduces a timed BPEL running example. %
\cref{sec:prelim} provides the necessary definitions and terminologies.
\cref{sec:dynamicAnalysis} introduces our formal semantics for BPEL extended with timing parameters.
\cref{sec:syncConstraint} presents the synthesis algorithms for \dLTC{}.
\cref{sec:rrefine} introduces \rLTC{}, and its usage for runtime adaptation of a service composition.
\cref{sec:evaluation} \LongVersion{presents our implementation and }evaluates our approach with four service composition examples.
\cref{sec:relatedWork} reviews related works.
Finally, \cref{sec:conclusion} concludes the paper, and outlines future work.

\section{A BPEL example with timed requirements}\label{sec:timeBpelExample}

\begin{figure}[t]
{\centering
\scalebox{.9}{
\begin{tikzpicture}[
oplus/.style={draw,circle, text width=0.5em,
  postaction={path picture={%
    \draw[black]
      (path picture bounding box.south west) -- (path picture bounding box.north east)
      (path picture bounding box.north west) -- (path picture bounding box.south east);}}},
block/.style = {rectangle, draw,
    text width=8em,align=center, minimum height=1.8em},
decision/.style = {diamond, draw,
    text width=1.5em,inner sep=0pt},
line/.style = {draw,thick, -latex'},
node distance=1.1cm and 0.4cm
]

\node (blank){};
\node [block,below of=blank, yshift=-0.25] (a) {Receive User};
\node [block, below of=a] (b) {Sync. Invoke DS};
\node [decision, below of=b] (c) {};
\node [block, below left of=c,xshift=-1.7cm, yshift=-0.3cm] (cl) {$\tick$ Reply indices};
\node [block, below right of=c,xshift=1.7cm, yshift=-0.3cm] (c2) {ASync.\ Invoke $\free$};
\node [oplus, below of=c2] (d) {};
\node [block, below left of=d,xshift=-1.5cm, yshift=-0.3cm](d1){OnMessage $\free$};
\node [block, text width=10em, below right of=d,xshift=1.5cm, yshift=-0.3cm](d2){OnAlarm 1 second};
\node [block, below of=d1] (e1) {$\tick$ Reply indices};
\node [block, text width=10em, below of=d2] (e2) {ASync. Invoke $\paid$};
\node [oplus, below of=e2] (f) {};
\node [block, below left of=f,xshift=-1.5cm, yshift=-0.3cm](f1){OnMessage $\paid$};
\node [block, text width=10em, below right of=f,xshift=1.5cm, yshift=-0.3cm](f2){OnAlarm 1 second};
\node [block, below of=f1] (g1) {$\tick$ Reply indices};
\node [block, below of=f2] (g2) {$\cross$ Reply `Failure'};
\path [line] (a) -- (b);
\path [line] (b) -- (c);
\path [line] (c) -| (cl) node [near start,anchor=south] {Indices exist};
\path [line] (c) -| (c2) node [near start,anchor=south] {~~~~~~~~~~Indices do not exist};
\path [line] (c2) -- (d);
\path [line] (d) -| (d1) node [near start,anchor=south]{};
\path [line] (d) -| (d2) node [near start,anchor=south]{};
\path [line] (d1) -- (e1);
\path [line] (d2) -- (e2);
\path [line] (e2) -- (f);
\path [line] (f) -| (f1) node [near start,anchor=south]{};
\path [line] (f) -| (f2) node [near start,anchor=south]{};
\path [line] (f1) -- (g1);
\path [line] (f2) -- (g2);

\end{tikzpicture}
}

}

\caption{\newsInfo{}}\label{fig:flowltsa}
\end{figure}

Let us introduce a \emph{\newsInfo{}} (\newsInfoShort) that will be used as a running example.
\newsInfoShort{} is a paid service and its goal is to provide updated stock indices to the subscribed users.
 It provides a service level agreement (SLA) to the subscribed users stating that it always responds within three seconds upon request.

\newsInfoShort{} has three component Web services, \ie{} a database service ($\dataS$), a free news feed service ($\free$) and a paid news feed service ($\paid$).
The strategy of the \newsInfoShort{} is calling the free service $\free$ before calling the paid service $\paid$ in order to minimize the cost.
Upon returning the result to the user,  the \newsInfoShort{} also stores the latest results in an external database service provided by $\dataS$ (storage of the results is omitted here).
The workflow of the \newsInfoShort{} is sketched in \cref{fig:flowltsa} in the form of a tree.
When a request is received from a subscribed customer (\code{Receive\ User}), it synchronously invokes (\ie{} invoke and wait for reply) the database service (\code{Sync.\ Invoke\ DS}) to request stock indices stored in the past minute.
Upon receiving the response from $\dataS$, the process reaches an \xml{if} branch (denoted by \mdecision).
If the indices are available (\code{Indices\ exist}), then they are returned to the user (\code{Reply\ indices}). Otherwise, $\free$ is invoked asynchronously (\ie{} the system moves on after the invocation without waiting for the reply).
A \xml{pick} construct (denoted by \mopick)\ls{Not sure if this bracket is necessary here. Since it represents a type to me (impressed by the notation <pick>).} is used here to await an incoming response (\xml{onMessage}) from previous asynchronous invocation or timeout (\xml{onAlarm}), whichever occurs.
If the response from $\free$ (\code{OnMessage\ \free}) is received within one second, then the result is returned to the user (\code{Reply\ indices}).
Otherwise, the timeout occurs (\code{OnAlarm\ 1\ second}), then \newsInfoShort{} stops waiting for the result from $\free$ and calls $\paid$  instead (\code{ASync.\ Invoke\ \paid}).
Similarly to $\free$, the result from $\paid$ is returned to the user, if the response from $\paid$ is received within one second.
Otherwise, it notifies the user regarding the failure of getting stock indices (\code{Reply\ }`\code{Failure}').
The states marked with a $\tick$ (resp.\ $\cross$) represent desired (resp.\ undesired) end states.

The global time requirement for \newsInfoShort{} is that \newsInfoShort{} should respond within three seconds upon request.
It is of particular interest to know the local time requirements for services $\paid$, $\free$, and $\dataS$, so as to fulfill the global time requirement.
This information can also help to choose a paid service $\paid$ which is both cheap and responds quickly enough.

In this example, an activity to avoid (which will be referred to as a ``bad activity'' in the following) is the reply activity that is triggered after the component service $\paid$ fails to respond within one second, which is marked with~$\cross$ in \cref{fig:flowltsa}.

\section{A formal model for parametric composite services}\label{sec:prelim}

\subsection{Variables, clocks, parameters, and constraints}\label{ssec:ActiveClocksocks}

Given a finite set $\Var$ of finite-domain \emph{variables}, a \emph{variable valuation} for $\Var$ is a function assigning to each variable a value in its domain.
We denote by~$\Val(\Var)$ the set of all variable valuations of~$\Var$.
Given a variable~$\variable \in \Var$ and a variable valuation~$\Valuation \in \Val(\Var)$, 
we denote by $\Valuation(\variable) = \varUnitialized$ the fact that variable~$\variable$ is uninitialized in valuation~$\Valuation$.

The clocks, parameters and constraints that we use in this work are similar to the ones used in the formalisms of (parametric) timed automata~\cite{AD94,AHV93} and (parametric) stateful timed CSP~\cite{SLDLSA13,ALSD14}.
Let  $\Clock= \{\clock_1, \dots, \clock_{\clockcard} \}$ (for some integer~$\clockcard$) be a finite set of clocks, \ie{} real-valued variables evolving at the same rate.
A \emph{clock valuation} is a function $\clockval : \Clock \rightarrow \grandrplus$, that assigns a non-negative real value to each clock.

Let $\Param = \{\param_1, \dots, \param_{\paramCard} \} $ (for some integer $\paramCard$) be a finite set of \emph{parameters}, \ie{} rational-valued constants that will be used here to represent the unknown response time of a component service.
A \emph{parameter valuation} is a function
$\pval : \Param \rightarrow \grandqplus$ assigning a non-negative rational value to each parameter.

Henceforth, we use $\clockval$ (resp.~$\pval$) to denote a clock (resp.\ parameter) valuation.

A \emph{linear term} over $\Clock \cup \Param$ is an expression of the form $\sum_{1 \leq i \leq N} \alpha_i z_i + d$ for some $N \in \grandn$, with $z_{i}\in \Clock \cup \Param$, $\alpha_{i} \in \grandqplus$ for $1 \leq i \leq N$, and $d \in \grandqplus$.
	We denote by $\setLXP$ the set of all linear terms over $\Clock$ and $\Param$.
	Similarly, we denote by \LongVersion{$\setLX$ (resp.\ }$\setLP$\LongVersion{)} the set of all linear terms over~\LongVersion{$\Clock$ (resp.\ }$\Param$\LongVersion{)}.
An \emph{inequality} over $\Clock$ and~$\Param$ is of the form $e \compOp e'$ where ${\compOp} \in \{<, \leq\}$, and $e$, $e' \in \setLXP$.

A \emph{convex constraint} (or \emph{constraint}) is a conjunction of inequalities.
We denote by $\setXP$ the set of all convex constraints over $\Clock$ and $\Param$.
	Similarly, we denote by \LongVersion{$\setX$ (resp.\ }$\setP$\LongVersion{)} the set of all convex constraints over~\LongVersion{$\Clock$ (resp.\ }$\Param$\LongVersion{)}.

Let $\Constraint \in \setXP$, $\Constraint[\pval]$ denotes the valuation of~$\Constraint$ with $\pval$, \ie{} the constraint over~$\Clock$ obtained by replacing each $\param \in \Param$ with~$\pval(\param)$ in~$\Constraint$.
Note that $\Constraint[\pval]$ can be written as
	$\Constraint \land \bigwedge_{\param_i \in \Param} \param_i = \pval(\param_i)$.
We say that $\clockval$ satisfies $\Constraint[\pval]$\LongVersion{, denoted by $\clockval \models \Constraint[\pval]$,} if the expression obtained by replacing each $\clock \in \Clock$ in~$\Constraint[\pval]$ with~$\clockval(\clock)$ evaluates to true.
	
Given $\Constraint \in \setXP$, we define $\timelaps{\Constraint}$ as the \emph{time elapsing} of~$\Constraint$, \ie{} the constraint over $\Clock$ and $\Param$ obtained from~$\Constraint$ by delaying all clocks by an arbitrary amount of time.
That is:
\[\timelaps{\Constraint} = \{ (\clockval', \pval) \mid \clockval \text{ satisfies } \Constraint[\pval] \land \forall \clock \in \Clock : \clockval'(\clock) = \clockval(\clock) + d, d \in \grandrplus \}\text{.}\]
Given $\Constraint \in \setXP$ and $\Clock' \subseteq \Clock$, we denote by $\prune{\Clock'}{\Constraint}$ the constraint in $\setXP$ that is obtained from $\Constraint$ by pruning the clocks in $\Clock'$; this can be achieved using variable elimination techniques such as Fourier-Motzkin (see, \eg{} \cite{s86}).
More generally, given $\Constraint \in \setXP$, we denote by $\projectP{\Constraint}$ the \emph{projection} of constraint~$\Constraint$
onto~$\Param$, \ie{} the constraint obtained from~$\Constraint$ by pruning all clock variables.
	Again, such a projection can be computed using Fourier-Motzkin elimination.

A \emph{non-necessarily convex constraint} (or NNCC) is a conjunction of disjunction of inequalities\footnote{%
	Without loss of generality, we assume here that all  NNCCs are in conjunctive normal form (CNF).};
	NNCCs are used to represent the synthesized local time constraint obtained via the methods proposed in this paper.
	Note that the negation of an inequality remains an inequality;
	however, the negation of a convex constraint becomes (in the general case) an NNCC.
We denote by $\setNNCCP$ the set of all NNCCs over~$\Param$.

Given $\Constraint \in \setNNCCP$, we say that $\pval$ \emph{satisfies} $\Constraint$, denoted by \mbox{$\pval \models \Constraint$}, if $\Constraint[\pval]$ evaluates to true.
$\Constraint$ is \emph{empty} if there does not exist a parameter valuation $\pval$ such that $\pval \models \Constraint$; otherwise $\Constraint$ is \emph{non-empty}.
Given two constraints $\Constraint_1, \Constraint_2\in \setNNCCP$, we say that $\Constraint_2$ is \emph{weaker} (or \emph{more relaxed}) than $\Constraint_1$, denoted by $\Constraint_1 \subs \Constraint_2$, if $\forall \pval:$ \mbox{$\pval \models \Constraint_1$} $\Rightarrow$ \mbox{$\pval \models \Constraint_2$}.

\subsection{Syntax of composite service processes}\label{ssec:syntax}

BPEL~\cite{WSBPEL20} is an industrial standard for implementing composition of existing Web services by specifying an executable workflow using predefined activities.
In this work, we assume the composite service is specified using the BPEL language.
Basic BPEL activities that communicate with component Web services are \xml{receive}, \xml{invoke}, and \xml{reply}, which are used to receive messages, invoke an operation of component Web services and return values respectively.
These activities are \emph{communication activities}.
The control flow of the service is defined using structural activities such as \xml{flow}, \xml{sequence}, \xml{pick} and \xml{if}.

A composite service $\CS$ makes use of a finite number of component services to accomplish a task.
Let $\compSet=\{\Service_1,\ldots,\Service_n\}$ be the set of all component services that are used by $\CS$.
In this work, we assume that the response time of a composite service is based on the time spent on individual communication activities, and the time incurred by internal operations of the composite service is negligible.\footnote{%
	We discuss the time incurred for internal operations in \cref{sec:discussion}.\ls{Do we need to justify why internal operation time can be neglected?}
}

Composite services are expressed using \emph{processes}.
We define a formal syntax definition in the following.
\begin{definition}\label{definition:process}
\emph{Processes} are defined using the following grammar:

\begin{tabular}{ c @{} c l l}
	$P \sdef$ & & $\receive{\Service}$ & receive activity \\
	& $\mid$ & $\reply{\Service}$ & reply  activity\\
	& $\mid$ & $\sInvoke{\Service}$ & synchronous invocation\\
	& $\mid$ & $\aInvoke{\Service}$ & asynchronous invocation\\
	& $\mid$ & $\minterleave{P}{Q}$ & concurrent activity \\
	& $\mid$ & $P\sequential Q$ & sequential activity \\
	& $\mid$ & $\mconditional{P}{b}{Q}$ & conditional activity \\
	& $\mid$ & $\mpicknew{\Service_i}{P_i}{a_j}{Q_j}$ & pick activity \\
\end{tabular}
\\
\noindent where $\Service$ is a component service, $P$ and $Q$ are composite service processes, $b$ is a Boolean expression, and $a_j \in \grandqplusNotZero$ are positive rational numbers, for $1 \leq j \leq k$.

\end{definition}

Let us describe below the BPEL syntax notations introduced in \cref{definition:process}:
\begin{itemize}
	\item $\receive{\Service}$ and $\reply{\Service}$ are used to denote ``receive from'' and ``reply to'' a service $\Service$, respectively;
	\item $\sInvoke{\Service}$ (resp{.}~$\aInvoke{\Service}$) denotes the synchronous (resp.\ asynchronous) invocation of  a component service $\Service$;
	\item  $\minterleave{P}{Q}$\ls{Just a concern, since the syntax definition is: $P = P|||Q$, it gives an impression of recursion, which is not intended in our case I believe. Similar for $P;Q$.  Maybe change it to $P = R|||Q$, where R and Q are BPEL activities. }\ea{agree! please do the modification throughout the manuscript} denotes the concurrent composition of BPEL activities $P$ and~$Q$;
	\item $P\sequential Q$ denotes the sequential composition of BPEL activities $P$ and $Q$;
	\item $\mconditional{P}{b}{Q}$ denotes the conditional composition, where $b$ is a guard condition on the process variables. If $b$ evaluates to true, BPEL activity $P$ is executed, otherwise activity $Q$ is executed;
	\item $\mpicknew{\Service_i}{P_i}{a_j}{Q_j}$ denotes the BPEL $pick$ composition, which contains two types of activities, \ie{} $onMessage$ activity and $onAlarm$ activity.
	An $onMessage$ activity $\onmsg{\Service_i}{P_i}$ is activated when the message from service $\Service_i$ arrives and BPEL activity $P_i$ is subsequently executed; an $onAlarm$ activity $\onalrm{a_j}{Q_j}$ is activated at $a_j$~time units, and BPEL activity $Q_j$ is subsequently executed.
	The $pick$ activity contains $n$ $onMessage$ activities and $k$ $onAlarm$ activities.
	Exactly one activity from these $n+k$ activities will be executed.
	If multiple activities are activated at the same time, one of the activities will be chosen non-deterministically for execution.
	Given a $pick$ activity $P$, we use $P.onMessage$ and $P.onAlarm$ to denote the $onMessage$ and $onAlarm$ branches of $P$ respectively.
\end{itemize}

A \emph{structural activity} is an activity that contains other activities. Concurrent, sequential, conditional, and pick activities are examples of structural activities.
An activity that does not contain other activities is called an \emph{atomic activity}, which includes receive, reply, synchronous invocation and asynchronous invocation activities.

Note that the communication activities can implicitly make use of variables for passing information.
For example, let $\Service$ be a component service that calculates the stock indices for a particular date.
For synchronous invocation $\sInvoke{\Service}$, it requires an input variable $v_i$ that specifies the date information, and an output variable $v_o$ to hold the return value from $\sInvoke{\Service}$.
To keep the notations concise, we abstract the usage and assignment of variables for communication activities.

We make the following assumption throughout this manuscript:

\begin{assumption}\label{assumption:bound}
	All loops have a bound on the number of iterations and on the execution time.
\end{assumption}

This assumption is necessary to ensure termination of our approach.
We believe it is reasonable in practice (see \cref{sec:discussion} for a discussion).

\subsection{Parametric composite service models}\label{ssec:LTS}

Let us now formally define composite service models and parametric composite service models.
Let $\npProcesses$ denote the set of all possible (non-parametric) composite service processes.

\begin{definition}[Composite service model]\label{def:CSM}
	A \emph{composite service model} $\CS$ is a tuple $(\Var, \Vinit,\NProcessInit)$, where $\Var$ is a finite set of variables, $\Vinit \in \Val(\Var)$ is an initial valuation that maps each variable to its initial value, and $\NProcessInit \in \npProcesses$ is a composite service process (defined according to the grammar of \cref{definition:process}) making use of the variables in~$\Var$.%
 \end{definition}

Each service comes with a \emph{response time}, which is a rational-valued constant, and can be seen as an upper bound on the time that a service needs to successfully return its answer.\ea{strange that it's not part of the definition!!! neither the set of services… neither the global times requirement (less critical though)}

\LongVersion{%
	Assume a component service $\Service$.
	Assume that the only communication activity that communicates with $\Service$ is the synchronous invocation activity $\lsInvoke(\Service)$.
	Upon invoking of service $\Service$, the activity $\lsInvoke(\Service)$ waits for the reply.
	The response time of $\Service$ is equivalent to the waiting time in $\lsInvoke(\Service)$.
	Therefore, by analyzing the time spent in $\lsInvoke(\Service)$, we can get the response time of  component service $\Service$.
}%
Given a composite service $\CS$, let  $t_i\in \grandrplus$ be the response time of component service $\Servicei{i}$ for $i\in\{1,\ldots,n\}$, and let $\compSet_t=\{t_1,\ldots,t_n\}$ be a set of component service response times that fulfill the global time requirement of service $\CS$.
Because $t_i$, for $i\in\{1,\ldots,n\}$, is a rational number, there are infinitely many possible values, even in a bounded interval (and even if one restricts these values to rational numbers).
A method to tackle this problem is to reason \emph{parametrically}, by considering these response times as unknown constants, or \emph{parameters}.

We now extend the definitions of services, composite service processes and composite service model to the parametric case.
First, a parametric service is a service~$\Service_i$, the response time of which is now a parameter $\param_i \in \Param$, instead of a rational-valued constant.
Then, a parametric composite service process is a service process defined according to the grammar of \cref{definition:process}, where services (``$\Service$'' in \cref{definition:process}) are now parametric services.
We denote by $\Processes$ the set of all possible parametric composite service processes.\ls{I notice this is a different font of P, which is used to define the composite service processes in Fig. 3. Just a suggestion, whether we need to use some notation which makes the difference more obvious?}
Finally, parametric composite service models are defined similarly to composite service models, except that the composite service processes are now parametric composite service processes.

 \begin{definition}[Parametric composite service model]\label{def:PCSM}
	A \emph{parametric composite service model} $\CS$\ls{It is the same notation as a composite service. } is a tuple $(\Var, \Vinit,\Param,\ProcessInit, \ConstraintInit)$, where $\Var$ is a finite set of variables; $\Vinit \in \Val(\Var)$ is an initial valuation that maps each variable to its initial value; $\Param$ is a finite set of parameters;
	$\ProcessInit \in \Processes$ is a parametric composite service process making use of the variables in~$\Var$ and $\ConstraintInit \in \setP$ is the initial %
	parametric constraint.
 \end{definition}

\begin{example}\label{example:PCSM}
	Let $\Var = \{ \variable_1\}$.
	Let $\Vinit$ be such that $\Vinit(\variable_1) = 0$.
	Let $\Param = \{ \param_1, \param_2 \}$.
	Let $\ProcessInit = \mconditional{ \mpick{\Service}{\sInvoke{\Service_1}}{1}{\sInvoke{\Service_2}} } {\variable_1 > 0} {\lstop}$.
	Let $\ConstraintInit = \param_1 < \param_2$.
	Let $\param_i$ denote the response time of $\sInvoke{\Service_i}$, $i \in \{ 1, 2 \}$.
	
	Then $\CS = (\Var, \Vinit,\Param,\ProcessInit, \ConstraintInit)$ is a parametric composite service model.
	
	\end{example}

\paragraph{Process and model valuation}
Given a parametric composite service process~$\Process$  with a parameter set $\Param=\{\param_1,\ldots,\param_\paramCard \}$ and given a parameter valuation  $(\pval(\param_1), \dots, \pval(\param_\paramCard))$, $\Process[\pval]$ denotes the \emph{valuation} of $\Process$ with~$\pval$, \ie{} the process where each occurrence of a parameter~$\param_i$ is replaced with its valuation~$\pval(\param_i)$.

Given a parametric composite service model $\CS$ with a parameter set $\Param=\{\param_1,\ldots,\param_\paramCard \}$, and given a parameter valuation  $(\pval(\param_1), \dots, \pval(\param_\paramCard))$, $\CS[\pval]$ denotes the \emph{valuation} of $\CS$ with~$\pval$, \ie{} the model $(\Var,\Vinit, \Param,\ProcessInit, C)$, where $C$ is $\ConstraintInit \land \bigwedge_{i = 1}^\paramCard (\param_i = \pval(\param_i))$.
Note that $\CS[\pval]$ can be seen as a non-parametric service model $(\Var,\Vinit, \ProcessInit[\pval])$.

\begin{example}
	Consider the parametric composite service model $\CS$ defined in \cref{example:PCSM}.
	Assume $\pval$ such that $\pval(\param_1) = 1$ and $\pval(\param_2) = 2$.
	Then $\ProcessInit[\pval] = \mconditional{ \mpick{\Service}{\sInvoke{\Service_1}}{1}{\sInvoke{\Service_2}} } {\variable_1 > 0} {\lstop}$, where the response time of $\sInvoke{\Service_1}$ is~1, and the response time of $\sInvoke{\Service_2}$ is~2.
	\end{example}

\subsection{Bad activities}\label{ssec:badAct}
Given a BPEL service $\CS$, we define a \emph{bad activity} as an atomic activity such that its execution leads the composite service $\CS$ to violate the global time requirement.
To distinguish bad activities, we allow the user to annotate a BPEL activity $A$ as a bad activity. %
The annotation can be achieved, for example, by using extension attributes of BPEL activities.
This work can be performed manually or using semi-automated procedures.

\begin{example}
	Consider again the example in \cref{sec:timeBpelExample}.
	Then ``Reply `Failure'{}'' is a bad activity, denoted in \cref{fig:flowltsa} by~$\cross$.
\end{example}
\section{A formal semantics for parametric composite services}\label{sec:dynamicAnalysis}

In this section, we provide our parametric composite service model with a formal semantics, defined in the form of a labeled transition system (LTS).
The semantics we use is inspired by the one proposed for (parametric) stateful timed Communicating Sequential Processes (CSP)~\cite{SLDLSA13,ALSD14}, that makes use of implicit clocks.

We first recall LTSs (\cref{ss:LTS}) and define symbolic states (\cref{ss:state}).
Following that, we define implicit clocks and the associated functions, \ie{} activation and idling (\cref{ss:clock}).
We then introduce our formal semantics (\cref{ssec:stateSpaceExploration}), and apply it to an example (\cref{subsec:example}).
We finally prove a technical result relating the non-parametric and the parametric service models (\cref{ss:theorems}).

\subsection{Labeled transition systems}\label{ss:LTS}

\begin{definition}[Labeled transition system] A \emph{labeled transition system (LTS)} %
	is a tuple $\LTS{} = (\States{}, \sinit, \EventsSet, \Steps)$, where
 	\begin{itemize}
		\item $\States{}$ is a set of states;
		\item  $\sinit \in \States{}$ is the initial state;
		\item  $\EventsSet$ is a set of actions; and
		\item $\Steps \subseteq \States{} \times \EventsSet \times \States{}$ is a transition relation.

\end{itemize}
\end{definition}

Given $\LTS{} = (\States{}, \sinit, \EventsSet, \Steps)$, a state $\mystate \in \States{}$ is  a \emph{terminal state} if there does not exist a state $\mystate' \in \States{}$ and an action $\action \in \EventsSet$ such that $(\mystate,\action,\mystate')\in \Steps$; otherwise, $\mystate$ is said to be a \emph{non-terminal state}.
There is a \emph{run} from a state $\mystate$ to state $\mystate'$, where $\mystate$, $\mystate'\in \States{}$, if there exists an alternating sequence of states and actions $\sequence{ \mystate_1, \action_1, \mystate_2, \ldots, \action_{n-1}, \mystate_n }$, where
		$\mystate_i \in \States{}$ for $1 \leq i \leq n$,
		$\action_i \in \EventsSet$ for $1 \leq i \leq n-1$,
		$\mystate_1=s$,
		$\mystate_n=\mystate'$, and
		$\forall i\in \{1,\ldots,n-1\}, (\mystate_i,\action_i,\mystate_{i+1})\in \Steps$.
A \emph{complete run} is a run that starts in the initial state $\sinit$ and ends in a terminal state.
Given a state $s \in \States{}$, we use $\Succ(\mystate)$ to denote the set of states reachable in one step from $\mystate$; formally,
$\Succ(\mystate)=\{\mystate' \mid \exists \action \in \EventsSet, \exists \mystate' \in \States{} : (\mystate,\action,\mystate') \in \delta\}$.

In the following, we introduce the notion of LTS starting from a state~$\mystate$ which is defined as the LTS containing $\mystate$ and all its successor states and transitions.

\begin{definition}[sub-LTS]\label{def:subLTS}
		Let $\LTS{} = (\States{}, \sinit, \EventsSet, \Steps)$ be an LTS, %
		and let $\mystate$ be a state of $\States{}$.
	The \emph{sub-LTS} of $\LTS{}$ starting from $\mystate$ is %
	$({\States{}}', s, \EventsSet', \Steps')$, where
	\begin{enumerate} %
		\item $\States{}' \subseteq \States{}$ is the set of states reachable from $\mystate \in \States{}$ in $\LTS{}$; %
		\item $\Steps' \subseteq \Steps$ is the transition relation satisfying the following condition:  $ (\mystate_1 , \action , \mystate_2) \in \Steps'$ if $\mystate_1, \mystate_2 \in \States{}'$ and $ (\mystate_1 ,\action , \mystate_2) \in \Steps$; and
		\item $\EventsSet' \subseteq \EventsSet$ is the set of all actions used in $\Steps'$, \ie{} $\{ \action \mid \exists \mystate_1, \mystate_2 \in \States' :  (\mystate_1 , \action , \mystate_2) \in \Steps'\}$.
	\end{enumerate}
\end{definition}

\subsection{Symbolic states}\label{ss:state}

In the following, we equip our parametric composite service models with a symbolic semantics, \ie{} a semantics, a run of which will capture a (possibly infinite) set of runs, for a (possibly infinite) set of parameter valuations.

Let us first define the notion of (symbolic) state of a parametric composite service model.

 \begin{definition}[State]\label{definition:state}
	Given a parametric composite service model $\CS = (\Var, \Vinit,\Param,\ProcessInit, \ConstraintInit)$, a (symbolic) \emph{state} of~$\CS$ %
	is a tuple $\mystate = (\Valuation, \Process,\Constraint,\Delay)$, where
		$\Valuation \in \Val(\Var)$ is a valuation of the variables,
		$\Process$ is a composite service process,
		$\Constraint$ is a constraint over $\setXP$,
		and
		$\Delay \in \setLP$ is the
	(parametric) elapsed time from the initial state~$\sinit$ to state~$\mystate$, excluding the idling time in state~$\mystate$. %
\end{definition}

Given a state $\mystate =(\Valuation,\Process,\Constraint,\Delay)$, we use the notation $\mystate.\Valuation$ to denote the field $\Valuation$ of~$\mystate$, and similarly for $\mystate.\Process$, $\mystate.\Constraint$ and $\mystate.\Delay$.
When a parametric composite service model $\CS$ has no variable, %
we denote each state $\mystate  \in \States{}$ by $(\Process,\Constraint,\Delay)$ for the sake of brevity.

\subsection{Implicit clocks}\label{ss:clock}

In order to provide parametric composite service models with a symbolic semantics, we use \emph{clocks} to record the elapsing of time.
Recall from \cref{ssec:ActiveClocksocks} that clocks are real-valued variables initially equal to~0, and evolving all at the same rate; some clocks may be reset to~0.
Clocks are used to record the time elapsing in several formalisms, in particular in timed automata (TAs)~\cite{AD94}.
In TAs, the clocks are defined as part of the models and state space.
It is known that the state space of the system may grow exponentially with the number of clocks and that the fewer clocks, the more efficient real-time model checking is~\cite{BY03}.
In (P)TAs, it is possible to dynamically reduce the number of clocks~\cite{DY96,Andre13FSFMA}.
An alternative approach is to define a semantics that create clocks on the fly when necessary, and prune them when they are no longer needed.
This approach was initially proposed for stateful timed CSP~\cite{SLDLSA13}, and shares similarities with \emph{firing times} in time Petri nets~\cite{Merlin74}.
This allows a smaller state space compared to the explicit clock approach.
We refer to this second approach~\cite{SLDLSA13} as the \emph{implicit clock approach}, and adopt this implicit clock approach in our work.

\subsubsection{Clock activation}
Clocks are implicitly associated with processes.
For instance, given a communication activity $\sInvoke{\Service}$, a clock starts measuring time once the activity becomes activated.
To introduce clocks on the fly, we define an activation function $\activation$ in the following definition, in the spirit of the one defined in~\cite{SLDLSA13,ALSD14}.

In short, this definition explains how to associate a new clock to a process: this clock will only be associated to the new processes with timing constraints, while it will not be associated to untimed processes nor to processes to which another implicit clock is already associated.

\begin{definition}\label{definition:activation}
	Given a process, we define the activation function $\activation$ using the following set of recursive rules:
	
	\begin{tabular}{llll}
     $\activation(A(\Service),\clock)$ & $=$ & $A(\Service)_\clock$ & A1 \\
    $\activation(mpick,x)$ & $=$ & $mpick_x $   & A2 \\
    $ \activation(A(\Service)_{x'},\clock)$ & $=$   & $A(\Service)_{x'}$   & A3 \\
    $ \activation(mpick_{x'},x)$   & $=$   & $mpick_{x'} $   & A4 \\
    $ \activation(P\fovr Q,\clock)$  & $=$   & $\activation(P,x)\fovr\activation(Q,x)$   & A5 \\
    $\activation(P\sequential Q, \clock)$   & $=$   & $\activation(P,\clock)\sequential Q$   & A6 \\
    \end{tabular}%
	\\
	\noindent
	where $A \in \{\lreceive,\lsInvoke,\laInvoke,\lreply\}$, $\fovr \in \{|||,\dres b\rres\}$,
	and
	$mpick = \mpicknew{\Service_i}{P_i}{a_j}{Q_j}$
\end{definition}

Let us explain \cref{definition:activation}.
Given a process~$P$, we denote by $P_\clock$ the corresponding process that has been associated with clock~$\clock$.
When a new state $\mystate$ is reached, the activation function is called to assign a new clock for each newly activated communication activity.
\begin{itemize}
	\item Rules A1 and A2 state that a new clock is associated with a BPEL communication activity~$A$ if~$A$ is newly activated.
	\item Rules A3 and A4 state that if a BPEL communication activity has already been assigned a clock, it will not be reassigned one.
	\item Rules A5 and A6 state that function $\activation$ is applied recursively to activate the child activities for BPEL structural activities.
	\item For rule A6, function $\activation$ is applied only to activity~$P$, but not to activity~$Q$, since activity~$P$ is the immediate subsequent activity (activity~$Q$ will be executed only after the completion of activity $P$).
\end{itemize}

\begin{example}\label{example:activation}
	Let $\Process = \minterleave{\sInvoke{\Service_1}}{\aInvoke{\Service_2}}$.
	Then, applying rules A5 and A1, $\activation(\Process, \clock) = \minterleave{\sInvoke{\Service_1}_{\clock}}{\aInvoke{\Service_2}_{\clock}}$.
	Note that $\clock$ is associated with both processes, as they are both simultaneously activated.
\end{example}
\begin{example}\label{example:activation2}
	Let $\Process = {\sInvoke{\Service_1}_{\clock'}} \sequential {\aInvoke{\Service_2}}$.
	Then, applying rules A6 and A3, $\activation(\Process, \clock) = {\sInvoke{\Service_1}_{\clock'}} \sequential {\aInvoke{\Service_2}}$.
	Indeed, the first invocation $\sInvoke{\Service_1}_{\clock'}$ is already associated to another clock~$\clock'$ (rule~A3) while the right-hand part of the sequence is not yet activated (rule~A6).
\end{example}

Given a process~$\Process$, we denote by $\ActiveClocks(\Process)$ the set of \emph{active clocks} associated with~$\Process$.

\begin{example}
	Assume process $\Process= \minterleave{\sInvoke{\Service_1}_{\clock_0}}{\sInvoke{\Service_2}_{\clock_1}}$.
	The set of active clocks associated with~$\Process$ is $ \ActiveClocks(\Process) = \{ \clock_0 , \clock_1 \}$.
\end{example}

\subsubsection{Idling function}\label{ssec:idlingfunc}
We define in \cref{definition:idling} below the function $\funIdle$ that, given a state $\mystate$, returns a constraint that specifies how long an activity can idle at state~$\mystate$.
The result is a constraint over $\Clock\cup\Param$.
This idling function is similar in essence to the \emph{time elapsing} on symbolic states (zones or parametric zones) defined for TAs or PTAs~\cite{BY03,HRSV02}.

\begin{definition}\label{definition:idling}
	Given a process, we define the idling function $\funIdle$ using the following set of recursive rules:

		\begin{tabular}{llll}
			$\funIdle(A(\Service)_x)$ & $=$ & $\clock \leq \param_{\Service}$ & I1\\
			$\funIdle(B(\Service)_x)$ & $=$ & $\clock = 0$ & I2\\
			$\funIdle(P\fovr Q)$ & $=$ & $\funIdle(P) \land \funIdle(Q)$ & I3\\
			$\funIdle(P\sequential Q)$ & $=$ & $\funIdle(P)$ & I4  \\	
			$\funIdle(mpick_x)$ & $=$ & $\clock \leq \param_{\Service} \land \bigwedge_{j = 0}^{k} \clock \leq a_j$ & I5  \\
		\end{tabular}
	\\
	where $ A \in \{\lreceive,\lsInvoke\}$,
	$B \in \{\laInvoke,\lreply\}$,
	$\fovr \in \{|||,\dres b\rres\}$,
	$mpick = \mpicknew{\Service_i}{P_i}{a_j}{Q_j}$, and
	$\param_{\Service}$ is the parametric response time of service $mpick_x$. %
\end{definition}

Let us explain \cref{definition:idling}.
\begin{itemize}
	\item Rule I1 considers the situation when the communication requires waiting for the response of a component service $\Service$, and the value of clock $x$ must not be larger than the response time parameter $\param_{\Service}$ of the service: that is, one can only remain in this state while $\clock \leq \param_{\Service}$ remains valid.
	\item Rule I2 considers the situation when no waiting is required: therefore, the clock constraint $\clock = 0$ implies that this state should be left within 0-time, as these actions are instantaneous.
	\item Rules I3 and I4 state that the function $\funIdle$ is applied recursively to activate the child activities of a BPEL structural activity.
	\item Similar to rule A6, for rule~I4, function $\activation$ is applied only to activity $P$, but not to activity $Q$, since only activity $P$ is executed next.
			Therefore, given a state $\mystate$ and activity $P\sequential Q$, we only need to consider how long the activity $P$ can idle at state~$\mystate$.
	\item Rule I5 states that the activity can idle only until $\param_{\Service}$ or any of the alarms~$a_j$ is reached.
			The conjunction comes from the fact that, as soon as any alarm reaches its time-out, then it will be triggered, therefore leading the system to leave this symbolic state.
\end{itemize}

\begin{example}
	Let $\Process = \minterleave{\sInvoke{\Service_1}}{\aInvoke{\Service_2}}$.
	Assume the response time of $\Service_i$ is $\param_i$ for $i \in \{ 1, 2 \}$.
	Recall from \cref{example:activation} that $\activation(\Process, \clock) = \minterleave{\sInvoke{\Service_1}_{\clock}}{\aInvoke{\Service_2}_{\clock}}$.
	Let us apply $\funIdle$ to $\activation(\Process, \clock)$.
	Applying rules I3, I1 and I2, we get $\clock \leq \param_1 \land \clock = 0$.
\end{example}

\subsection{Operational semantics}\label{ssec:stateSpaceExploration}

The operational semantics will be defined in the form of an LTS.
The actions labeling the LTS will be sequences of rules; these rules will be a set of rules (similar to those of parametric stateful timed CSP~\cite{ALSD14}) defining the transitions of the semantics, and will be explained below.
Let \[\Rules = \{ \rSInv,
\rRec,
\rReply,
\rAInv,
\rCondOne,
\rCondTwo,
\rCondThree,
\rCondFour,
\rSeqOne,
\rSeqTwo,\\% HACK
\rFlowOne,
\rFlowTwo,
\}
\cup (\rPickOne \times \grandn)
\cup (\rPickTwo \times \grandn)
\]%
\ls{I am wondering if we should explain the meaning of these notations first before using it here. I was a bit confused and was trying to look back to find the definitions of these notations when I first read the draft. }\ea{what do you want to explain? this may be a good idea indeed}
be the set of rules that will be used by the LTS.
Two rules ($\rPickOne$ and $\rPickTwo$) are associated with a positive integer, so as to remember which subprocess is derived (this will be explained later on).
Let $\Sequences{\Rules}$ denote the set of \emph{sequences} of rules, \ie{} non-empty ordered elements of~$\Rules$ (possibly used several times).
An example of a sequence of rule is $\sequence{\rRec, \rReply, \rRec, (\rPickOne, 2)}$.
Sequence concatenation is denoted by operator~$+$.

We can now define the semantics of a parametric composite service model in the form of an LTS.
Let $\ClockSeq = \sequence{\clock_0, \clock_1, \cdots}$ be a sequence of clocks.
We will need $\ClockSeq$ to pick a fresh clock when applying the clock activation function $\activation$ defined previously.
\begin{definition}[semantics of composite services]\label{def:semantics}
	Let $\CS = (\Var, \Vinit,\Param,\ProcessInit, \ConstraintInit)$ be a parametric composite service model.
	The \emph{semantics of $\CS$} (hereafter denoted by $\LTS{\CS{}}$) is the LTS
	$(\States{}, \sinit, \Sequences{\Rules}, \Steps)$ where
		$$\begin{array}{r @{\ = \ } l}
			\States{} & \{(\Valuation,\Process,\Constraint,\Delay) \in \Val(\Var) \times \Processes \times \setXP \times \setLP \}, \\
			\sinit & (\Vinit, \ProcessInit, \ConstraintInit, 0)\\
		\end{array}$$
	and the transition relation $\Steps$ is the smallest transition relation satisfying the following.
	For all $(\Valuation,\Process,\Constraint,\Delay) \in \States{}$, if $\clock$ is the first clock in the sequence~$\ClockSeq$ which is not in $\ActiveClocks(\Process)$, and
	$ (\Valuation, \activation(\Process, \clock), \Constraint \land \clock = 0, \Delay) \stackrel{\ruleseq}{\hookrightarrow}  (\Valuation',\Process',\Constraint',\Delay')$
	where $\Constraint'$ is satisfiable,
	then we have: $\big((\Valuation,\Process,\Constraint,\Delay), \ruleseq, (\Valuation', \Process', \prune{\Clock \setminus \ActiveClocks(\Process')}{\Constraint'},\Delay') \big) \in \Steps$.
\end{definition}

The transition relation $\hookrightarrow$ is specified by a set of rules, given in \cref{appendix:semantics}. %
Let us first explain these rules, after which we will go back to the explanation of \cref{def:semantics}.
The transition relation is labeled by a sequence of rules, that allows one to remember by using which sequence of rules a process evolves into another one.

\paratitle{Synchronous invocation}
	Rule $\rSInv{}$ states that a state $s=(\Valuation, \lsInvoke{(\Service)}_x, \Constraint, \Delay)$ may evolve into the state $\mystate'=(\Valuation',\lstop, (\clock=\param_{\Service})\wedge \timelaps{\Constraint}, D+\param_{\Service})$, %
	where $\lstop$ is the activity that does nothing, and $\param_{\Service}$ is the parametric response time of component service $\Service$.
	Note that, from\ls{in}\ea{no, ``from''; the definition states that the resulting must be satisfiable} \cref{def:semantics}, the condition $(\clock=\param_{\Service})\wedge \timelaps{\Constraint}$ is necessarily satisfied (otherwise this evolution is not possible).
	\LongVersion{%
		The resulting constraint is the intersection of constraints $\timelaps{\Constraint}$ and $\clock=\param_\Service$ (recall that the constraint $\timelaps{\Constraint}$ denotes the time elapsing of~$\Constraint$).
	}%
	Furthermore, the parametric duration from the initial state ($\Delay$) is incremented by~$\param_{\Service}$.
	Rules $\rRec$, $\rReply$ and $\rAInv$ are similar.

\paratitle{Pick activity} Rule $\rPickOne$ encodes the transition that takes place due to an $onMessage$ activity, where $\param_i$ denotes the parametric response time of $\Process_i$.
	Let us explain the constraint $(\clock=\param_i) \wedge \funIdle(mpick_x) \wedge \timelaps{\Constraint}$.
	First, after the transition, the current clock $\clock$ needs to be equal to the parametric response time of service~$\Service_i$, \ie{} $\clock = \param_i$.
	Second, the constraint $\funIdle(mpick_x)$ is added to ensure that $\clock$ remains smaller or equal to the maximum duration of the $mpick_\clock$ activity.
	Third, the constraint $\timelaps{\Constraint}$ denotes the time elapsing of~$\Constraint$.
	Observe that the transition in~$\hookrightarrow$ is labeled using the pair $(\rPickOne, i)$ so as to remember that the $i$th process (\ie{} $P_i$) has been selected.
	
	Rule $\rPickTwo$ (for an $onAlarm$ activity) is similar; observe that, instead of using the parametric response time, we use the time stipulated by the alarm (\ie{} $a_j$) of process~$Q_j$.

	\paratitle{Conditional activity} Given a conditional composition $\mconditional{A}{b}{B}$,
	the guard condition $b$ is a Boolean, hence its values are in $\{true, false\}$.
	As a consequence, given a valuation~$\Valuation$ of the variables, then $\Valuation(b) \in \{true, false,\varUnitialized\}$\LongVersion{ (recall that $\varUnitialized$ denotes an uninitialized variable)}.
	We have that $\Valuation(b)=\varUnitialized$ when the evaluation of $b$ is unknown, due to the fact that there may be uninitialized variables in~$b$.
	Since $b$ might be evaluated to either true or false at certain stages at runtime, we explore both activities $A$ and $B$ when $\Valuation(b)=\varUnitialized$ so as to reason about all possible scenarios.
	The case of $\Valuation(b)=\varUnitialized$ is captured by rules $\rCondOne$ and $\rCondTwo$, and the cases where $\Valuation(b)\in \{true, false\}$ are captured by rules $\rCondThree$ and $\rCondFour$.

	\paratitle{Sequential activity} $\rSeqOne$ states that if activity $A'$ is not a $\lstop$ activity (\ie{} activity $A'$ has not finished its execution), then a state containing activity $\msequence{A}{B}$ may evolve into a state containing activity $\msequence{A'}{B}$. Otherwise, if $A$ is a $\lstop$ activity (\ie{} activity $A$ has finished its execution), then the state\LongVersion{ containing activity $\msequence{A}{B}$} may \LongVersion{discharge activity $A$ and }evolve into\LongVersion{ a state containing}~$B$.
	This is captured by $\rSeqTwo$.

	\paratitle{Concurrent activity} For concurrent activity $\minterleave{A}{B}$, both activities $A$ and activity $B$ are executed.
	This is captured by $\rFlowOne$ and $\rFlowTwo$ respectively.
	$\rFlowOne$ states that if state $(\Valuation,A,C,D)$ can evolve into $(\Valuation',A',\Constraint',\Delay')$, then a state containing $\minterleave{A}{B}$ can evolve into a state containing $\minterleave{A'}{B}$, if $\Constraint' \land \funIdle(B)$ holds.
	That is, the clock constraints in $\Constraint'$ cannot exceed the duration activity $B$ can last for.
	Rule $\rFlowTwo$ is dual.

\medskip

Let us now explain \cref{def:semantics}.\label{explanation:def:semantics}
Starting from the initial state $\sinit = (\Vinit, \ProcessInit, \ConstraintInit, 0)$, we iteratively construct successor states as follows.
Given a state $(\Valuation, \Process, \Constraint, \Delay)$, a fresh clock~$\clock$ which is not currently associated with~$P$ is picked from~$\ClockSeq$.
The state $(\Valuation, \Process, \Constraint, \Delay)$ is transformed into $(\Valuation, \activation(\Process, \clock), \Constraint \land \clock = 0, \Delay)$, \ie{} timed processes which just become activated are associated with~$\clock$ and~$\Constraint$ is conjuncted with $\clock = 0$.
Then, a firing rule  is applied to get a target state $(\Valuation', \Process', \Constraint', \Delay')$. %
Lastly, clocks which do not appear within $\Process'$ are pruned from~$\Constraint'$.
More in details, the expression 
$\prune{\Clock \setminus \ActiveClocks(\Process')}{\Constraint'}$
denotes that we remove all clocks from the obtained constraint~$\Constraint'$ by existential quantification, except those which are still active in the successor~$\Process'$ of~$\Process$
(recall that $\prune{\Clock}{\Constraint}$ was defined in \cref{ssec:ActiveClocksocks}).

Observe that one clock is introduced and zero or more clocks may be pruned during a transition.
In practice, a clock is introduced only when necessary; if the activation function does not activate any subprocess, no new clocks are created.

\paragraph{Good and bad states}

Let us define good and bad states in the LTS obtained from \cref{def:semantics}.
The execution of a bad activity will make %
the execution of $\CS$ end in an undesired terminal state, which we refer to as a \emph{bad state}.
A terminal state which is not a bad state is called a \emph{good state}.
\LongVersion{%
	The synthesized local time requirement needs to guarantee the avoidance of all bad states and the termination of each run in a good state.
	The fact that each run must end in a good state is explained as follows:
	The non-determinism can be resolved at runtime depending on the variable values%
	, or the response time of a component service. %
	Therefore, we must guarantee that, regardless of the branch chosen by the composite service at runtime, it will end in a good state.
}

\subsection{Application to an example} \label{subsec:example}

\newcommand{\ltsko}{k_0}
\newcommand{\ltspo}{mpick}
\newcommand{\ltsso}[1]{\mystate_{#1}}
\newcommand{\ltsuser}{User}
\newcommand{\ltsreply}{reply}
\newcommand{\ltstime}{\param_{\paid}}
\newcommand{\ltsstop}{stop}
\newcommand{\ltsrgood}{r_{good}}
\newcommand{\ltsrbad}{r_{bad}}

Consider a composite service $\CS$ starting from $pick(\paid\conditionArrow \reply{\ltsuser}, alrm(1) \conditionArrow [\reply{\ltsuser}]_{bad})$.
Assume $\ltstime$ is the parametric response time of service~$\paid$.
(Note that $\CS$ is a part of the \newsInfoShort{} example from \cref{sec:timeBpelExample}.)
The states of $\CS$ computed according to \cref{def:semantics} are given in \cref{fig:PickActivity}, including intermediate states (detailed in the following).
Since $\CS$ has no variable, then $\Valuation=\emptyset$ in all states; therefore, we omit the component $\Valuation$ from all states for \LongVersion{the }sake of brevity.

\begin{figure}[t]
	\tikzset{bag/.style={text centered,yshift=-0.2cm},
 service/.style={align=left, text width=11cm}}
\tikzstyle{every edge}=[draw,->,font=\scriptsize]

{\centering
\begin{tikzpicture}[->,auto,node distance=2mm]
	\node[bag](s0){$s_0:(mpick, true, 0)$};
	\node[bag, below= of s0](s1){$s_0^x:(mpick_x,x=0,0)$};
	\node[bag, below = of s1,xshift=-8em](s2){$s_1':(r_{good},x = \ltstime \wedge x \leq 1, \ltstime)$};
	\node[bag, below = of s1,xshift=7em](s3){$s_2':(r_{bad}, x=1 \wedge x \leq \ltstime, 1)$};
	\node[bag, below = of s2](s4){$s_1:(r_{good},\ltstime\leq 1,\ltstime)$};
	\node[bag, below = of s3](s5){$s_2:(r_{bad},\ltstime\geq 1,1)$};
	\node[bag, below = of s4](s6){$s_1^x:((r_{good})_x,\ltstime\leq 1 \wedge x=0,\ltstime)$};
	\node[bag, below = of s5](s7){$s_2^x:((r_{bad})_x,\ltstime\geq 1 \wedge x=0, 1)$};
	\node[bag, below = of s6](s8){$s_3':(Stop,\ltstime\leq 1\wedge x=0 ,\ltstime)$};
	\node[bag, below = of s7](s9){$s_4':(Stop,\ltstime\geq 1\wedge x=0, 1)$};
	\node[bag, below = of s8](s10){$s_3:(Stop, \ltstime\leq 1,\ltstime)\tick$};
	\node[bag, below = of s9](s11){$s_4:(Stop, \ltstime\geq 1, 1)\cross$};
	\path (s0) edge[intermediate] node{(act)} (s1);
	\path (s1) edge[] node[left]{$\sequence{(\rPickOne,1)}$} (s2);
	\path (s1) edge[] node{$\sequence{(\rPickTwo,1)}$} (s3);
	\path (s2) edge[intermediate] node{(pruning)} (s4);
	\path (s3) edge[intermediate] node{(pruning)} (s5);
	\path (s4) edge[intermediate] node{(act)} (s6);
	\path (s5) edge[intermediate] node{(act)} (s7);
	\path (s6) edge[] node{$\sequence{\rReply}$} (s8);
	\path (s7) edge[] node{$\sequence{\rReply}$} (s9);
	\path (s8) edge[intermediate] node{(pruning)} (s10);
	\path (s9) edge[intermediate] node{(pruning)} (s11);
\end{tikzpicture}

} 

\noindent{}where $mpick =pick(\paid\conditionArrow\ltsrgood, alrm(1) \conditionArrow \ltsrbad)$, $\ltsrgood = \reply{\ltsuser}$, $ \ltsrbad=[\reply{\ltsuser}]_{bad}$, and $\ltstime$ is the parametric response time of service $\paid$.
	\caption{Computing states of service $\CS$ (including intermediate states)}
	\label{fig:PickActivity}
\end{figure}

\begin{itemize}
\item At state $\ltsso{0}$, the activation function assigns clock $x$ to record time elapsing of pick activity $\ltspo$, with $x$ initialized to zero.
	The tuple becomes the intermediate state $\ltsso{0}^x=(\ltspo_x, x=0,0)$.

\item From intermediate state $\ltsso{0}^x$, the process may evolve into the intermediate state
$\ltsso{1}'$ by applying the rule $\rPickOne{}$, if the constraint $\Constraint_1 =((x=\ltstime) \land idle(\ltspo_x) \land \timelaps{(x=0)})$, where $idle(\ltspo_x)=(x\leq \ltstime \land x\leq 1)$ and $\timelaps{(x=0)}$ (\ie{} $x\geq0$), is satisfiable.
	Intuitively, $\Constraint_1$ denotes the constraint where $\ltstime$ time units elapsed since clock~$\clock$ has started.
	In fact, $\Constraint_1$ is satisfiable (for example with $\ltstime=0.5$ and $x=0.5$).
	Therefore, it may evolve into the intermediate state $\ltsso{1}'=(\ltsrgood,(x = \ltstime) \wedge \funIdle(mpick_x) \wedge \timelaps{(x=0)}, \ltstime)=(\ltsrgood, (x=\ltstime) \land x \leq 1, \ltstime)$.
	Since clock $x$ is not used anymore in $\ltsso{1}'.P$ which is $\ltsrgood$, it is pruned.
	After pruning of clock variable $x$ and simplification of the expression, the intermediate state $\ltsso{1}'$ becomes the state $\ltsso{1}=(\ltsrgood, \ltstime \leq 1, \ltstime)$.

\item From intermediate state $\ltsso{0}^x$, the process may also evolve into the intermediate state $\ltsso{2}'$, by applying the rule $\rPickTwo{}$, if the constraint $\Constraint_2 =((x=1) \land idle(\ltspo_x) \land \timelaps{(x=0)})$, where $idle(\ltspo_x)=(x\leq \ltstime \land x\leq 1)$ and $\timelaps{(x=0)}$ (\ie{} $x\geq0$), is satisfiable.
It is easy to see that $\Constraint_2$ is satisfiable; therefore, the process may evolve into the intermediate state $\mystate_2'=(\ltsrbad, (x=1) \land x\leq \ltstime, 1)$. After clock pruning from intermediate state $\mystate_2'$, it becomes state $\mystate_2=(\ltsrbad,  \ltstime \geq 1, 1)$.

\item From state $\ltsso{1}$, activation function assigns clock $x$ to the reply activity $\ltsrgood$, and the process evolves into intermediate state $\ltsso{1}^x$. From\LongVersion{ intermediate state}~$\ltsso{1}^x$, the process may evolve into intermediate state $\ltsso{3}'$ by applying rule $\rReply{}$, if the constraint $\Constraint_3= ((x=0) \land \timelaps{(\ltstime \leq 1)})$ is satisfiable, where $\timelaps{(\ltstime \leq 1)}= \ltstime \leq 1$. In fact it is, and therefore it evolves into state $\ltsso{3}'=(Stop, \ltstime \leq 1 \land (x=0), \ltstime)$. After pruning of the non-active clock, it evolves into the terminal state $\ltsso{3}=(Stop, \ltstime \leq 1, \ltstime)$.
Since the terminal state is not caused by a bad activity, $\ltsso{3}$ is considered as a good state, denoted by $\tick$ in \cref{fig:PickActivity}.

  \item From state $\ltsso{2}$, the process may also evolve into the terminal state $\ltsso{4}=(Stop, \ltstime \geq 1, 1)$. %
  Since the terminal state is caused by a bad activity, it is considered as a bad state, denoted by~$\cross$ in \cref{fig:PickActivity}.
 \end{itemize}

Note that all states $\mystate_i^x$ and $\mystate_j'$, where $i, j\in \grandn$ and $0 \leq i \leq 4$, are intermediate states. State $\mystate_i^x$ is the state $\mystate_i$ after clock assignment operations are applied. State $\mystate_j'$ is the state $\mystate_j$ before clock pruning operations are applied.
These intermediate states are given in \cref{fig:PickActivity} to illustrate in details the\ls{of}\ea{I think ``the'' is correct only; unless ``in details'' is not correct?} application of the semantics.
The LTS of $\CS$ (without the intermediate states) is given in \cref{fig:PickActivitywithoutIn}.

\begin{figure}[t]
	\tikzset{bag/.style={text centered,yshift=-0.2cm},
 service/.style={align=left, text width=10cm}}
\tikzstyle{every edge}=[draw,->,font=\scriptsize]

{\centering
\begin{tikzpicture}[auto, node distance=2mm]
	\node[bag](s0){$s_0:(mpick, true, 0)$};
	\node[bag, below left= of s0](s4){$s_1:(r_{good},\ltstime\leq 1,\ltstime)$};
	\node[bag, below right= of s0](s5){$s_2:(r_{bad},\ltstime\geq 1,1)$};
	\node[bag, below = of s4](s10){$s_3:(Stop, \ltstime\leq 1,\ltstime)\tick$};
	\node[bag, below = of s5](s11){$s_4:(Stop, \ltstime\geq 1, 1)\cross$};
	
	\path (s0) edge[] node[left]{$\sequence{(\rPickOne,1)}$} (s4);
	\path (s0) edge[] node{$\sequence{(\rPickTwo,1)}$} (s5);
	\path (s4) edge[] node{$\sequence{\rReply}$} (s10);
	\path (s5) edge[] node{$\sequence{\rReply}$} (s11);
\end{tikzpicture}

}

where $mpick = pick(\paid \conditionArrow\ltsrgood, alrm(1) \conditionArrow \ltsrbad)$, $\ltsrgood = \reply{\ltsuser}$, $ \ltsrbad=[\reply{\ltsuser}]_{bad}$, and $\ltstime$ is the parametric response time of service $\paid$.
	\caption{LTS of service $\CS$}
	\label{fig:PickActivitywithoutIn}
\end{figure}

\subsection{A technical result: the reachability condition}\label{ss:theorems}

We defined the operational semantics of parametric composite service models as an LTS, the states of which contain information on clocks and parameters in the form of a constraint~$\Constraint$.
We now show that, for any reachable state of this LTS along a run, a parameter valuation $\pval$ satisfies~$\Constraint$ iff the model valuated with~$\pval$ has an equivalent run.
This is called the \emph{reachability condition}.
Similar results have been proved for parametric timed automata~\cite{HRSV02}, parametric time Petri nets~\cite{TLR09} or parametric stateful timed CSP~\cite{ALSD14}.

We first need several definitions and intermediate results.
Given a parametric service model~$\CS$ and a parameter valuation~$\pval$, let us relate runs of~$\LTS{\CS}$ and $\LTS{\CS[\pval]}$.
We will say that two runs are equivalent if they share the same discrete support, \ie{} follow the same application of sequences of rules regardless of the actual timing values.

\begin{definition}[equivalent runs]\label{def:equivalent-runs}
	Let $\CS$ be a parametric service model, and
	let~$\pval$ be a parameter valuation.
	
	Let $\varrun = \sequence{ (\Valuation_0,\Process_0,\Constraint_0,\Delay_0), \ruleseq_0, (\Valuation_1,\Process_1,\Constraint_1,\Delay_1), \ldots, \ruleseq_{n-1},$ $(\Valuation_n,\Process_n,\Constraint_n,\Delay_n) }$ be a run of $\LTS{\CS[\pval]}$.
	Let $\varrun' = \sequence{ (\Valuation_0',\Process_0',\Constraint_0',\Delay_0'), \ruleseq_0', (\Valuation_1',\Process_1',\Constraint_1',\Delay_1'), \ldots, \ruleseq_{n-1}', (\Valuation_n',\Process_n',\Constraint_n',$ $\Delay_n') }$ be a run of $\LTS{\CS}$.
	
	The two runs $\varrun$ and $\varrun'$ are \emph{equivalent} if
	$\Valuation_i = \Valuation_i'$ and $\Process_i = \Process_i'[\pval]$ for $0 \leq i \leq n$
	and
	$\ruleseq_i=\ruleseq_i'$ for $0 \leq i \leq n-1$.
\end{definition}

The following lemma states that, given a run of~$\LTS{\CS[\pval]}$, there exists a unique equivalent run in~$\LTS{\CS}$.
\begin{proposition}\label{proposition:unique-equivalent-run}
	Let $\CS$ be a parametric service model, and
	let~$\pval$ be a parameter valuation.
	Let $\varrun_\pval$ be a run of $\LTS{\CS[\pval]}$.
	
	Then there exists a unique run of $\LTS{\CS}$ equivalent to~$\varrun_\pval$.
\end{proposition}
\begin{proof}
	By induction on the length of the runs.
	We prove in fact a slightly stronger result:
		given a state $(\Valuation, \Process, \Constraint, \Delay)$ of a run~$\varrun$ in $\LTS{\CS[\pval]}$,
		and given a state $(\Valuation', \Process', \Constraint', \Delay')$ of the equivalent run~$\varrun'$ in $\LTS{\CS}$,
	we show that these two runs are not only equivalent, but also that $\Constraint \subseteq \Constraint'$.
	
\paratitle{Base case} From \cref{def:semantics}, the initial state of $\LTS{\CS}$ is $(\Vinit, \ProcessInit, \ConstraintInit, 0)$.
		The initial state of $\LTS{\CS[\pval]}$ is $(\Vinit, \ProcessInit[\pval], \ConstraintInit[\pval], 0)$.
		Since $\ConstraintInit[\pval] \subseteq \ConstraintInit$, then the result trivially holds.
		
\paratitle{Induction step}
Assume $\varrun_\pval$ is a run of $\LTS{\CS[\pval]}$ of length~$m$ reaching state $(\Valuation_1, \Process_1, \Constraint_1, \Delay_1)$;
		assume there exists a unique run of $\LTS{\CS}$ equivalent to~$\varrun_\pval$ and of length~$m$, reaching state $(\Valuation_1', \Process_1', \Constraint_1', \Delay_1')$.
		From \cref{def:equivalent-runs}, it holds that $\Valuation_1 = \Valuation_1'$ and $\Process_1 = \Process_1'[\pval]$.
		From the induction hypothesis, it holds that  $\Constraint_1 \subseteq \Constraint_1'$.

		Let $(\Valuation_2, \Process_2, \Constraint_2, \Delay_2)$ be the successor state of $(\Valuation_1, \Process_1, \Constraint_1, \Delay_1)$ via a given sequence of rules~$\ruleseq$ in~$\varrun_\pval$.

		Assume $(\Valuation_2, \Process_2, \Constraint_2, \Delay_2)$ is obtained from $(\Valuation_1, \Process_1, \Constraint_1, \Delay_1)$ by applying rule~$\rSInv$ in \cref{appendix:semantics}.
		Since $\Constraint_1 \subseteq \Constraint_1'$, then rule~$\rSInv$ can also be applied to $(\Valuation_1', \Process_1', \Constraint_1', \Delay_1')$, yielding a state $(\Valuation_2', \Process_2', \Constraint_2', \Delay_2')$.
		Now, we have:\\
		$\Constraint_1 \subseteq \Constraint_1' \Longrightarrow \timelaps{\Constraint_1} \subseteq \timelaps{\Constraint_1'}$\\
		$\Longrightarrow (x=\pval(\param_{\Service}) \wedge \timelaps{\Constraint_1}) \subseteq (x=\param_{\Service} \wedge \timelaps{\Constraint_1'})$\\
		$\Longrightarrow \Constraint_2 \subseteq \Constraint_2'$.\\
		In particular, $\Constraint_2 \subseteq \Constraint_2'$ implies that $\Constraint_2'$ is non-empty, hence the state $(\Valuation_2', \Process_2', \Constraint_2', \Delay_2')$ is a valid state.
		In addition, since $\Process_1 = \Process_1'[\pval]$ and rule~$\rSInv$ derives to $\lstop$, then $\Process_2 = \Process_2'[\pval]$.
		Variables are updated in the same manner on both sides, hence $\Valuation_2 = \Valuation_2'$.\ea{but there could be some non-determinism in variable updates, don't you think so?}%
		\tth{Please forget about my previous comments, actually after rereading it again, I think that the proof is correct. On the other hand, what could likely result in the non-determinism of variable updates that you could think of?}
		\ea{I suspect there might be some non-deterministic variable modification, but this is not a real problem (if this happens, we can still remember the ``past'', \ie{} which exact sequence of transition was taken in the process, so as to guarantee uniqueness of the symbolic counterpart)}\tth{Ok, then we can leave that first, i can't immediately think of any case like that}
		The proof is similar for other rules in \cref{appendix:semantics}.

		Finally, %
		the successor state $(\Valuation_2', \Process_2', \Constraint_2', \Delay_2')$ is the unique successor state of $(\Valuation_1', \Process_1', \Constraint_1', \Delay_1')$ in $\LTS{\CS}$ via this sequence of rules.
		Hence there exists a unique run of $\LTS{\CS}$ equivalent to~$\varrun_\pval$ and of \mbox{length~$m+1$}.
	
\end{proof}

We now prove the dual result.
\cref{proposition:unique-equivalent-run-opposite} states that, given a run~$\varrun$ of~$\LTS{\CS}$, there exists a unique equivalent run in~$\LTS{\CS[\pval]}$, provided $\pval$ satisfies the parametric constraint associated with the last state of~$\varrun$.
\LongVersion{%
	The result is dual to what we proved in \cref{proposition:unique-equivalent-run}.
}

\begin{proposition}\label{proposition:unique-equivalent-run-opposite}
	Let $\CS$ be a parametric service model, and
	let~$\pval$ be a parameter valuation.
	Let $\varrun$ be a run of $\LTS{\CS}$ ending in a state $(\Valuation_n, \Process_n, \Constraint_n, \Delay_n)$.
	
	For any $\pval \models \projectP{\Constraint_n}$, there exists a unique run of $\LTS{\CS[\pval]}$ equivalent to~$\varrun$.
\end{proposition}
\begin{proof}
	By induction on the length of the runs.
	We prove in fact a slightly stronger result:
		given a state $(\Valuation', \Process', \Constraint', \Delay')$ of a run~$\varrun'$ in $\LTS{\CS}$,
		and given a state $(\Valuation, \Process, \Constraint, \Delay)$ of the equivalent run~$\varrun$ in $\LTS{\CS[\pval]}$,
	we show that these runs are not only equivalent, but also that $\Constraint = \Constraint'[\pval]$.

Base Step: From \cref{def:semantics}, the initial state of $\LTS{\CS}$ is $(\Vinit, \ProcessInit, \ConstraintInit, 0)$.
		The initial state of $\LTS{\CS[\pval]}$ is $(\Vinit, \ProcessInit[\pval], \ConstraintInit[\pval], 0)$.
		Since $\ConstraintInit=true$ then $\ConstraintInit = \ConstraintInit[\pval]$.
		Hence the result trivially holds in that case.

	 Induction step: Assume $\varrun$ is a run of $\LTS{\CS}$ of length~$m$ reaching state $(\Valuation'_1, \Process'_1, \Constraint'_1, \Delay'_1)$.
		Let $(\Valuation_2', \Process_2', \Constraint_2', \Delay_2')$ be the successor state of $(\Valuation_1', \Process_1', \Constraint_1', \Delay_1')$ via a sequence of rules~$\ruleseq$ in~$\varrun$.
		Let $\pval \models \projectP{\Constraint'_2}$.
		Assume there exists a unique run of $\LTS{\CS[\pval]}$ equivalent to~$\varrun$ and of length~$m$, reaching state $(\Valuation_1, \Process_1, \Constraint_1, \Delay_1)$.
		From \cref{def:equivalent-runs}, it holds that $\Valuation_1 = \Valuation_1'$ and $\Process_1 = \Process_1'[\pval]$.
		From the induction hypothesis, it holds that $\Constraint_1 = \Constraint_1'[\pval]$.

		Assume $(\Valuation_2', \Process_2', \Constraint_2', \Delay_2')$ is obtained from $(\Valuation_1', \Process_1', \Constraint_1', \Delay_1')$ by applying rule~$\rSInv$ in \cref{appendix:semantics}.
		Recall that $\Constraint_1 = \Constraint_1'[\pval]$;
		since $\Process_1 = \Process_1'[\pval]$ (from \cref{def:equivalent-runs}), we can apply rule~$\rSInv$ to $(\Valuation_1, \Process_1, \Constraint_1, \Delay_1)$, yielding a state $(\Valuation_2, \Process_2, \Constraint_2, \Delay_2)$.
		From \cref{appendix:semantics}, we know that $\Constraint_2 = (x=\param_{\Service} \wedge \timelaps{(\Constraint_1}))$ and $\Constraint_2' = (x=\param_{\Service} \wedge (\timelaps{\Constraint_1'}))$.
		Now, we have:\\
		$ \Constraint_2 = \big(x=\pval(\param_{\Service}) \wedge \timelaps{(\Constraint_1})\big)
			\\~~~~= \big(x=\pval(\param_{\Service}) \wedge \timelaps{(\Constraint_1'[\pval])}\big)$\hfill(induction hypothesis)$
			\\~~~~= \big(x=\pval(\param_{\Service}) \wedge \timelaps{(\Constraint_1' \land \bigwedge_{\param_i \in \Param} \param_i = \pval_i)} \big)$\hfill(definition of valuation)$
			\\~~~~= \big(x=\pval(\param_{\Service}) \wedge \timelaps{(\Constraint_1')}\big) \land \bigwedge_{\param_i \in \Param} \param_i = \pval_i ~~$\hfill(property of time elapsing)$
			\\~~~~= \big(x=\param_{\Service} \wedge \timelaps{(\Constraint_1')}\big) \land \bigwedge_{\param_i \in \Param} \param_i = \pval_i ~~$\hfill(definition of valuation)$
			\\~~~~= \Constraint_2' \land \bigwedge_{\param_i \in \Param} \param_i = \pval_i$\hfill(definition of $\Constraint_2'$)$
			\\~~~~= \Constraint_2'[\pval]$\hfill(definition of valuation)$
			$
		\\
		\ea{%
			I fixed the proof; there is still the claim ``property of time elapsing'' that is unproved here, but I cannot find a proof I had written in some report a long time ago; let's see whether some reviewer asks for it (in which case I'll add a lemma in the preliminaries)
			\\
		}
		Note that adding $x=\param_{\Service}$ while keeping satisfiability of the expression is only true because $\pval \models \projectP{\Constraint_2'}$.
		This implies that $\Constraint_2'[\pval]$ is non-empty, hence the state $(\Valuation_2, \Process_2, \Constraint_2, \Delay_2)$ is a valid state.
		In addition, since $\Process_1 = \Process_1'[\pval]$ and rule~$\rSInv$ derives to $\lstop$, then $\Process_2 = \Process_2'[\pval]$.
		Similarly, variables are updated in the same manner on both sides, hence $\Valuation_2 = \Valuation_2'$.
		The proof is similar for other rules in \cref{appendix:semantics}.

		The proof of uniqueness is identical to that of~\cref{proposition:unique-equivalent-run}.
	
\end{proof}

\cref{proposition:unique-equivalent-run,proposition:unique-equivalent-run-opposite} give the following theorem.

\begin{theorem}[reachability condition]\label{theorem:reachability}
	Let $\CS$ be a parametric service model, and
	let~$\pval$ be a parameter valuation.
	Let $\varrun$ be a run of $\LTS{\CS}$ ending in a state $(\Valuation_n, \Process_n, \Constraint_n, \Delay_n)$.
	
	There exists a run of $\LTS{\CS[\pval]}$ equivalent to~$\varrun$ iff $\pval \models \projectP{\Constraint_n}$.
\end{theorem}
\section{Synthesizing the static LTC}\label{sec:syncConstraint}

Given $\CS=(\Var, \Param,\ProcessInit, \ConstraintInit)$, the \emph{global time requirement} for $\CS$ requires that, for every state $(\Valuation, \Process,\Constraint,\Delay)$ reachable from the initial state $( \Vinit,  \ProcessInit,\ConstraintInit,0)$ in its LTS, the constraint $\Delay \leq  \globalDelay$ is satisfied, where $\globalDelay \in \grandrplus$ is the \emph{global time constraint}.
The \emph{local time requirement} requires that if the response times of all component services of $\CS$ satisfy the \emph{local time constraint} (LTC) $\localDelay \in \setP$, then the service $\CS$ satisfies the global time requirement.

In this section, given a global time constraint $\globalDelay$ for a service $\CS$, we present an approach to synthesize the static LTC (\dLTC{}) $\localDelay$\LongVersion{ based on the LTS}.
The \dLTC{} will be given in the form of an NNCC over~$\Param$.
We show that if the response times of all component services of $\CS$ satisfy the local time requirement, then the service $\CS$ will end in a good state within $\globalDelay$ time units.

\subsection{Motivation}

Let $\param_i\in \grandqplus$ be the parametric response time of component service $\Servicei{i}$ for $i\in\{1,\ldots,n\}$, and let $\Param = \{\param_1,\ldots,\param_n\}$ be the set of component service parametric response times.
Using constraints over~$\Param$, we can represent an infinite number of possible response times symbolically.
The local time requirement of component services of $\CS{}$ is specified as a constraint over~$\Param$.
An example of a local time requirement is $(\param_1 \leq 6) \land (\param_2 \leq 5)$.
This local time requirement specifies that, in order for $\CS$ to satisfy the global time requirement, service $\Servicei{1}$ needs to respond within 6 time units, and service $\Servicei{2}$ needs to respond within 5 time units.
A local time requirement can also be in the form of a dependency between parametric response times, \eg{} $(\param_2 \leq \param_1 \implies \param_1 + \param_2\leq 6) \land (\param_1 \leq \param_2 \implies \param_1\leq 6)$.
\LongVersion{%
	This example requires that when service $\Servicei{2}$ responds not slower than service $\Servicei{1}$, then the sum of the response times of services $\Servicei{1}$ and $\Servicei{2}$ must be at most 6 seconds; however, if service $\Servicei{1}$ responds not slower than service $\Servicei{2}$, then service $\Servicei{1}$ must respond within 6 seconds.
}

In the following, we will propose a technique to synthesize the static LTC in the form of a convex over~$\Param$.
We first give an intuition concerning how to handle the good states (\cref{ss:goodstates}) and the bad states (\cref{ss:badstates});
then, we give the full synthesis algorithm (\cref{ssec:theaglorithm}), apply it to an example (\cref{ssec:ltcexample}) and prove its soundness (\cref{ssec:soundness}).

\subsection{Addressing the good states}\label{ss:goodstates}

We assume a composite service $\CS$ and its LTS $\LTS{\CS{}}=(\States{}, \sinit, \Sequences{\Rules}, \Steps)$; let $\States{}_{good}$ be the set of all good states of~$\LTS{\CS{}}$.
\LongVersion{%
	Given $\LTS{\CS{}}$, our goal is to synthesize the local time requirement for service $\CS$.
}We make two observations here.
First, from \cref{theorem:reachability}, a good state $\mystate_g=(\Valuation_g, \Process_g,\Constraint_g,\Delay_g) \in \States{}_{good}$ is reachable from the initial state $\sinit$ iff $\Constraint_g$ is satisfiable.
Second, whenever the good state $\mystate_g$ is reached, we require that the total delay from initial state $\sinit$ to state $\mystate_g$ must be no larger than the global time constraint $\globalDelay$, \ie{} $\Delay_g \leq \globalDelay$.
To sum up, given a good state $\mystate_g=(\Valuation_g, \Process_g,\Constraint_g,\Delay_g)$ where $\mystate_g\in \States{}_{good}$, we require the constraint $(\projectP{\Constraint_g} \implies (\Delay_g \leq \globalDelay))$ to hold.
The constraint means that whenever $\mystate_g$ is reachable from~$\sinit$, the total (parametric) delay from~$\sinit$ to $\mystate_g$ must be less than the global time constraint $\globalDelay$.
The synthesized \dLTC{} for $\CS$ must include the conjunction of such constraints for each good state $\mystate_g\in {\States{}}_{good}$, that is:
$$\bigwedge_{(\Valuation_g, \Process_g,\Constraint_g,\Delay_g)\in {\States{}}_{good}} (\projectP{\Constraint_g} \implies (\Delay_g \leq \globalDelay))\text{.}$$

\begin{example}\label{example:good}

	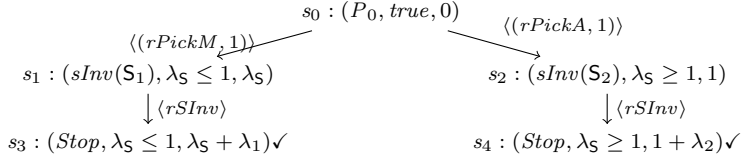
\begin{figure}[t]
	\centering
	\tikzset{bag/.style={text centered,yshift=-0.2cm},
 service/.style={align=left, text width=10cm}}
\tikzstyle{every edge}=[draw,->,font=\scriptsize]

{\centering
\begin{tikzpicture}[auto, node distance=2mm]
	\node[bag](s0){$s_0:(\ProcessInit, true, 0)$};
	\node[bag, below left= of s0](s4){$s_1:(\sInvoke{\Service{}_1},\param_{\Service{}}\leq1,\param_{\Service{}})$};
	\node[bag, below right= of s0](s5){$s_2:(\sInvoke{\Service{}_2},\param_{\Service{}}\geq1,1)$};
	\node[bag, below = of s4](s10){$s_3:(Stop, \param_{\Service{}}\leq1,\param_{\Service{}}+\param_1)\tick$};
	\node[bag, below = of s5](s11){$s_4:(Stop,\param_{\Service{}}\geq1,1+\param_2)\tick$};
	
	\path (s0) edge[] node[left]{$\sequence{(\rPickOne,1)}$} (s4);
	\path (s0) edge[] node{$\sequence{(\rPickTwo,1)}$} (s5);
	\path (s4) edge[] node{$\sequence{\rSInv}$} (s10);
	\path (s5) edge[] node{$\sequence{\rSInv}$} (s11);
\end{tikzpicture}

}
	\caption{LTS of composite service $\CS{}$}\label{fig:pickltsgood}
	\end{figure}

	Let us consider a composite service $\CS{}$ whose process component is $ \ProcessInit = \mpick{\Service}{\sInvoke{\Service_1}}{1}{\sInvoke{\Service_2}}$%
	, where $\Service$ is a component service.
	Assume that $\sInvoke{\Service_j}$ is a component service with parametric response time $\param_j$, for $j \in \{1,2\}$, and $\Service$ has a response time~$\param_\Service$.
	Suppose the global time requirement of the composite service $\CS{}$ is to respond within five seconds.
	\cref{fig:pickltsgood} shows the LTS of~$\CS{}$.

	For composite service $\CS{}$ in \cref{fig:pickltsgood}, we have two good states (states $\mystate_3$ and $\mystate_4$), and the synthesized local time requirement for composite service $\CS{}$ is:

	\[(\param_{\Service}\leq 1) \implies (\param_{\Service}+\param_1 \leq 5) \land (\param_{\Service} \geq 1) \implies (1+\param_2 \leq 5)\]

\end{example}

\subsection{Addressing the bad states}\label{ss:badstates}
Another goal we want to achieve is to avoid all bad states in $\LTS{\CS{}}$. Let $\States{}_{bad}$ be the set of all bad states of service $\LTS{\CS{}}$.
Given a bad state $\mystate_b=(\Valuation_b, \Process_b,\Constraint_b,\Delay_b) \in \States{}_{bad}$, this bad state must not be reachable from the initial state $\sinit$.
Hence, in order to prevent $\Constraint_b$ to be satisfiable, we require that the parameters be taken in the negation of the projection of $\Constraint_b$ onto~$\Param$, \ie{} we require that $\neg \projectP{(\Constraint_b)}$ be satisfiable because of the reachability condition (\cref{theorem:reachability}).
In addition to the good state constraint given in \cref{ss:goodstates}, the synthesized \dLTC{} for $\CS$ must also include the conjunction of such constraints for each bad state $\mystate_b\in {\States{}}_{bad}$, that is:
$$\bigwedge_{(\Valuation_b, \Process_b,\Constraint_b,\Delay_b)\in {\States{}}_{bad}} \big(\neg \projectP{(\Constraint_b)}\big)\text{.}$$

\begin{example}
	Consider a variant~$\CS'$ of \cref{example:good}, where $\sInvoke{\Service_2}$ is now treated as a bad activity, denoted by $[\sInvoke{\Service_2}]_{bad}$.
	This service results in the LTS shown in \cref{fig:pickltsbad}, where state $\mystate_4$ is a bad state.
\LongVersion{%
	Note that the constraint $\mystate_4.\Constraint=\param_{\Service}\geq 1$ is introduced by the $pick$ activity.}
	From \cref{theorem:reachability}, a way to avoid the reachability of $\mystate_4$ is to negate its associated constraint~$\Constraint$. %
	Therefore, the local time requirement for composite service $\CS{}'$ is $(\projectP{\mystate_3.\Constraint} \implies (\mystate_3.\Delay  \leq  \globalDelay)) \land \lnot (\projectP{\mystate_4.\Constraint})$:
	the first term guarantees the reachability of~$\mystate_3$ while the second term guarantees the non-reachability of~$\mystate_4$.
	Therefore, this NNCC ensures that any complete run of the service ends in a good state.
	(This will be proved in \cref{ssec:soundness}.)

	\begin{figure}[t]
		\centering
		\tikzset{bag/.style={text centered,yshift=-0.2cm},
 service/.style={align=left, text width=10cm}}
\tikzstyle{every edge}=[draw,->,font=\scriptsize]

{\centering
\begin{tikzpicture}[auto, node distance=2mm]
	\node[bag](s0){$s_0:(\ProcessInit', true,0)$};
	\node[bag, below left= of s0](s4){$s_1:(\sInvoke{\Service{}_1},\param_{\Service{}} \leq 1,\param_{\Service{}})$};
	\node[bag, below right= of s0](s5){$s_2:([\sInvoke{\Service{}_2}]_{bad},\param_{\Service{}} \geq 1,1)$};
	\node[bag, below = of s4](s10){$s_3:(Stop, \param_{\Service{}} \leq 1,\param_{\Service{}}+\param_1)\tick$};
	\node[bag, below = of s5](s11){$s_4:(Stop,\param_{\Service{}} \geq 1,1+\param_2)\cross$};
	
	\path (s0) edge[] node[left]{$\sequence{(\rPickOne,1)}$} (s4);
	\path (s0) edge[] node{$\sequence{(\rPickTwo,1)}$} (s5);
	\path (s4) edge[] node{$\sequence{\rSInv}$} (s10);
	\path (s5) edge[] node{$\sequence{\rSInv}$} (s11);
\end{tikzpicture}

}
		\caption{LTS of composite service $\CS{}'$}
		\label{fig:pickltsbad}
	\end{figure}
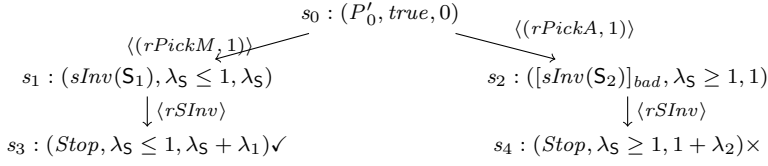
\end{example}

\subsection{Synthesis algorithms}\label{ssec:theaglorithm}

\cref{algo:ltc} presents the entry algorithm for synthesizing the \dLTC{} for a given service $\CS$, by traversing the LTS of $\CS$.
\cref{algo:ltc} simply calls $\SynthesizeConstraint(\mystate)$ applied to the initial state~$\sinit$; this latter algorithm $\SynthesizeConstraint$ is given in \cref{algo:synConstraint}.

\begin{algorithm}[t]
	\SetKwInOut{Input}{input}\SetKwInOut{Output}{output}
	
		\Input{Composite service model $\CS$ with LTS $\LTS{\CS{}}$ of initial state $\sinit$}

	\Output{The \dLTC{} $\localDelay \in \setNNCCP$}

	\BlankLine
	
	\Return{$ \SynthesizeConstraint(\sinit)$}\nllabel{algo:ltr:2}\;

\caption{$\algoDLTC(\CS)$\label{algo:ltc}}
\end{algorithm}
Given a state $\mystate=(\Valuation, \Process, \Constraint,\Delay)$ in the LTS of service $\CS$, $\SynthesizeConstraint(\mystate)$ returns a parameter constraint as follows.
If state $\mystate$ is a good state (\cref{algo:syn:1}), then it returns the constraint $\projectP{\mystate.\Constraint} \implies (\mystate.\Delay \leq \globalDelay)$ (\cref{algo:syn:2}), where $\globalDelay$ is the given global time constraint of the service $\CS$.
If state $\mystate$ is a bad state (\cref{algo:syn:3}), then the negation of the current constraint $\projectP{\mystate.\Constraint}$ is returned (\cref{algo:syn:4}).
Finally, if $\mystate$ is a non-terminal state (\cref{algo:syn:6}), the algorithm returns the conjunction of the result of the algorithm recursively applied on the successors of~$\mystate$ (\cref{algo:syn:8}).

\begin{algorithm}[t]

	\SetKwInOut{Input}{input}\SetKwInOut{Output}{output}
	
		\Input{State $\mystate$ of LTS}

	\Output{The constraint for LTS that starts at $\mystate$}

	\BlankLine

\uIf{$\mystate$ is a good state\nllabel{algo:syn:1}} {

 \Return{$\big(\projectP{\mystate.\Constraint} \implies (\mystate.\Delay \leq \globalDelay) \big)$}\nllabel{algo:syn:2}\;
}

\uElseIf{$\mystate$ is a bad state \nllabel{algo:syn:3}}{
		\Return{$\lnot (\projectP{\mystate.\Constraint})$}\nllabel{algo:syn:4}\;
 }
 \Else{
	\tcp{$\mystate$ is a non-terminal state\nllabel{algo:syn:6}}

  \Return{$\bigwedge_{\mystate' \in \Succ(\mystate)} \SynthesizeConstraint(\mystate') $ \nllabel{algo:syn:8}}\;
 }

\caption{$\SynthesizeConstraint(\mystate)$}
\label{algo:synConstraint}

\end{algorithm}
\subsection{Application to the running example}\label{ssec:ltcexample}

\begin{figure}[t]
{\centering
 \tikzset{
  bag/.style={text centered},
  aux/.style={font=\footnotesize},
  service/.style={align=left, text width=8cm}
}
\begin{tikzpicture}[node distance=3mm and 7mm,-stealth,scale=0.8,transform shape]
\node[bag] (s0) {$s_0:(S,true,0)$};
\node[bag,below =of s0] (s1) {$s_1:(\mconditional{\ltsrgood}{b}{A_1},true,\param_{\dataS})$};
\node[bag,below =of s1] (s2) {$s_2:(A_1,true,\param_{\dataS})$};
\node[bag,right =of s2] (s3) {$s_3:({\ltsrgood},true,\param_{\dataS})$};
\node[bag,below =of s2, yshift=-12mm] (s4) {$s_4:(P_1,true,\param_{\dataS})$};
\node[bag,below =of s3] (s5) {$s_5:(Stop,true,\param_{\dataS})\tick$};
\node[bag,below =of s4] (s6) {$s_6:(A_2, \param_{\free} \geq 1,\param_{\dataS} + 1)$};
\node[bag,right =of s6] (s7) {$s_7:(\ltsrgood,\param_{\free} \leq 1,\param_{\dataS} + \param_{\free})$};
\node[bag,below =of s6,yshift=-12mm] (s8) {$s_{8}:(P_2,\param_{\free} \geq 1,\param_{\dataS} + 1)$};
\node[bag,below =of s7] (s9) {$s_{9}:(Stop,\param_{\free} \leq 1,\param_{\dataS} + \param_{\free})\tick$};
\node[bag,below =of s8] (s10) {$s_{10}:(\ltsrbad,\param_{\paid} \geq 1 \land \param_{\free} \geq 1,\param_{\dataS} + 2)$};
\node[bag,right =of s10,xshift=-4mm] (s11) {$s_{11}:(\ltsrgood,\param_{\paid} \leq 1 \land \param_{\free} \geq 1,\param_{\dataS} + 1 + \param_{\paid})$};
\node[bag,below =of s10] (s12) {$s_{12}:(Stop,\param_{\paid} \geq 1 \land \param_{\free} \geq 1,\param_{\dataS} + 2)$};
\node[bag,below =of s11] (s13) {$s_{13}:(Stop,\param_{\paid} \leq 1 \land \param_{\free} \geq 1,\param_{\dataS} + 1 + \param_{\paid})$};

\draw (s0) -- node[aux,auto] {$\sequence{ \rSInv , \rSeqTwo }$} (s1); %
\draw (s1)  -| node[aux,auto,near start] {$\sequence{\rCondOne}$} (s3); %
\draw (s1) -- node[aux,auto] {$\sequence{\rCondOne}$} (s2); %
\draw (s2) -- node[aux,left] {$\sequence{ \rAInv , \rSeqTwo } $} (s4); %
\draw (s3)  -- node[aux,auto] {$\sequence{\rReply}$} (s5); %
\draw (s4) -- node[aux,auto]{$\sequence{(\rPickTwo,1)} $} (s6); %
\draw (s4)  -| node[aux,auto,near start] {$\sequence{(\rPickOne,1)}$} (s7); %
\draw (s6) -- node[aux,auto] {$\sequence{ \rAInv , \rSeqTwo }$}(s8); %
\draw (s7) -- node[aux,auto] {$\sequence{\rReply} $}(s9); %
\draw (s8) -- node[aux,auto] {$\sequence{(\rPickTwo,1)}$}(s10); %
\draw (s8)  -| node[aux,auto,near start] {$\sequence{(\rPickOne,1)} $} (s11); %
\draw (s10) -- node[aux,auto] {$ \sequence{\rReply}$}(s12); %
\draw (s11) -- node[aux,auto] {$\sequence{\rReply}$}(s13); %
\node[anchor=west,service,below = of s13,yshift=5mm,xshift=-16mm](s15){$\cross\ \ \ \ \ \ \ \ \ \ \ \ \ \ \ \ \ \ \ \ \ \ \ \ \ \ \ \ \ \ \ \ \ \ \ \ \ \ \ \ \ \ \ \ \  \tick$};
\end{tikzpicture}

}

${S} = (\sInvoke{\dataS}{\sequential}\ltsrgood \dres b \rres A_1)$\\
${A}_{1} = (\aInvoke{\free}{\sequential} P_1)$\\
${P}_{1} = (pick(FS\conditionArrow \ltsrgood, alrm(1) \conditionArrow A_2))$\\
${A}_{2} = (\aInvoke{\paid}{\sequential} P_2)$\\
${P}_{2} = (pick(PS\conditionArrow \ltsrgood,alrm(1)\conditionArrow \ltsrbad))$\\
${r}_{{good}} = (reply(User))$\\
${r}_{{bad}} = \protect([reply(User)]_{bad})$

\caption{LTS of the \newsInfoShort{}}
\label{fig:running}
\end{figure}

Consider again the running example~\newsInfoShort{} introduced in \cref{sec:timeBpelExample}.
Assume the parametric response times of $\free$, $\paid$ and $\dataS$ are $\param_\free$, $\param_\paid$ and $\param_\dataS$, respectively.
Recall that $\globalDelay = 3$.

\cref{fig:running} shows the LTS of \newsInfoShort{}.
The \dLTC{} resulting from the application of \algoDLTC{} is:

  \[\big( (\param_{\dataS}\leq 3) \land (\param_{\free}\leq 1) \implies (\param_{\dataS}+\param_{\free} \leq 3) \big) \land \\
  \big( (\param_{\free}\geq 1 \land \param_{\paid}\leq 1) \implies (\param_{\dataS}+\param_{\paid} \leq 2) \big) \land
  \lnot (\param_{\free}\geq 1 \land \param_{\paid}\geq 1)\]	

After simplification\footnote{For readability, we give the constraint as output in disjunctive normal form (DNF), instead of the usual conjunctive normal form (CNF).} using Z3~\cite{conf/tacas/MouraB08}, a state-of-the-art Satisfiability Modulo Theories (SMT) solver developed by Microsoft Research,\LongVersion{\footnote{%
	We first translate our expressions into Z3 expressions, and apply Z3 built-in strategies (\eg{} simplify) to get the equivalent DNF formulas.
}}
we get the following \dLTC{}:

\[ (\param_{\free} < 1 \land \param_{\dataS} + \param_{\free} \leq 3) \lor\\
  (\param_{\paid} < 1 \land \param_{\free} > 1 \land \param_{\dataS} + \param_{\paid} \leq 2) \lor\\
  (\param_{\paid} < 1 \land \param_{\dataS} + \param_{\free} \leq 3 \land \param_{\dataS} + \param_{\paid} \leq 2)
\]

This result provides us useful information on how the component services collectively satisfy the global time constraint.
That is useful when selecting component services.
For the case of \newsInfoShort{}, one way to fulfill the global time requirement of \newsInfoShort{} is to select component service \free{} with response time that is less than 1~second, and component services \dataS{} and \free{} where the summation of their response times should be less than or equal to 3~seconds.
For example, a suitable valuation is $\pval$ such that $\pval(\param_{\free}) = 0.5$, $\pval(\param_{\dataS}) = 1.5$ and $\pval(\param_{\free}) = 0.8$.

\subsection{Termination and soundness of \algoDLTC{}}\label{ssec:soundness}

\LongVersion{
	We show the termination and soundness of synthesis of \algoDLTC{}.
}

\subsubsection{Termination}

\begin{lemma}\label{lemma:acyclic}
	Let $\CS$ be a service model. Then $\LTS{\CS{}}$ is acyclic and finite.
\end{lemma}
\begin{proof}
	From \cref{assumption:bound} and from the fact that there are no recursive activities in BPEL.
\end{proof}

\begin{proposition}\label{proposition:termination}
	Let $\CS$ be a service model.
	Then $\algoDLTC(\CS)$ terminates.
\end{proposition}
\begin{proof}
	From \cref{lemma:acyclic}, $\LTS{\CS{}}$ is acyclic.
	\cref{algo:ltc} is obviously non-recursive.
	Now, \cref{algo:synConstraint} is recursive (\cref{algo:syn:8}).
	However, due to the acyclic nature of $\LTS{\CS{}}$ and the fact that \cref{algo:synConstraint} is called recursively on the \emph{successors} of the current state, then no state is explored more than once.
	This ensures termination.
\end{proof}

\begin{remark}[Complexity of \cref{algo:synConstraint}]\label{remark:complexity:algo:synConstraint}
	First, note that all states of $\LTS{\CS{}}$ are explored by \cref{algo:synConstraint}: indeed, the algorithm is recursively called on non-terminal states, and stops only on terminal states---that have no successors anyway.
	So, the algorithm time is constant in the number of states of $\LTS{\CS{}}$.
	In addition, the number of conjuncts in the result of~\cref{algo:synConstraint} is at most the number of states of $\LTS{\CS{}}$, and less if not all states are terminal states.
\end{remark}

\subsubsection{Soundness}\label{ssec:soundnesssec}
Let us prove that for any parameter valuation satisfying the output of $\algoDLTC$, any complete run ends in a good state, and all reachable good states are reachable within the global delay~$\globalDelay$.

In the following,
given a run~$\varrun_\pval$ of $\LTS{\CS[\pval]}$,
from \cref{proposition:unique-equivalent-run}
we can safely refer to the run of $\LTS{\CS}$ equivalent to~$\varrun_\pval$.

The following lemmas will be used to prove the subsequent \cref{theorem:soundness}.

\ea{Tian Huat: big issue again: what is exactly $\mystate.\Constraint$? Is that a constraint on $\Param$ or a constraint on $\Clock \cup \Param$? If the latter, then everything needs to be rewritten, as the claims are not true as such ! And according to \cref{definition:state}, it is on $\Clock \cup \Param$…}\tth{In fact, we can show that it is only $\Param$.  To see why, at the initial state $\mystate{}_0$, we start with $\mystate{}_0.\Constraint=true$. If process $P$ is an atomic activity, we will always end in $Stop$ (due to rules $\rSInv$, $\rRec$, $\rReply$, or $\rAInv$ in~\cref{appendix:semantics}). In such we can always prune $\Clock$. For structural activities, only rules $\rPickOne$, $\rPickTwo$, $\rFlowOne$, and $\rFlowTwo$ make changes to $C$. For $\rPickOne$ and $\rPickTwo$, their idle functions will  introduce constraints like $x\leq t$ for some $t \in \Param$ or $x \leq a$ for some $a \in  \grandrplus$. Without loss of generality, say we have a constraint $x\leq t_j$ where $t_j \in \Param$. Our final simplified constraint will be  $t_i\leq t_j$ (without any $x \in \Clock$).\ea{macros !!}
For rules $\rFlowOne$ and $\rFlowTwo$, their idle functions will be applied recursively, and ultimately it will be either applying $I1$, $I2$ or $I5$ (see \cref{definition:idling}). Similar to the cases of $\rPickOne$ and $\rPickTwo$, our final simplified constraints will not involve any $x \in \Clock$.}\tth{could this serve as a lemma?}\ea{two options: either we add the lemma you suggest (which is probably true indeed) or we safely add a projection onto the parameters; I don't have any preference (I would go for the second option if you don't have any preference either, but I'm fine with both options)}\tth{second option is fine for me :)}\ea{done; let me know if you see anything wrong}\tth{sorry, but I can't find where it is?}\ea{everywhere :D I added a lot of projections ($\projectP{C}$ instead of $C$) throughout the paper}

\begin{lemma}\label{lemma:nobad}
	Let $\CS$ be a service model.
	Let $\pval \models \algoDLTC(\CS)$.
	Then no bad state is reachable in $\LTS{\CS[\pval]}$.
\end{lemma}
\begin{proof}
	Let $\Kresult = \algoDLTC(\CS)$.
	$\Kresult$ is a conjunction of ``good'' parameter constraints (accumulated from \cref{algo:syn:2} in \cref{algo:synConstraint}) and ``bad'' parameter constraints (accumulated from \cref{algo:syn:4} in \cref{algo:synConstraint}).
	Hence, $\Kresult$ contains at least the negated constraints of all bad states.
	Hence, from \cref{theorem:reachability}, the bad states are unreachable for any $\pval \models \Kresult$.
\end{proof}
\begin{lemma}\label{lemma:atleastonegood}
	Let $\CS$ be a service model.
	Let $\pval \models \algoDLTC(\CS)$.
	Then any complete run of $\LTS{\CS[\pval]}$ ends in a good state.
\end{lemma}
\begin{proof}
    First, note that the initial state $\sinit$ is reachable in $\LTS{\CS[\pval]}$ (since $\sinit.\Constraint=\mathit{true}$).
	If the initial state is the only state, then from \cref{lemma:nobad}, it is also not a bad state; hence it is a good state.
	Now, if it is not the only state, %
	from the fact that all runs of $\LTS{\CS}$ end either in a good state or in a bad state,
	from the absence of bad states (\cref{lemma:nobad}),
	and from \cref{theorem:reachability},
	then any run of $\LTS{\CS[\pval]}$ ends in a good state.
\end{proof}
\begin{lemma}\label{lemma:goodontime}
	Let $\CS$ be a service model.
	Let $\pval \models \algoDLTC(\CS)$.
	Then for all good state $(\Valuation, \Process_g, \Constraint, d)$ of $\LTS{\CS[\pval]}$,
	$d \leq \globalDelay$.
\end{lemma}
\begin{proof}
	Let $\mystate_g=(\Valuation, \Process_g, \Constraint,\Delay)$ be a reachable state in $\LTS{\CS{}}$ such that $\mystate_g$ is a good state.
	From \cref{def:semantics}, $\Constraint$ is satisfiable\ (and hence $\projectP{\Constraint}$ too).
Since $\mystate_g$ is a good state, Algorithm $\SynthesizeConstraint$ added a constraint $\projectP{C} \implies \Delay \leq \globalDelay$ to the result.
	Hence, $\algoDLTC(\CS) \subseteq (\projectP{\Constraint} \implies \Delay \leq \globalDelay$).
	Now, for any $\pval \models \algoDLTC(\CS)$, we have that $\pval \models (\projectP{\Constraint} \implies \Delay \leq \globalDelay)$, and hence all reachable states in $\LTS{\CS[\pval]}$ are such that $d \leq \globalDelay$.
\end{proof}

We can now formally state the soundness of $\algoDLTC$.

\begin{theorem}\label{theorem:soundness}
	Let $\CS$ be a service model.
	Let $\pval \models \algoDLTC(\CS)$.
	Then:
	\begin{enumerate}
		\item Any complete run of $\LTS{\CS[\pval]}$ ends in a good state.%
		\item For all good state $(\Valuation, \Process_g, \Constraint, d)$ of $\LTS{\CS[\pval]}$,
	$d \leq \globalDelay$.
	\end{enumerate}
\end{theorem}
\begin{proof}
	From \cref{lemma:atleastonegood,lemma:goodontime}.
\end{proof}

Given a composite service $\CS{}$, and assume $S_g=\{\mystate_1,\ldots,\mystate_n\}$ be the set of all good states in $\LTS{\CS}$.
In the following proposition, we show that any $\pval \models \algoDLTC(\CS)$ necessarily satisfies (at least) one of the good states' constraints, \ie{} $\pval \models \projectP{\mystate_i.\Constraint}$ for some $\mystate_i \in S_g$.

Indeed, recall $\algoDLTC(\CS)$ is a conjunction of good and bad constraints.
In %
the following proposition,
we show that the good constraints of the form $(\Constraint_{1} \implies r_1 \land \ldots \land \Constraint_{n} \implies r_n)$ will not hold trivially by just having $\Constraint_{i} =false$, for all $i \in \{1,\ldots{},n\}$.

\begin{proposition}\label{prop:p-satisfy-sg}
Let $\CS$ be a service model, and $\States{}_{good}$ be the set of all good states in $\LTS{\CS}$.
Let $\pval \models \algoDLTC(\CS)$.

	Then $\exists \mystate \in \States{}_{good} : \pval \models \projectP{\mystate.\Constraint}$.
\end{proposition}
\begin{proof}
	From \cref{algo:synConstraint}, $\algoDLTC(\CS)$ is a conjunction of ``good'' constraints (accumulated at \cref{algo:syn:2}\LongVersion{ in \cref{algo:synConstraint}}) and ``bad'' constraints (accumulated at \cref{algo:syn:4}\LongVersion{ in \cref{algo:synConstraint}}).
	That is, assume $\algoDLTC(\CS) = (\Constraint_g \land \Constraint_b)$, where
	$\Constraint_g=\bigwedge_{\mystate_i \in {\States{}}_{good}} (\projectP{\mystate_i.\Constraint} \implies (\mystate_i.\Delay \leq \globalDelay))$, and $\globalDelay$ be the global time constraint,
	and $\Constraint_b=\bigwedge_{\mystate_j \in {\States{}}_{bad}} \neg (\projectP{s.\Constraint_j})$.
	Hence, since $\pval \models \algoDLTC(\CS)$ then $\pval \models \Constraint_g$, hence $\exists \mystate \in \States{}_{good} : \pval \models \projectP{\mystate.\Constraint}$.
\end{proof}

\subsection{Incompleteness of \algoDLTC{}}\label{ssec:completeness}
A limitation of \algoDLTC{} is that it is incomplete, \ie{} it does not include all parameter valuations that could give a solution to the problem of the local time requirement.
Given an expression $\mconditional{A}{a=1}{B}$, since $a$ may be unknown at design time, we explore both branches (activities $A$ and $B$) for synthesizing the \dLTC{}.
Nevertheless, only exactly one of these activities will be executed at runtime.
Including constraints from activities $A$ and $B$ will make the constraints stricter than necessary; therefore some of the feasible parameter valuations are excluded---this makes the synthesis by \algoDLTC{} incomplete.
This can be seen as a trade-off to make the synthesized local time requirement more general, \ie{} to hold in any composite service instance.
In \cref{sec:rrefine}, we will introduce a method that leverages on runtime information to mitigate this problem.

\section{Runtime refinement of local time requirement}\label{sec:rrefine}

In order to improve the local time requirement computed statically using the algorithms presented in \cref{sec:syncConstraint}, we introduce in this section a \emph{refined} local time requirement, together with its usage for runtime adaptation of a service composition.

\subsection{Motivation}\label{sec:rrefinemotivation}

Let us consider a composite service $\CS$.
Assume that we have selected a set of component services such that their stipulated response times fulfill the \dLTC{} of~$\CS$.
Since the composite service is executed under a highly evolving dynamic environment,
the design time assumptions may evolve at runtime.
For instance, the response times of component services could be affected by network congestion.
This might result in the non-conformance of stipulated response times for some component services.
However, the non-conformance of stipulated response times of component services does not necessary imply that the composite service will not satisfy its global time requirement.
This is because the \dLTC{} is synthesized at the design time to hold in \emph{any} execution trace of~$\CS$; whereas at runtime, the runtime information can be used to synthesize a more relaxed constraint for $\CS$.

More specifically, given a composite service $\CS$, we have two pieces of runtime information that may help to synthesize a more relaxed constraint: the execution path that has been taken by $\CS$, and the elapsed time of~$\CS$.
First, the execution path taken by $\CS$ can be used for LTS simplification.
This is because in the midst of execution, some of the execution traces can be disregarded and therefore a weaker LTC, that includes more parameter valuations, may be synthesized.
Second, the time elapsed of $\CS$ can be used to instantiate some of the response time parameters with real-valued constants; this makes the synthesized LTC contain less uncertainty and be more precise.

\begin{example}
	For example, consider the \newsInfoShort{} composite service, the LTS of which is depicted in \cref{fig:running}.
	Assume a valuation $\pval$ satisfying the $\dLTC$.
	At runtime, after invocation of the component service $\dataS{}$, \newsInfoShort{} will be at state~$\mystate_2$.
	Assume that $\dataS{}$ does not conform to its stipulated response time.
	Therefore, it is desirable to check whether invoking $\free{}$ can still satisfy the global time requirement of $\CS{}$.
	One can make use of \dLTC{} for this purpose.
	Nevertheless, a more precise LTC may be synthesized at state $\mystate_2$.

	The first observation is that, from state $\mystate_2$, we can safely ignore the constraints from the good state $\mystate_5$, since it is not reachable from~$\mystate_2$.
	The second observation is that the delay from state $\mystate_0$ to state $\mystate_2$ (say $r$ time units, with $r \in \grandrplus$) is known.
	For this reason, we can substitute the delay component of state $\mystate_2$, which is the response time $\pval(\param_{\dataS{}})$, with the actual time delay~$r$.
	This motivates the use of runtime information of the composite service to refine the LTC.
	We refer to the runtime refined LTC as the runtime LTC (denoted by \rLTC{}).
	In addition to this refinement, we can also simplify the LTS by pruning the states corresponding to past states (\eg{} $\mystate_0$, $\mystate_1$ in \cref{fig:running}), as well as the successors of these past states that were not met in practice (\eg{} $\mystate_3$ and $\mystate_5$ in \cref{fig:running}), because another branch was taken at runtime.
	We show the LTS of \newsInfoShort{} before and after simplification in \cref{fig:workflow:beforesim,fig:workflow:aftersim} respectively.

	\begin{figure}[htb]
		\centering
		\begin{subfigure}[b]{0.4\textwidth}
			\centering
			\begin{tikzpicture}[
 block/.style = {circle, draw,
    text width=1em,align=center,inner sep=0pt},
    line/.style = {draw,thick, -latex},
  service/.style={align=left, text width=0.5cm},
node distance=.8cm and 0.4cm
]

\node[block](s0){$s_0$};
\node[block, below of = s0](s1){$s_1$};
\node[block, below of =s1](s2){$s_2$};
\node[block, below right of =s1, xshift=2mm, yshift=2mm] (s3){$s_3$};
\node[block, below of =s3,yshift=2mm] (s5){$s_5$};
\node[block, below of =s2] (s4){$s_4$};
\node[block, below of =s4] (s6){$s_6$};
\node[block, right of =s6] (s7){$s_7$};
\node[service, below of=s6, xshift=0](s8){$\cdots$};
\node[service, below of=s7, xshift=0](s9){$\cdots$};
\path [line] (s0)--(s1);
\path [line] (s1)--(s2);
\path [line]  (s1)-|(s3);
\path [line] (s2)--(s4);
\path [line] (s4)--(s6);
\path [line] (s4)-|(s7);
\path [line] (s6)--(s8);
\path [line] (s7)--(s9);
\path [line] (s3)--(s5);
\end{tikzpicture}
			\caption{Before simplification}
			\label{fig:workflow:beforesim}
		\end{subfigure}
		\begin{subfigure}[b]{0.4\textwidth}
			\centering
			\begin{tikzpicture}[
 block/.style = {circle, draw,
    text width=1em,align=center,inner sep=0pt},
line/.style = {draw,thick, -latex},
  service/.style={align=left, text width=0.5cm},
node distance=.8cm and 0.4cm
]

\node[block](s2){$s_2$} ;
\node[block, below of =s2] (s4){$s_4$};
\node[block, below of =s4] (s6){$s_6$};
\node[block, right of =s6] (s7){$s_7$};
\node[service, below of=s6](s8){$\cdots$};
\node[service, below of=s7](s9){$\cdots$};
\path [line] (s2)--(s4);
\path [line] (s4)--(s6);
\path [line] (s4)-|(s7);
\path [line] (s6)--(s8);
\path [line] (s7)--(s9);
\end{tikzpicture}
			\caption{After simplification}
			\label{fig:workflow:aftersim}
		\end{subfigure}
		\caption{LTS Simplification of \newsInfoShort{}}\label{fig:beandsy}
	\end{figure}
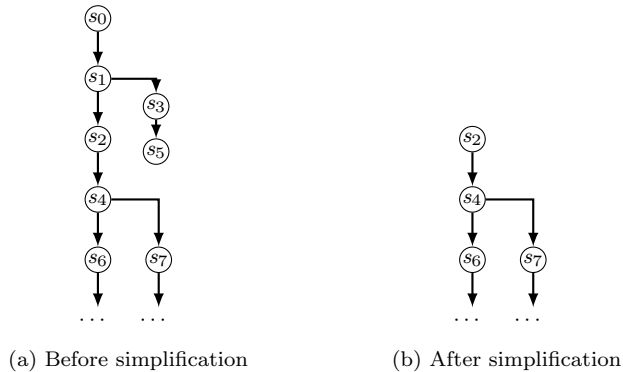

	By incorporating the runtime information, the resulting \rLTC{} at state $\mystate_2$ is:

	\medskip
	{\small\noindent
	\fbox{\begin{minipage}{8.3cm}
	$\big ((\param_{\free}\leq 1) \implies (r +\param_{\free} \leq 3) \big) \land \\
	\big ( (\param_{\free}\geq 1 \land \param_{\paid}\leq 1) \implies (r +\param_{\paid} \leq 2) \big) \land \\
	\lnot (\param_{\free}\geq 1 \land \param_{\paid}\geq 1)$
	\end{minipage}}
	}
	\medskip

\end{example}

\LongVersion{%
	Although the synthesized \rLTC{} is still incomplete, nevertheless by incorporating runtime information, it allows synthesizing a constraint weaker than \dLTC{}.
	By allowing more parameter valuations, \rLTC{} mitigates the problem of the incompleteness of the \dLTC{}.
}

\subsection{Runtime adaptation of a BPEL process}\label{ssec:rrruntimemonitoring}

We now introduce a service adaptation framework to improve the conformance of global time requirement for a composite service.
\LongVersion{%
	The framework makes use of \rLTC{}.
}%
The architecture of the framework is shown in \cref{fig:architecture}.
There are two modules in the framework--- the Runtime Engine Module (\modRE{}) and the Service Monitoring Module (\modSM{}).
\modRE{} provides an environment for the execution of a BPEL service; here, we use ApacheODE~\cite{ApacheODE}, an open source BPEL engine. We instrument the runtime component of Apache ODE
to communicate with the service monitoring module.

\begin{figure}[tb]
	\centering
	\includegraphics[width=1.8in]{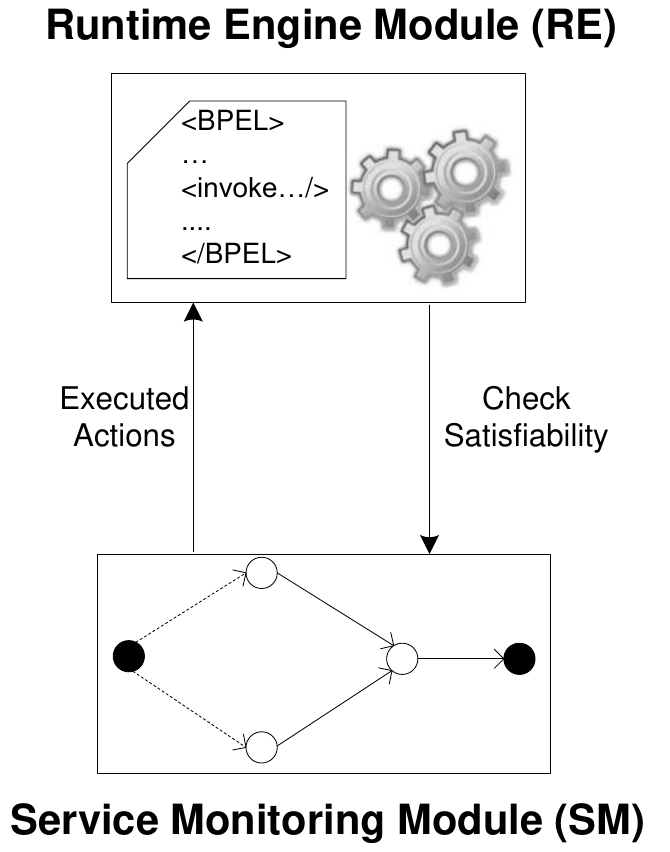}
	\caption{Service adaptation framework}
	\label{fig:architecture}
\end{figure}
\modSM{} is used to monitor the execution of a BPEL service.
During the deployment of a service $\CS$, \modSM{} generates the LTS of~$\CS$  and stores it in the cache of \modSM{} so that it is available when $\CS$ is executing.

During the execution of the composite service $\CS$, the sequences of rules
from \modRE{} are used to update the active state $\mystate_a \in \States{}$ of LTS stored in \modSM{}.
The %
sequence of rules
is\ls{are}\ea{I think this should be ``is'': ``sequence'' is singular (but English plural rules might differ from French?) But according to \url{https://gmat.economist.com/gmat-advice/gmat-verbal-section/sentence-correction/using-singular-vs-plural-verb-after-relative-pronoun} it looks the same in both languages.} also stored as part of the current execution run.
\modSM{} also keeps track of the total execution time for this execution run, as well as the response time for each component service invocation.

Prior to the invocation of a component service $\Service$, \modRE{} will consult \modSM{} to check the satisfiability of \rLTC{}. If  the \rLTC{} of $\mystate_a$ is satisfiable, then \modSM{} will instruct \modRE{} to continue invoking $\Service$ as usual.
Otherwise, some kind of mitigation procedure may be triggered. One of the possible mitigation procedures is to invoke a backup service of $\Service$, $\Service_{bak}$, which has a faster stipulated response time than $\Service$ (that may come with a cost).

\begin{example}
	Consider again the running example \newsInfoShort{} in \cref{sec:timeBpelExample}.
	An example of $\Service$ and $\Service_{bak}$, are services \free{} and \paid{} respectively.
\end{example}

In the following, we introduce the details on the synthesis of \rLTC{} (\cref{ssec:rralgo}) and  satisfiability checking (\cref{ssub:satChecking}).

\subsection{Algorithm for runtime refinement}\label{ssec:rralgo}

A way to calculate the \rLTC{} could be to run \algoDLTC{} (\cref{algo:synConstraint}) from a state $\mystate$ in the LTS.
However, this requires traversing the state-space repeatedly for every calculation of the \rLTC{}. To make it more efficient, we extend \algoDLTC{} by calculating the \rLTC{} for each state $\mystate$ during the synthesis of the LTC at the design time.
Therefore, at runtime, we only need to retrieve the synthesized \rLTC{} of the corresponding state for direct usage.

\algoRLTC{} (given in \cref{algo:synConstraintR}) synthesizes the \rLTC{} for each state in the LTS.
Before explaining the algorithm, let us introduce a few notations used in \cref{algo:synConstraintR}.
First, we assume that states in the LTS of~$\CS$ are augmented with an additional ``field'' to store the computed \rLTC{}.
We use $\mystate.\rLTC$ to denote the $\rLTC$ associated with state~$\mystate$.
Additionally, we use the following shorthand to perform a conjunction of pairs of parametric constraints $(cons_i.g, cons_i.b)$ such that the resulting pair is such that its left-hand (resp.\ right-hand) side is the conjunction of all left-hand (resp.\ right-hand) sides:
$\bigsqcap \big( (cons_1.g,cons_1.b), \dots, (cons_n.g,cons_n.b) \big)$
denotes
$\big( (cons_n.g \land \dots \land cons_n.g) , (cons_n.b \land \dots \land cons_n.b) \big)$.

Given a composite service $\CS{}$ together with its associated LTS, and a state in $\LTS{\CS}$, \algoRLTC{} returns a constraint pair $\Constraint_s=(g,b)$, where $g,b \in \setP$.
In this pair, $g$ (resp.~$b$) denotes the constraint associated to a good (resp.\ bad) state.
Given a constraint pair $\Constraint_s$, we use $\Constraint_s.g$ (resp.\ $\Constraint_s.b$) to refer to the first (resp.\ second) component of $\Constraint_s$.
Variables $d_f$ and $r_f$ are \emph{free variables}, which are variables to be substituted at runtime. In particular,
given a state $\mystate$, free variables $d_f$ and $r_f$ in $\mystate.\rLTC$ are to be substituted by the delay component $\mystate.D \in \setLP$ and the actual delay $r \in \grandrplus$ from the initial state to the state $\mystate$ respectively.

\begin{algorithm}[htb]

	\SetKwInOut{Input}{input}\SetKwInOut{Output}{output}
	\Input{Composite service $\CS{}$}
	\Input{LTS~$\LTS{\CS}$ of $\CS{}$}
	\Input{State $\mystate$ in LTS of $\CS{}$}

	\Output{Constraint pair for sub-LTS of $\CS{}$ starting with $\mystate$}

	\BlankLine

	$cons \assign \emptyset$\label{line:sLTC:1}\;
	
	\uIf{$\mystate$ is a good state\label{line:sLTC:2}}{
		$cons \assign \big(\projectP{\mystate.C} \conditionArrow (\mystate.D - d_f + r_f\leq \globalDelay) , true \big) $\label{line:sLTC:3}\;
		$\mystate.\rLTC   \assign cons.g \land (d_f=\mystate.D)$\label{line:sLTC:4}\;
	}

	\uElseIf{$\mystate$ is a bad state\label{line:sLTC:5} }{
		$cons \assign(true, \lnot (\projectP{\mystate.C}))$\label{line:markbad1}\label{line:sLTC:6}\;
		$\mystate.\rLTC  \assign cons.b$\label{line:sLTC:7}\;
	}

	\Else{
		\tcp{$\mystate$ is a non-terminal state\label{line:sLTC:8}}
		
		$cons \assign \bigsqcap_{\mystate' \in \Succ(\mystate)} \algoRLTC(\mystate')$\label{line:sLTC:9}\;
		$\mystate.\rLTC \assign cons.g \land cons.b \land (d_f=\mystate.D)$\label{line:sLTC:13}\;
	}

	\Return $cons$\;

\caption{$\algoRLTC(\CS{}, \LTS{\CS}, \mystate)$~\label{algo:synConstraintR}}

\end{algorithm}

Let us now explain \algoRLTC{} in details.
Given a good state $\mystate$ (\cref{line:sLTC:2}), $\mystate.\rLTC{}$ is assigned with value $cons.g$, with free variable $d_f$ substituted with  $\mystate.D$ (\cref{line:sLTC:4}); note that substitution is here achieved using conjunction of the constraint with the equality $d_f=s.D$.
As an illustration, consider the good state $\mystate_{13}$ in the \newsInfoShort{} example (the LTS of which is given in \cref{fig:running}).
At runtime, assume the active state is at state $\mystate_{13}$, and assume that it takes $r\in \grandrplus$ time units to execute from the initial state $\mystate_0$ to state $\mystate_{13}$.
Therefore, the previously unknown parametric  response time  in the delay component of state $\mystate_{13}$, \ie{} $t_{\dataS}+1+t_{\paid{}}$, can be substituted with the real value~$r$.
To achieve this, at \cref{line:sLTC:3}, we subtract away the free variable $d_f$, which is to be substituted with the response time parameter of state $\mystate_{13}$, and add back the free variable $r_f$, which is to be substituted with the real value~$r$.
We substitute the free variable $d_f$ at \cref{line:sLTC:4}.
For free variable $r_f$, it is only substituted in \cref{algo:chkSat} at runtime when the delay is known.
In the case of the \newsInfoShort{} example, the $\rLTC{}$ of state $\mystate_{13}$ after substituting free variable $r_f$ with value $r$ (\ie{} $\mystate_{13}.\rLTC{} \land{} (r_f=r)$) is $((t_{\paid{}} \leq 1 \land t_{\free{}} \geq 2) \implies (r \leq 3))$. %

When $\mystate$ is a bad state (\cref{line:sLTC:5,line:sLTC:6,line:sLTC:7}), we simply compute the negation of the associated constraint so as to keep the system reaching this bad state (just as in \cref{algo:synConstraint}).

When $\mystate$ is a non-terminal state (\cref{line:sLTC:8}), $\mystate.\rLTC{}$ is assigned with the conjunction of all good and bad constraints computed by recursively calling \algoRLTC{} on the successor states of~$\mystate$, where free variable $d_f$ is substituted with $\mystate.D$ (\cref{line:sLTC:13}).
\LongVersion{%
	The reason for taking the conjunction of both good and bad constraints is to guarantee any complete run from state~$\mystate$ ends in a good state, and to avoid the reachability of any bad state from~$\mystate$.
	Also note that the \rLTC{} of the initial state $\mystate_0$ is the same as its \dLTC{}, \ie{} $\mystate_0.\rLTC{} = \dLTC{}$; the reason is that at the initial state $\mystate_0$, there is no runtime information for refining the \dLTC{}, hence the refined LTC is equal to the static LTC.
	In fact, one can see \cref{algo:synConstraintR} as a generalization of \cref{algo:ltc}, in the sense that \cref{algo:synConstraintR} can be applied to any state (not only the initial one), and can benefit from the current partial execution.
}

\subsection{Satisfiability checking}\label{ssub:satChecking}

We now introduce a satisfiability checking algorithm.
This satisfiability checking is done before the invocation of a component service.
Suppose that, before the invocation of a component service~$\Service_i$, $\CS{}$ is at the active state~$\mystate_a$.
The satisfiability of the \rLTC{} at $\mystate_a$  will be checked before~$\Service_i$ is invoked.
If it is satisfiable, then it will invoke~$\Service_i$ as usual.
Otherwise, some mitigation procedures will be triggered.
A mitigation procedure could consist of invoking a faster backup service~$\Service_i'$ instead of~$\Service_i$.

\begin{algorithm}[ht!]

	\SetKwInOut{Input}{input}\SetKwInOut{Output}{output}

	\Input{LTS of the parametric composite service $\CS$,
	Active state $\mystate_a \in \States{}$,
	Elapsed time $r \in \grandrplus$,
	Set of parametric response times~$\Param$,
	Parameter valuation~$\pval$}
	\Output{True if the local time constraint at $\mystate_a$ is satisfiable, false otherwise}
	
	\BlankLine
	
	\Return{$\issat{}((\bigwedge_{1\leq i \leq n} \param_{i} \leq \pval(\param_i)) \implies (\mystate_a.\rLTC{} \land{} (r_f=r)))$\label{line:checksat:1}\; }

	\caption{\algoCheckSat($\LTS{\CS}, \mystate_a, r, \Param, \pval$)}
	\label{algo:chkSat}
\end{algorithm}

We give in \cref{algo:chkSat} the algorithm checking the satisfiability of \rLTC{} at state $\mystate_a\in Q$.
With the assumption that all component services will reply within their stipulated response times ($\bigwedge_{1\leq i \leq n} \param_{i} \leq \pval(\param_i)$), \algoCheckSat{} checks whether the \rLTC{} at state $\mystate_a$ can be satisfied with free variables $r_f$ substituted with the actual elapsed time $r \in \grandrplus$.
The function $\issat{}$ \LongVersion{(\cref{line:checksat:1}) }returns true iff the input constraint is satisfiable.

\subsection{Termination and soundness of \algoRLTC{}}\label{ss:refinement:soundness}

\LongVersion{%
	We now prove the termination and soundness of the synthesis of \algoRLTC{}.
}

\subsubsection{Termination}

\begin{proposition}\label{proposition:refinement:termination}
	Let $\CS$ be a service model, $\mystate$ be a state in $\LTS{\CS}$.

	Then $\algoRLTC(\CS,\LTS{\CS}, \mystate)$ terminates.
\end{proposition}
\begin{proof}
	Observe that \cref{algo:synConstraintR} is recursive (on \cref{line:sLTC:9}).
	However, due to the acyclic nature of $\LTS{\CS{}}$ (from \cref{lemma:acyclic}) and the fact that \cref{algo:synConstraintR} is called recursively on the successors of the current state, then no state is explored more than once.
	This ensures termination.
\end{proof}

\subsubsection{Soundness}

\cref{theorem:refinement:soundness-dynamic} formally states the correctness of our runtime refinement algorithm.
\LongVersion{%
	It generalizes \cref{theorem:soundness} to the case of runtime refinement.
}

\begin{theorem}\label{theorem:refinement:soundness-dynamic}
	Let $\CS$ be a service model. %
	Let $\LTS{\CS}$ be the LTS of~$\CS$.	
	Let $\mystate$ be the current state in~$\LTS{\CS}$ and $r$ be the current elapsed time.
	
	Fix $\pval \models \algoRLTC(\CS, \LTS{\CS}, \mystate)$.
	Then:

	\begin{enumerate}
		\item there exists a run in~$\LTS{\CS[\pval]}$ ending in some state $\mystate_\pval$ such that this run is equivalent to a run of $\LTS{\CS}$ ending in~$\mystate$;

		\item any complete run of the sub-LTS of $\LTS{\CS[\pval]}$ starting from $\mystate_\pval$ ends in a good state;
		
		\item for all good states $(\Valuation, \Process_g, \Constraint, d)$ in the sub-LTS of $\LTS{\CS[\pval]}$ starting from $\mystate_\pval$, then $d \leq \globalDelay$.
	\end{enumerate}
	
\end{theorem}
\begin{proof}
	\begin{enumerate}
		\item From \cref{proposition:unique-equivalent-run-opposite}.
		\item From \cref{def:subLTS}, the sub-LTS of $\LTS{\CS[\pval]}$ starting from $\mystate_\pval$ contains the successors of $\mystate_\pval$ in $\LTS{\CS[\pval]}$, and hence any complete run of the sub-LTS of $\LTS{\CS[\pval]}$ starting from $\mystate_\pval$ corresponds to the end of some complete run of $\LTS{\CS[\pval]}$.
		From \cref{lemma:atleastonegood}, any complete run of $\LTS{\CS[\pval]}$ ends in a good state, which gives the result.
		\item Any good state of the sub-LTS of $\LTS{\CS[\pval]}$ starting from $\mystate_\pval$ is also a good state of $\LTS{\CS[\pval]}$.
		From \cref{lemma:goodontime}, for all good state of $\LTS{\CS[\pval]}$, $d \leq \globalDelay$, which gives the result.
	\end{enumerate}
\end{proof}

\begin{remark}[Complexity of \cref{algo:synConstraintR}]\label{remark:complexity:algo:synConstraintR}
	First, note that all states of $\LTS{\CS{}}$ are explored by \cref{algo:synConstraintR}: indeed, the algorithm is recursively called on non-terminal states, and stops only on terminal states---that have no successors anyway.
	So, the algorithm time is constant in the number of states of the sub-LTS of $\LTS{\CS[\pval]}$ starting from $\mystate_\pval$.
	
	Let us now investigate the worst-case number of conjuncts in the result of~\cref{algo:synConstraintR}.
	The algorithm returns the good conjuncts ($cons.g$), the bad conjuncts ($cons.b$) and a last term (``$d_f=\mystate.D$'') (\cref{line:sLTC:13} in \cref{algo:synConstraintR}).
	Any good terminal state or bad terminal state adds exactly one conjunct to either $cons.g$ or~$cons.b$.
	Therefore, the number of conjuncts is exactly the number of terminal states, plus one due to the last term.
\end{remark}

\subsection{Discussion}\label{sec:discussion}

\paragraph{Termination}
From \cref{proposition:refinement:termination}, our method terminates due to the fact that BPEL composite services do not support recursion, and \cref{assumption:bound} on the loop activities ensuring that the upper bound on the number of iterations and the time of execution are known.
We briefly discuss how to enforce this assumption in the presence of loops in the composite service.
The upper bound on the number of iterations could be either inferred by using loop bound analysis tool (\eg{} \cite{Ermedahl_1317:2007}), or could be provided by the user otherwise.
In the worst case, an alternative option is to set up a bound arbitrary but ``large enough''.
Concerning the maximum time of loop executions, it could be enforced by using proper timeout mechanism in BPEL.

\paragraph{Time for internal operations}\label{paragraph:internal}
For simplicity, we do not account for the time taken for the internal operations of the system. In reality, the time taken by the internal operations may become significant, especially when the process is large.
We can provide a more accurate synthesis of the constraints by including an additional constraint $t_{overhead} \leq b$, where $t_{overhead} \in \grandrplus$ is a time overhead for an internal operation, and $b \mem \grandrplus$ is a machine dependent upper bound for $t_{overhead}$.
The method to obtain an estimation of $b$ is beyond the scope of this work; interested readers may refer to, \eg{} \cite{DBLP:conf/www/MoserRD08}.

\LongVersion{
	\paragraph{Completeness of \algoRLTC{}}
	The \rLTC{} computed by our algorithm \algoRLTC{} is still incomplete in general,
	with the same reason for the incompleteness of the \dLTC{} computed by \algoDLTC{} as discussed in \cref{ssec:completeness}. Nevertheless, it helps to mitigate the problem of incompleteness of the \dLTC{} with LTS simplification, as illustrated in \cref{sec:rrefinemotivation}.

	\paragraph{Bad activities}
	The bad activities are the activities triggered when timeout occurs.
	For the running example \newsInfoShort, it is a reply activity that reports the user on the timeout of a composite
	service.
	As an additional example, it could also be an invocation activity to
	log the timeout event upon the timeout of a composite service.
	With the rule of thumb that a bad activity is always triggered upon the timeout of a composite service, identifying a bad activity would become an easy task;
	devising techniques for (semi-)automating this task is left as future work.
	On the other hand, specifying bad activities is not mandatory. If the user cannot identify a bad activity in the composite service, (s)he has the option not to specify any. Doing so, all activities in the composite service are treated as good activities.
	This implies that the synthesized constraints only provide the following guarantee: any possible complete run of the composite service is able to satisfy the global time requirement upon completion. It does not consider the situation where the execution of a complete run could directly lead to the violation of the global time requirement, \eg{} the complete run that contains $\mystate_0$, $\mystate_2$, and $\mystate_4$ in \cref{fig:pickltsbad}.
}

\newcommand{\paramC}{p_c}
\newcommand{\paramE}{t_e}

\section{Evaluation}\label{sec:evaluation}

As a proof of concept, we applied our method to several examples.
After briefly presenting our implementation, %
we describe the examples we use (\cref{ssec:casestudies}).
We then evaluate our methods for the synthesis of local time requirement at the design time (\cref{ssec:synLTR}) and for the runtime refinement (\cref{ssec:monitorLTR}).

\paragraph{Implementation}
We implemented our algorithms for synthesizing the \dLTC{} and \rLTC{}\LongVersion{ (\viz{} \algoDLTC{} and \algoRLTC{})} in \ToolBPEL{}, a tool developed in \Csharp{}.
We use an \emph{ad-hoc} input syntax very close to that of \cref{definition:process}.
Our prototype implementation uses basic state space reduction techniques, notably zone inclusion (see, \eg{} \cite{HSW16,NPP18} for recent such techniques in the (parametric) timed setting), to prune whole branches of the state space. %
The front-end GUI relies on the PAT model checker~\cite{SLDP09}.

The simplification of the final results of \dLTC{} and \rLTC{} is achieved using Microsoft Z3~\cite{conf/tacas/MouraB08}.
For the runtime adaptation, we use Apache ODE 1.3.6 as runtime engine module (\modRE{}).
The service monitoring module (\modSM{}) is developed in \Csharp{}, which uses Microsoft Z3 for the satisfiability checking.
The tool and examples can be downloaded at~\cite{toolsforltr}.

\newcommand{\rrmethod}{\textit{rr}}
\subsection{Examples}\label{ssec:casestudies}

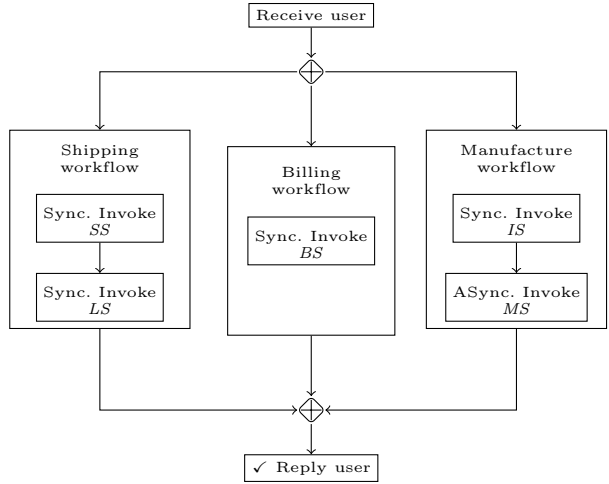
\begin{figure}[t]
	\centering
	\tikzset{
  my box/.style = {draw, minimum width = 3em, minimum height=0.7em},
  service/.style={align=center, text width=2cm},
}
{\scalefont{0.7}
\begin{tikzpicture}[node distance=4mm]
\node[my box,align=center] (b) {\\Sync.\ Invoke\\ \compSS{}};
\node[my box,align=center, below =of b](d) {\\Sync.\ Invoke\\ \compLS{}};
\node [service,above of =b, yshift=4mm](ship){Shipping\\ workflow};

\draw[->] (b)--(d);

\node[draw, fit=(ship) (b) (d), minimum height=2.5cm] (left part) {};

\node[my box,align=center, right= 3.8cm of b] (e) {\\Sync.\ Invoke\\ \compIS{}};
\node[my box,align=center, below = of e] (i) {\\ASync.\ Invoke\\ \compMS{}};
\node [service,above of =e, yshift=4mm](Manu){Manufacture\\ workflow};
\draw[->] (e)--(i);

\node[draw, fit= (e) (Manu) (i), minimum height=2.5cm] (right part) {};

\node[my box,align=center, right= 1.1cm of b, yshift=-3.0mm] (f) {\\Sync.\ Invoke\\ \compBS{}};
\node[draw, fit= (f), minimum height=2.5cm, minimum width=2.2cm] (middle part) {};
\node [service,below of =middle part, yshift=12mm](bill){Billing\\ workflow};

\node[draw,diamond, above=of middle part, rounded corners=1.5pt, yshift=4mm](a){};
\draw  ([yshift=-\Shift]a.north)
    -- ([yshift=+\Shift]a.south)
       ([xshift=+\Shift]a.west)
    -- ([xshift=-\Shift]a.east);

\node[draw,diamond, below=of middle part, rounded corners=1.5pt, yshift=-4mm](g){};
\draw  ([yshift=-\Shift]g.north)
    -- ([yshift=+\Shift]g.south)
       ([xshift=+\Shift]g.west)
    -- ([xshift=-\Shift]g.east);

\draw[->] (a)-|(left part);
\draw[->] (a)--(middle part);
\draw[->] (a)-|(right part);

\draw[->] (left part)|-(g);
\draw[->] (middle part)--(g);
\draw[->] (right part)|-(g);

\node [my box,align=center, above = of a](ru){Receive user};
\node [my box,align=center, below = of g](rp){$\tick$\ Reply user};

\draw[->](ru)--(a);
\draw[->](g)--(rp);
\end{tikzpicture}
} 
	\caption{Computer Purchasing Service (CPS)}
	\label{fig:CPS}
\end{figure}

\begin{figure}[t]
	\centering
	\tikzset{
  my box/.style = {draw, minimum width = 3em, minimum height=0.7em},
  service/.style={align=center, text width=2cm},
  oplus/.style={draw,circle, text width=0.5em,
  postaction={path picture={%
    \draw[black]
      (path picture bounding box.south west) -- (path picture bounding box.north east)
      (path picture bounding box.north west) -- (path picture bounding box.south east);}}},
      decision/.style = {diamond, draw,
    text width=1.5em,inner sep=0pt},
}
{\scalefont{0.7}
\begin{tikzpicture}[node distance=4mm]
\node[my box,align=center] (s1) {ASync.\ Invoke\\ \compFS{}};
\node[oplus,below = of s1](s2){};
\node[my box, align=center, below left= of s2](s3){OnMessage\\ \compFS{}};
\node[my box, align=center, below right= of s2](s4){OnAlarm\\ 2\ seconds};
\node[my box, align=center, below = of s4](s5){ASync.\ Invoke\\ \compFS{}$_{bak}$};
\node[oplus,below = of s5](s6){};
\node[my box, align=center, left= of s6, xshift=-2mm](s7){OnMessage\  \compFS{}$_{bak}$};
\node[my box, align=center, below = of s6](s8){OnAlarm\ 1\ second};
\node[my box, align=center, below = of s8](s9){$\cross$ res='false'};
\draw[->](s1)--(s2);
\draw[->](s2)-|(s3);
\draw[->](s2)-|(s4);
\draw[->](s4)--(s5);
\draw[->](s5)--(s6);
\draw[->](s6.west)--(s7);
\draw[->](s6.south)-|(s8);
\draw[->](s8)--(s9);
\node[service, above of=s1, yshift=4mm, text width=3cm](flight){Flight\ request\ workflow};
\node[draw, fit= (s1) (s4) (s3) (s5) (s7) (s8) (s9) (flight), minimum height=2.5cm] (left part) {};

\node[my box,align=center, right of = s1,xshift=4.0cm] (a1) {ASync.\ Invoke\\ \compHS{}};
\node[oplus,below = of a1](a2){};
\node[my box, align=center, below left= of a2](a3){OnMessage\\ \compHS{}};
\node[my box, align=center, below right= of a2](a4){OnAlarm\\ 2\ seconds};
\node[my box, align=center, below = of a4](a5){ASync.\ Invoke\\ \compHS{}$_{bak}$};
\node[oplus,below = of a5](a6){};
\node[my box, align=center,left= of a6, xshift=-2mm](a7){OnMessage\ \compHS{}$_{bak}$};
\node[my box, align=center, below = of a6](a8){OnAlarm\ 1\ second};
\node[my box, align=center, below = of a8](a9){$\cross$ res='false'};
\draw[->](a1)--(a2);
\draw[->](a2)-|(a3);
\draw[->](a2)-|(a4);
\draw[->](a4)--(a5);
\draw[->](a5)--(a6);
\draw[->](a6.west)--(a7);
\draw[->](a6.south)-|(a8);
\draw[->](a8)--(a9);
\node[service, above of=a1, yshift=4mm, text width=3cm](hotel){Hotel\ request\ workflow};
\node[draw, fit= (a1) (a4) (a3) (a5) (a7) (a8) (a9) (hotel), minimum height=2.5cm] (right part) {};

\coordinate (aux) at ($(left part.north)!.5!(right part.north)$);
\node[draw,diamond, above=of aux, rounded corners=1.5pt](g){};
\draw  ([yshift=-\Shift]g.north)
    -- ([yshift=+\Shift]g.south)
       ([xshift=+\Shift]g.west)
    -- ([xshift=-\Shift]g.east);
\draw[->](g)-|(left part);
\draw[->](g)-|(right part);

\coordinate (aux) at ($(left part.south)!.5!(right part.south)$);
\node[draw,diamond, below=of aux, rounded corners=1.5pt](h){};
\draw  ([yshift=-\Shift]h.north)
    -- ([yshift=+\Shift]h.south)
       ([xshift=+\Shift]h.west)
    -- ([xshift=-\Shift]h.east);
\draw[->](left part)|-(h);
\draw[->](right part)|-(h);

\node[my box, align=center, below =  of h](reply){Reply result};
\draw[->](h)->(reply);
\node[my box, align=center, above= of g](res){res='true'};
\node[my box, align=center, above= of res](rec){Receive\ user};
\draw[->](rec)->(res);
\draw[->](res)->(g);
\end{tikzpicture}
} 
	\caption{Travel Booking Service (TBS)}
	\label{fig:TBS}
\end{figure}
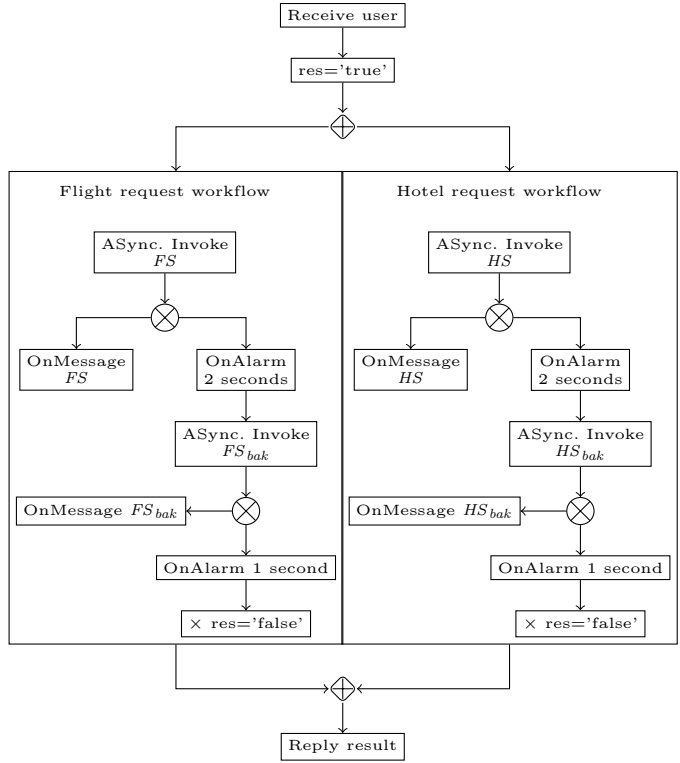

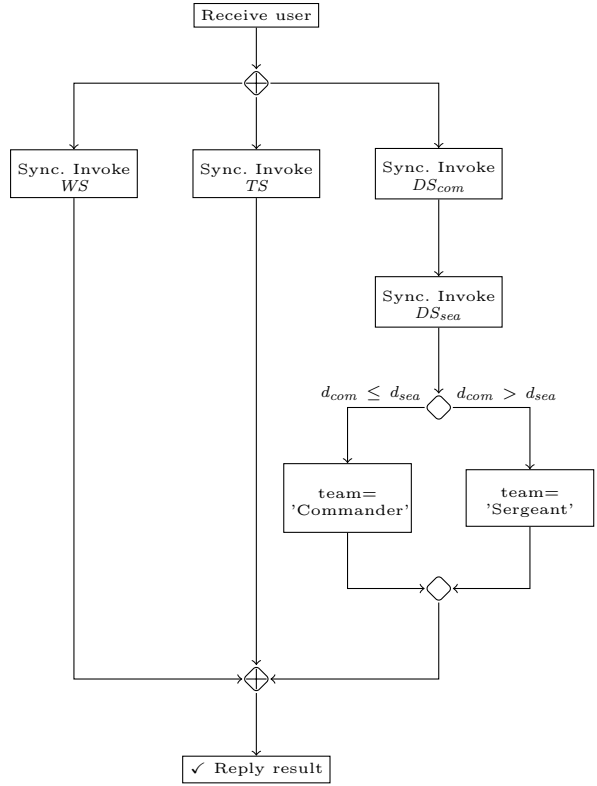
\begin{figure}[t]
	\centering
	\tikzset{
  my box/.style = {draw, align=center,minimum width = 2em, minimum height=0.2em},
  service/.style={align=center, text width=2cm},
}
{\scalefont{0.7}
\begin{tikzpicture}[node distance=7mm]

\node [my box,align=center](s0){Receive user};
\node[draw,diamond, below of =s0, rounded corners=1.5pt,yshift=-2mm](a){};
\draw  ([yshift=-\Shift]a.north)
    -- ([yshift=+\Shift]a.south)
       ([xshift=+\Shift]a.west)
    -- ([xshift=-\Shift]a.east);
\node[my box,align=center, below of =a, yshift=-5mm](s1){\\Sync.\ Invoke\\ \compTS{}};
\node[my box,align=center, below of =a, xshift=-24mm, yshift=-5mm](s2){\\Sync.\ Invoke\\ \compWS{}};
\node[my box,align=center, below of =a, xshift=24mm, yshift=-5mm](s3){\\Sync.\ Invoke\\ \compCom{}};
\node[my box,align=center, below of =s3,yshift=-10mm](s4){\\Sync.\ Invoke\\ \compSea{}};
\node[draw,diamond, below of =s4, rounded corners=1.5pt,yshift=-7mm](b){};
\node[my box,align=center, below of =b,text width=1.5cm, xshift=-12mm, yshift=-5mm,minimum height=3.4em](s5){\\team=\\ 'Commander'};
\node[my box,align=center, below of =b,text width=1.5cm, xshift=12mm, yshift=-5mm,minimum height=2.8em](s6){\\team=\\ 'Sergeant'};
\node[draw,diamond, below of =b, rounded corners=1.5pt,yshift=-17mm](c){};
\node[draw,diamond, below of =s1, rounded corners=1.5pt,yshift=-60mm](d){};
\draw  ([yshift=-\Shift]d.north)
    -- ([yshift=+\Shift]d.south)
       ([xshift=+\Shift]d.west)
    -- ([xshift=-\Shift]d.east);
\node[my box,align=center, below of =d, yshift=-5mm](s7){$\tick$ Reply result};

\draw[->] (s0)--(a);
\draw[->] (a)--(s1);
\draw[->] (a)-|(s2);
\draw[->] (a)-|(s3);
\draw[->] (s3)--(s4);
\draw[->] (s4)--(b);
\draw[->] (b)-|(s5) node [near start,anchor=south, xshift=-2mm] {$d_{com}$\ $\leq$\ $d_{sea}$};
\draw[->] (b)-|(s6) node [near start,anchor=south, xshift=2mm] {$d_{com}$\ $>$\ $d_{sea}$};
\draw[->] (s5)|-(c);
\draw[->] (s6)|-(c);
\draw[->] (c)|-(d);
\draw[->] (s1)--(d);
\draw[->] (s2)|-(d);
\draw[->] (d)--(s7);
\end{tikzpicture}
} 
	\caption{Rescue Team Service (RS)}
	\label{fig:RST}
\end{figure}

\subsubsection{Stock market indices service (\newsInfoShortI{})}
This is the running example introduced in \cref{sec:timeBpelExample}.

\subsubsection{Computer purchasing services (CPS)}
The goal of a CPS \LongVersion{(such as Dell.com) }is to allow a user to purchase a computer system online using credit cards.
Our CPS makes use of five component services, namely Shipping Service (\compSS{}), Logistic Service (\compLS{}), Inventory Service (\compIS{}), Manufacture Service (\compMS{}), and Billing Service (\compBS{}).
The global time requirement of the CPS is to respond within three seconds. The CPS workflow is shown in \cref{fig:CPS}.
The CPS starts upon receiving the purchase request from the client with credit card information, and the CPS spawns three workflows (\viz{} shipping workflow, billing workflow, and manufacture workflow) concurrently.
In the shipping workflow, the shipping service provider is invoked synchronously for the shipping service on computer systems. Upon receiving the reply, \compLS{} (which is a service provided by the internal logistic department) is invoked synchronously to record the shipping schedule.
In the billing workflow, the billing service (which is offered by a third party merchant) is invoked synchronously for billing the customer with credit card information. In the manufacture workflow, \compIS{} is invoked synchronously to check for the availability of the goods. Subsequently, \compMS{} is invoked asynchronously to update the manufacture department regarding the current inventory stock.  Upon receiving the reply message from \compLS{} and  \compBS{}, the result of the computer purchasing will be returned to the user.

\subsubsection{Travel booking service (TBS)}
The goal of a travel booking service (TBS) \LongVersion{(such as Booking.com) }is to provide a combined flight and hotel booking service by integrating two independent existing services.
TBS provides an SLA for its subscribed users, saying that it must respond within five seconds upon request.
The travel booking system has four component services, namely Flight Service (\compFS{}), Backup Flight Service (\compFS{}$_{bak}$), Hotel Service (\compHS{}) and Backup Hotel Service (\compHS{}$_{bak}$).
The TBS workflow is given in \cref{fig:TBS}.
Upon receiving the request from users, the variable $res$ is assigned to true.
After that, TBS spawns two workflows (\viz{} a flight request workflow, and a hotel request workflow)  concurrently.
In the flight request workflow, it starts by invoking \compFS{}, which is a service provided by a flight service booking agent.
If service \compFS{} does not respond within two seconds, then \compFS{} is abandoned, and another backup flight service \compFS{}$_{bak}$ is invoked.
If \compFS{}$_{bak}$ returns within one second, then the workflow is completed; otherwise the variable $res$ is assigned to false.
The hotel request workflow shares the same process as the flight request workflow, by replacing \compFS{} with \compHS{} and \compFS{}$_{bak}$ with \compHS{}$_{bak}$.
The booking result will be replied to the user if $res$ is true; otherwise, the user will be informed of the booking failure.\label{chg:revision3_reviewer1_comment7}

\subsubsection{Rescue team service (RS)}
The goal of a Rescue Team service (RS) is to identify the place, weather, and nearest rescue team, by the longitude and latitude on Earth.
RS makes use of three component services, namely Terra Service (\compTS{}), Weather Service (\compWS{}) and Distance Service (\compDS{}). The global requirement of the RS is to respond within five seconds.
The RS workflow is given in \cref{fig:RST}.
RS starts upon receiving longitude and latitude coordinates from the user.
After that, it invokes Terra Service (\compTS{}), Weather Service (\compWS{}),
and Distance Service (\compDS{}) concurrently.
Service \compTS{} (resp.~\compWS{}) will return the name of the place (resp.\ the weather information)
that corresponds to the longitude and latitude.
\compDS{} is used to calculate the distance between each rescue team and the event location.
In particular, \compCom{} and \compSea{} are used to calculate the distance between commander team and sergeant team to the event location.
If the distance to the event of the commander team ($d_{com}$) is not larger than the distance to the event of the sergeant team ($d_{sea}$), then the commander team will be chosen.
Otherwise, the sergeant team will be chosen.
Subsequently, the place, weather and rescue team information is returned to the user.

\subsection{Synthesis of local time requirement}\label{ssec:synLTR}

\subsubsection{Environment of the experiments}
We run our algorithms to synthesize the \dLTC{}  and \rLTC{} for the four examples on a computer with Intel Core I5 2410M CPU with 4\,GiB RAM.

\subsubsection{Evaluation results}\label{subsec:evaresultssub}
The details of the synthesis are shown in \cref{fig:ltcsynthesis}.
The \textbf{\#states} and \textbf{\#transitions} columns provide the information of number of states and transitions of the LTS, respectively.
We repeated all experiments 30 times; we report here the average time for each experiment. %
The \textbf{\dLTC} and \textbf{\rLTC} columns provide the average time (in seconds) spent for synthesizing \dLTC{} (for the entire LTS),  and  \rLTC{} (for each state in the LTS), respectively. TBS takes a longer time than SMIS, CPS, and RS for synthesizing \dLTC{} and \rLTC{}, as it contains a larger number of states and transitions compared to SMIS, CPS, and RS.
Nevertheless, since both  \dLTC{} and \rLTC{} are synthesized offline, the time for synthesizing the constraints (less than two seconds) for TBS is considered to be reasonable.

\begin{table}[t]
  \centering
    \begin{tabular}{|c|c|c|c|c|}
    \hline
    \textbf{Example} & \textbf{\#states} & \textbf{\#transitions} & \textbf{\dLTC{} (s)} & \textbf{\rLTC{} (s)} \\
    \hline
    SMIS & 14  & 13  & 0.0076 & 0.0078 \\
    \hline
    TBS & 683 & 3677 & 1.8501 & 1.9000 \\
    \hline
    CPS & 120 & 119 & 0.0529 & 0.0559 \\
    \hline
    RS & 85 & 134 & 0.0701 & 0.0733 \\
    \hline
    \end{tabular}

    \caption{Synthesis of \dLTC{} and \rLTC{}}
	\label{fig:ltcsynthesis}%

\end{table}

\LongVersion{%
	The main overhead on synthesizing \dLTC{} and \rLTC{} is due to the calculation of the constraint component $\Constraint$ of each new state $\mystate=(\Valuation,\Process,\Constraint,\Delay)$. Calculation of a constraint component $C'$ is required for each transition in the LTS, in order to create the successor state $\mystate'=(\Valuation',\Process',\Constraint',\Delay')$.
	If two created states $\mystate,\mystate'$ in the LTS are found to be equal, they will be merged into a single state.\label{th:merge}
	Therefore, the number of transitions will have a greater effect than the number of states on the time spent for synthesizing \dLTC{} and \rLTC{}.
}

\LongVersion{
\begin{table}[t]
  \centering
    \scalebox{0.9}{
    \begin{tabular}{|c|c|c|c|}
    \hline
    \textbf{Case Studies} & \textbf{\#states} & \textbf{\#transitions} & \textbf{\dLTC{} (s)} \bigstrut\\
    \hline
    SMIS & 17  & 16  & 0.0090 \bigstrut\\
    \hline
    TBS & 938 & 4614 & 1.9811  \bigstrut\\
    \hline
    CPS & 144 & 143 & 0.0626  \bigstrut\\
    \hline
    RS & 117 & 166 & 0.0740\bigstrut\\
    \hline
    \end{tabular}%
    }
  \caption{Synthesis for \dLTC{} based on \protect\cite{TanICSE13}}
  \label{fig:rltcsynthesis2}

\end{table}

	For comparison with~\cite{TanICSE13}, we also provided the information of synthesis of \dLTC{} based on~\cite{TanICSE13} in \cref{fig:rltcsynthesis2} (\dLTC{} is given in seconds).
	The additional ``and/or states'' (used in~\cite{TanICSE13} for synthesizing the \dLTC{}) yield an increased number of both states and transitions.
	As a consequence, this yields an increased time for synthesizing the \dLTC{}, for all case studies.
}

The synthesized \dLTC{} for SMIS %
has been given in \cref{ssec:ltcexample},
while the synthesized \dLTC{} for CPS, TBS, and RS are shown in \cref{fig:synRes}.
\LongVersion{%
	All the \dLTC{} are simplified and in DNF form.
	It is worth noting that the \dLTC{} of CPS and RS can be represented in one line representation (\ie{} only one inequality) after simplification.
}%
Note that $\param_{\compMS{}}$ does not appear in the \dLTC{} of CPS.
The reason is that \compMS{} is invoked asynchronously without expecting a response; therefore its response time is irrelevant to the global time requirement of~CPS.

The synthesized \rLTC{} are used for runtime adaptation at runtime.
We will evaluate the runtime adaptation of a composite service with \rLTC{} in the following section.

\begin{figure}[t]
	\begin{subfigure}[t]{0.45\textwidth}
		{\centering
			$(\param_{\compSS{}} + \param_{\compLS{}} + \param_{\compIS{}} + \param_{\compBS{}}) \leq  3$
		
		}
		\caption{\dLTC{} for CPS}
		\label{fig:subfig2}
	\end{subfigure}
	\begin{subfigure}[t]{0.45\textwidth}
		{\centering
			$(\param_{TS}+\param_{WS}+ 2 \cdot \param_{\dataS}) \leq 5$

		}
		\caption{\dLTC{} for RS}
	\end{subfigure}
	
	\begin{subfigure}[t]{\textwidth}
		\[
		\big((2 \cdot \param_{HSbak} < \param_{FSbak})\land (2 \cdot  \param_{FSbak} < \param_{HSbak}) \land (\param_{HSbak} < 1) \land (\param_{FSbak} < 1)\big) \\
		\lor \big((\param_{HSbak}< 1) \land (\param_{FSbak} < 1) \land (\param_{FSbak}+ \param_{HSbak} \leq 1)\big) \\
		\lor \big((\param_{HSbak}<1) \land (\param_{\compFS{}}<2)\big)
		\lor \big((\param_{\compHS{}}<2) \land (\param_{FSbak}<1)\big)
		\lor \big((\param_{\compHS{}}<2) \land (\param_{\compFS{}}<2)\big)
		\]
		\caption{\dLTC{} for TBS}
	\end{subfigure}
	\caption{Synthesized \dLTC{}}
	\label{fig:synRes}
\end{figure}

\subsection{Runtime adaptation}\label{ssec:monitorLTR}

We now conduct experiments to answer the following two questions:

\paratitle{Q1} What is the \emph{overhead} of the runtime adaptation?

\paratitle{Q2} What is the \emph{improvement} provided by the runtime adaptation?

\subsubsection{Environment of the experiments}\label{ssec:expsetup}
The evaluation was conducted using two different physical machines, connected by a 100 Mbit LAN.
One machine is running ApacheODE~\cite{ApacheODE} to host the \modRE{} module to execute the BPEL program, configured with Intel Core I5 2410M CPU with 4\,GiB RAM.
The other machine hosts the \modSM{} module, configured with Intel I7 3520M CPU with 8\,GiB RAM.

To test the composite service under controlled situation, we introduce the notion of \emph{execution configuration}. An execution configuration defines a particular execution scenario for the composite service.
Formally, an execution configuration $E$ is a tuple $(M, R)$, where $M$ decides which path to choose for an \xml{if} activity
and $R$ is a function that maps a component service $\Service_i$ to a real value $r \in \grandrplus$, which represents the response time of~$\Service_i$.
We discuss how an execution configuration $E=(M, R)$ is generated.
$M$ is generated by choosing one of the branches of an \xml{if} activity uniformly among all possible branches. %

Let $\CS$ be a composite service
	model, where a component service $\Service_i$ of $\CS$ has a stipulated response time $\pval(\param_i) \in \grandqplus$.
Then $R(\Service_i)$ will be assigned with a response time within the stipulated response time $\pval(\param_i)$ with a probability of $\paramC{}\in \grandqplus \cap [0,1]$.
$\paramC{}$ is the \emph{response time conformance threshold}.  More specifically,  $R(\Service_i)$ will be assigned with a value in $[0, \pval(\param_i)]$ uniformly with a probability of $\paramC{}$, and assigned to a value in $(\pval(\param_i), \pval(\param_i)+\paramE{}]$ uniformly with a probability of $1-\paramC{}$.
$\paramE{} \in \grandrplus$ is the \emph{exceeding threshold}; and assume after $\pval(\param_i)+\paramE{}$ seconds, the component service $\Service_i$ will be automatically timeout by $\modRE{}$ to prevent an infinite delay.

Given a composite service $\CS$, and an execution configuration $E$, a \emph{run} is denoted by $\varrun(\CS,\adaptMech,E)$, where the first argument is the composite service $\CS$ that is running, the second argument $\adaptMech{} \in \{\rrmethod{}, \emptyset \}$ is the adaptive mechanism where $\rrmethod{}$ denotes the runtime adaptation, and $\emptyset$ denotes no runtime adaptation.
\LongVersion{%
	Two runs $\varrun(\CS,\adaptMech,E)$ and $\varrun(\CS',\adaptMech',E')$ are equal if $\CS=\CS'$, $\adaptMech=\adaptMech'$ and $E=E'$.
	Note that all equal runs have the same execution paths and response times for all service invocations.
}

\LongVersion{
	We use two kinds of instrumentation for runtime engines used for adaptive and non-adaptive runs, respectively.
	For both adaptive and non-adaptive runs, we instrument both runtime engines to execute conditional statements based on~$M$.
	For adaptive runs only, we instrument the runtime engine according to~\cref{ssec:rrruntimemonitoring}.
}

\begin{table}[t]
	\centering
	\begin{tabular}{|c|c|c|}
	\hline
	\textbf{Example} & \textbf{Avg.\ \#SAT} & \textbf{Avg.\ SAT runtime (s)}\\
	\hline
	SMIS & 1.74  & 13 \\
	\hline
	TBS & 2.25 & 17 \\
	\hline
	CPS & 4.00 & 27 \\
	\hline
	RS & 4.00 & 19 \\
	\hline
	\end{tabular}%

    \caption{Satisfiability checking}\label{fig:dltcsynthesis}%

\end{table}%

\begin{table*}[t]
  \centering
\begin{tabular}{|c|c|c|c|c|c|}
    \hline
   & \textbf{$\paramC{}$} &\textbf{$N_{se}$} & \textbf{$N_{e}$}&  \textbf{Improvement (\%)} & \textbf{Avg.\ backup service} \\
    \hline
    \multirow{4}[8]{*}{SMIS} & 0.9 & 9441 & 8976 & 5.18 &  0.127 \\
\cline{2-6}        & 0.8 & 9211 & 8374 & 10.00 &  0.352   \\
\cline{2-6}        & 0.7 & 8109 & 6965 & 16.42 & 0.577 \\
\cline{2-6}        & 0.6 & 7593 & 6348 & 19.61 & 0.702   \\
    \hline
    \multirow{4}[8]{*}{TBS} & 0.9 & 10000 & 9743 & 2.64 & 0.384   \\
\cline{2-6}        & 0.8 & 10000 & 9364 & 6.79 & 0.779  \\
\cline{2-6}        & 0.7 & 10000 & 8460 & 18.20 & 0.948    \\
\cline{2-6}        & 0.6 & 10000 & 7700 & 29.87  & 1.05 \\
    \hline
    \multirow{4}[8]{*}{CPS} & 0.9 & 9523 & 8809 & 8.11 & 1.259  \\
\cline{2-6}        & 0.8 & 9241 & 7156 & 29.14 & 1.509 \\
\cline{2-6}        & 0.7 & 8504 & 6108 & 39.23 & 2.014  \\
\cline{2-6}        & 0.6 & 8430 & 5650 & 49.20 & 2.578 \\
    \hline
 \multirow{4}[8]{*}{RS} & 0.9 & 8181 & 7271 & 12.52 & 1.787 \\
\cline{2-6}        & 0.8 & 7201 & 7011 & 2.71 & 1.589 \\
\cline{2-6}        & 0.7 & 6590 & 5227 & 26.08 & 1.659  \\
\cline{2-6}        & 0.6 & 5609 & 4146 & 35.29 & 1.54 \\
    \hline
    \end{tabular}%

    \caption{Improvement of runtime conformance}  \label{fig:impExp}
\end{table*}%

\subsubsection{Evaluation results}\label{ssec:evalresults}
We conducted two experiments Exp1 and Exp2, to answer the questions Q1 and~Q2, respectively.
Each experiment goes through 10,000 rounds of simulations, and an execution configuration $E$ is generated for each round of simulation.
Given a composite service $\CS{}$, we assume that for each component service~$\Service_i$ with a stipulated response time $\pval(\param_i)$, there exists a backup service~$\Service'_i$, with a stipulated response time $\pval(\param_i)/2$ and a conformance threshold of 1.
Suppose that before the invocation of a component service~$\Service_i$, $\CS{}$ is at active state $\mystate_a$.
The satisfiability of the \rLTC{} at $\mystate_a$  will be checked (using \cref{algo:chkSat}) before~$\Service_i$ is invoked.
If it is satisfiable, then it will invoke~$\Service_i$ as usual.
Otherwise, the faster backup service~$\Service_i'$ will be invoked instead, as a mitigation procedure.

\LongVersion{%
	We now describe both experiments Exp1 and Exp2.
}

\paragraph{Experiment Exp1}\label{pExp1}
Given a composite service $\CS$, in order to measure the overhead, we use an execution configuration $E=(M, Q)$ for an adaptive run $\varrun(\CS,\rrmethod{},E)$, and non-adaptive run $\varrun(\CS,\emptyset,E)$.
We have modified the runtime adaptation mechanism for $\rrmethod{}$ so that, if the \rLTC{} of the active state is checked to be unsatisfiable, component service~$\Service_i$ will still be used (instead of~$\Service_i'$).
The purpose for this modification is to make $\varrun(\CS,\rrmethod{},E)$ and $\varrun(\CS,\emptyset,E)$ invoke the same set of component services, so that we can effectively compare the overhead of $\varrun(\CS,\rrmethod{},E)$.

\paratitle{Results} Suppose at round $k$, the times spent for $\varrun(\CS,\rrmethod{},E)$ and $\varrun(\CS,\emptyset,E)$ are $r_{rr}^k \in \grandrplus$ time units and $r_{\emptyset}^k \in \grandrplus$ time units respectively.
The overhead $O_k$ at round $k$ is the time difference between $r_{rr}^k$ and $r_{\emptyset}^k$, \ie{} $O_k=r_{rr}^k - r_{\emptyset}^k$. The average overhead at round $k$ is calculated using \cref{eq:overhead}.
\begin{equation}
Avg.\ overhead=(\sum\limits_{i=1}^k O_i)/k \label{eq:overhead}
\end{equation}

The main source of overhead for runtime adaptation comes from the satisfiability checking with \cref{algo:chkSat}.
We make use of Z3~\cite{conf/tacas/MouraB08} for this purpose.
Other sources of overhead include update of active state in \modSM{}, and communications between \modSM{} and \modRE{}.

\LongVersion{%
	\input{figures/overhead}

	The experiment results can be found in \cref{fig:overheadExp}.
}%
The average overheads of \newsInfoShort, CPS, TBS, and RS after 10,000 rounds are 15\,ms, 21\,ms, 30\,ms, and 23\,ms respectively.
The results convey to us that the additional operations involved in the runtime adaptation, including the satisfiability checking, can be done efficiently.

We further evaluate the overhead on satisfiability checking. \cref{fig:dltcsynthesis} shows the results of satisfiability checking. The average number of satisfiability checking for each round (Avg.\ \#SAT) is calculated using \cref{eq:overheadk} where $N_i$ is the total number of satisfiability checking for $i$-th round and $r$ is the total number of running rounds.
The average time (given in milliseconds) spent on satisfiability checking for each round (Avg.\ SAT\ runtime) is calculated using \cref{eq:overheadR}, where $T_i$ is the time spent on satisfiability checking for $i$-th round. \cref{fig:dltcsynthesis} shows that the satisfiability checking has contributed most of the overhead of runtime adaptation.

\begin{equation}
Avg.\ \#SAT=(\sum\limits_{i=1}^r N_i)/r \label{eq:overheadk}
\end{equation}

\begin{equation}
Avg.\ SAT\ runtime=(\sum\limits_{i=1}^r T_i)/r \label{eq:overheadR}
\end{equation}

\paragraph{Experiment Exp2}
In this second experiment, we measure the improvement for the conformance of global constraints due to \rrmethod{}. Given a composite service $\CS$, an execution configuration $E$, two runs $\varrun(\CS,\rrmethod{},E)$ and $\varrun(\CS,\emptyset,E)$ are conducted for each round of simulation. $N_{se}$ is the number of executions that satisfy global constraints for composite service with \rrmethod{}, and $N_{e}$ is the number of executions that satisfy global constraints for composite service without \rrmethod{}, the improvement is calculated by \cref{eq:improvement}.
\begin{equation}
Improvement=\frac{(N_{se}-N_{e})*100}{N_{e}}\label{eq:improvement}
\end{equation}
\paratitle{Results} The experiment results can be found in~\cref{fig:impExp}.  The $Improvement\ (\%)$ column provides the information of improvement (in percentage) that is calculated using \cref{eq:improvement}.  The \emph{Avg.\ Backup Service} column provides the average number of backup service used (calculated by summing up the number of backup services used for 10,000 rounds, and divided by 10,000).

The decrement of $\paramC{}$ represents the undesired situation where component services have a higher chance for not conforming to their stipulated response time.
This may be due to situations such as poor network conditions.
For each example, the improvement provided by the runtime adaptation increases when $\paramC{}$  decreases. This shows that runtime adaptation improves the conformance of global time requirement.
In addition, the average number of backup service used increases when $\paramC{}$ decreases.
This shows the adaptive nature of runtime adaptation with respect to different $\paramC{}$---more corrective actions are likely to perform when the chances that component services do not satisfy their stipulated response time increase.

\LongVersion{
\paragraph{Answer to research questions}
}
The results in Exp1 and Exp2 have shown that the runtime adaptation has a low overhead, and improves the runtime conformance, especially when the response time conformance threshold of the component services is low.

\LongVersion{
\subsection{Threats to validity}
There are several threats to validity. %
The first threat to validity is due to the fact that we assume a uniform distribution of response time for evaluation of runtime adaptation.
To address this issue, more experimentations with real-world services should be performed.
This said, our experiments provide a first idea that our assumptions are realistic.

The second threat to validity is stemmed from our choice to use a few example values as experimental parameters, that include global constraints and termination thresholds, in order to cope with the combinatorial explosion of options. To resolve this problem, it is clear that even more experimentations with different case studies and experimental parameters should be performed, so that we could further investigate the effects that have not been made obvious by our case studies and experimental parameters.
}

\subsection{Threats to validity}
Our experiments show a good efficiency of our technique for the examples we considered; these are arguably on the smaller side, but we claim that they are non-trivial enough to not be analyzable by hand, and therefore our technique proposes what we believe to be a valuable contribution.

\section{Related work}\label{sec:relatedWork}

\paragraph{Model-based analysis of Web services using LTSs}
Our method is related to using LTSs for model-based analysis of Web services.
In~\cite{DBLP:conf/icse/BianculliGP11}, the authors propose an approach to obtain behavioral interfaces in the form of LTSs of external services by decomposing the global interface specification.
It also has been used in model checking the safety and liveness properties of BPEL services. For example, Foster \emph{et al.}~\cite{FosterThesis,FosterICSE06} transform BPEL process into FSP~\cite{DBLP:books/daglib/0016245}, subsequently using a tool named ``WS-Engineer'' for checking safety and liveness properties.
Simmonds \emph{et al.}~\cite{DBLP:conf/sigsoft/SimmondsBC10} propose a user-guided recovery framework for Web services based on LTSs.
Our work uses LTSs in synthesizing local time requirement.

\paragraph{Constraint synthesis for scheduling problems}
Our work shares common techniques with work for constraint synthesis for scheduling problems.
The use of models such as parametric timed automata (PTAs)~\cite{AHV93} and parametric time Petri nets (PTPNs)~\cite{TLR09} for solving such problems has received recent attention.
In particular, in~\cite{CPR08,LPPRC10,FLMS12}, parametric constraints are inferred, guaranteeing the feasibility of a schedule using PTAs extended with stopwatches (see, \eg{} \cite{AM02}).
In~\cite{ALSD14}, we proposed a parametric, timed extension of CSP, to which we extended the ``inverse method'', a parameter synthesis algorithms preserving the discrete behavior of the system (see, \eg{} \cite{AS13}).
Although PTAs or PTPNs might have been used to encode (part of) the BPEL language, our work is specifically adapted and optimized for synthesizing local timing constraint in the area of service composition.

\paragraph{Finding suitable quality of service} %
Our method is related to the finding of a suitable quality of service (QoS) for the system~\cite{journals/tweb/YuZL07}.
The authors of~\cite{journals/tweb/YuZL07} propose two models for the QoS-based service composition problem: a combinatorial model and a graph model. The combinatorial model defines the problem as a multidimension multichoice 0-1 knapsack problem. The graph model defines the problem as a multiconstraint optimal path problem. A heuristic algorithm is proposed for each model: the WS-HEU algorithm for the combinatorial model and the MCSP-K algorithm for the graph model.
The authors of~\cite{conf/bpm/ArdagnaP05} model the service composition problem as a mixed integer linear problem where constraints of global and local component service can be specified.
The difference with our work is that, in their work, the local constraint is specified, whereas in ours, the local constraint is synthesized.
An approach of decomposing the global QoS to local QoS has been proposed in~\cite{conf/www/AlrifaiR09}. It uses the mixed integer programming (MIP) to find optimal decomposition of QoS constraint.
However, the approach only concerns simple sequential composition of Web services method calls, without considering complex control flows and timing requirements.

\paragraph{Response time estimation}
Our approach is also related to response time estimation.
In~\cite{DBLP:conf/valuetools/KraftPCD09}, the authors propose to use linear regression method and a maximum likelihood technique for estimating the service demands of requests based on their response times.
\cite{DBLP:journals/internet/Menasce04} has also discussed the impact of slow services on the overall response time on a transaction that use several services concurrently.
Our work is focused on decomposing the global requirement into local requirement, which is orthogonal to these works.
Our work~\cite{DBLP:conf/fm/LiTC14} complements with this work by proposing a method on building LTCs that under-approximate the \dLTC{} of a composite service.
The under-approximated LTCs consisting of independent constraints over components, which can be used to improve the design, monitoring and repair of component-based systems under time requirements.

\paragraph{Service monitoring}
Our method is related to service monitoring. Moser \emph{et al}.~\cite{DBLP:conf/www/MoserRD08} present VieDAME, a non-intrusive approach to monitoring. VieDAME allows monitoring of BPEL composite service on quality of service attributes, and existing component services are replaced based on different replacement strategies.
They make use of the aspect-oriented approach (AOP); therefore the VieDAME engine adapter could be interwoven into the BPEL runtime engine at runtime.
Baresi \emph{et al.}~\cite{DBLP:journals/tse/BaresiG11} propose an idea of self-supervising BPEL processes by supporting both service monitoring and recovery for BPEL processes.
They propose using Web Service Constraint Language (WSCoL) to specify the monitoring directives to indicate properties that need to hold during the runtime of composite service.
They also make use of the AOP approach to integrate their monitoring adapters with the BPEL runtime engine.
Our work is orthogonal to the aforementioned works, as we do not assume any particular service monitoring framework for monitoring the composite service, and those methods can be used to aid our monitoring approach, as discussed in \cref{ssec:rrruntimemonitoring}.
Our previous work~\cite{DBLP:conf/www/TanCA0LD14} proposes an automated approach based on a genetic algorithm to calculate the recovery plan that can guarantee the satisfaction of functional properties of the composite service after recovery.

\paragraph{Service selection}
In~\cite{ZBDKS03,zeng2004qos},
the authors present an approach that makes use of global planning to search dynamically for the best
services component for service composition. Their approach involves the use of mixed integer programming (MIP) techniques to find the optimal selection of component services.
Ardagna \emph{et al.}~\cite{ACMPP07} extend the MIP methods to include local constraints. Cardellini \emph{et al.}~\cite{CCGPM09} propose a methodology to integrate different adaptation mechanisms for combining concrete services to an abstract service, in order to achieve a greater flexibility in facing different operating environments.
Our work is orthogonal to aforementioned works, as it does not assume particular formulation of the MIP problems.

Although the method in aforementioned works efficiently for small case studies, it suffers from scalability problems when the size of the case studies becomes larger, since the time required grows exponentially with the size of problem.
To address this problem, Yu \emph{et al.}~\cite{journals/tweb/YuZL07} propose a heuristic algorithm that could be used to find a near-optimal solution.
The authors proposed two QoS compositional models, a combinatorial model and a graph model. The time complexity for the combinatorial model is polynomial, while the time complexity for the graph model is exponential. However the algorithm does not scale with the increasing number of web services. To address this problem, Alrifai \emph{et al.} present an approach that pruned the search space using skyline methods, and they make use of a hierarchical $k$-means clustering method~\cite{KMeansClustering} for representative selection. The work of Alrifai \emph{et al.} is the closest to ours.
Despite its reasonable performance, a limitation for the MIP approach is that it cannot deal with non-linear objective functions or  aggregated constraints.
To address this problem, Canfora \emph{et al.}~\cite{CPEV05} have formulated the problem as a genetic algorithm problem.  Genetic algorithms (GA) are algorithms based on stochastic search methods, that support non-linear objective functions. Two different GAs encodings are proposed in~\cite{CPEV05,zhang2003demand}.
In~\cite{zhang2003demand}, the authors propose to encode the chromosome using binary strings, and every gene is a chromosome representing a service candidate with value 0 (respectively, 1) that represents the unselected (respectively, selected) service. Therefore the length of the genome can be very long, given a large number of service candidates.
In~\cite{CPEV05}, the authors propose to encode the chromosome using an integer value which represents the index of the concrete services stored in an array.
This coding scheme results in shorter chromosomes, and the length of a chromosome is independent of the number of service candidates.
In~\cite{gao2007QoS}, Gao \emph{et al.}\ propose a tree coding scheme to represent the service composition.
They reported a 40\% performance improvement with respect to the single dimension coding scheme used in~\cite{CPEV05}.
This is because the tree coding scheme does not need to recalculate the entire fitness value each time compare to the single-dimension encoding.
Our work does not assume any particular encoding scheme and it can be used with any existing coding techniques.
In~\cite{ai2008penalty}, Ai \emph{et al.}\ proposed an approach extending the GA methods for handling inter-service dependencies and conflicts using a penalty-based genetic algorithm.
Our work does not assume a particular fitness function.
In~\cite{ma2008quick}, Ma \emph{et al.}\ proposed an enhanced initial population policy and an evolution policy based on  population diversity and a relation matrix coding scheme.
They considered all concrete services for each service class starting from the initial population. Different from their approach, we only consider a subset of services with high local utility value from the start, and we progressively add more services later on.
In~\cite{SN20} the problem of functionally equivalent service composition is considered.

\paragraph{Verification of services}
Concerning verification of services, Filieri \emph{et al}.~\cite{DBLP:conf/icse/FilieriGT11} focus on checking the reliability of component (service)-based systems. They make use of Discrete Time Markov Chain (DTMC) to check the reliability of models at runtime.
Our previous works~\cite{DBLP:conf/icfem/ChenT0LPL13,DBLP:conf/icse/ChenT0LD14} develop a tool to verify combined functional and non-functional requirements of Web service composition.
In contrast, the current work focuses on response time: given the global response time of the composite service, we synthesize the response time requirement for component services at design time and refine it at runtime.
Schmieders \emph{et al}.~\cite{DBLP:conf/servicewave/SchmiedersM11}  proposed the SPADE approach.
SPADE invokes the BOGOR model checker to model check the SLAs at design time and at runtime.
Our work is different from theirs in two aspects.
First, we focus on the synthesis of the local time requirement, which is a formal requirement on the response time requirement of component services.
Second, at runtime, \cite{DBLP:conf/servicewave/SchmiedersM11} performs model checking on a given state to check whether an adaptation is needed. In contrast, we have precomputed the constraints for every state at design time.
Therefore, we only require evaluation of constraints by substituting the free variables at runtime, and this allows a more efficient runtime-analysis.

\paragraph{BPEL}

In~\cite{GMJ08} a template-based system is used to reconfigure service composition, using BPEL.
In~\cite{Pautasso09}, service composition using RESTful (Representational State Transfer) is performed using the BPEL extention ``BPEL for REST''.
In~\cite{EFTG10}, a tool based on Services Creation Environments and BPEL is proposed, that also allows translation to Java.

The work~\cite{TBM13} automates the formalization and verification process of BPEL.
It extends the existing spring framework to represent BPEL activities with Java bean, which is subsequently transformed into XML-based formal model (like colored Petri nets) for verification.
The work~\cite{MMKG16} also transforms BPEL services into Probabilistic Labeled Transition Systems, which are then used to conduct probabilistic model checking to verify reliability properties on the BPEL models.
Their word did not consider timing requirements.

\paragraph{Surveys}
Finally, composition of Web services has been recently surveyed.
In \cite{OARFCG15} the larger domain of composition of convergent services is surveyed; however, BPEL is still surveyed in this work.
In \cite{LDB16}, a taxonomy of Web service composition is provided with different directions surveyed such as language, knowledge reuse, automation, tool support, execution platform or target user.

\section{Conclusion and future work}\label{sec:conclusion}

\subsection{Conclusion}
We have presented a model-based approach for synthesizing local time constraints for component services of a composite service~$\CS$, knowing its global time requirement.
Our approach makes use of parameterized timed techniques.
\LongVersion{
The local time constraint can guarantee the satisfaction of the global response time requirement.
Our proposed techniques consist of static and dynamic checking of global time requirement based on the \dLTC{} and \rLTC{} of component services respectively.
}

We first proposed a design-time synthesis algorithm, that utilizes the parametric constraints from the LTS, to synthesize static local time constraint (\dLTC{}) for component services.
The \dLTC{} is then used to select a set of component services that could collectively satisfy the global time requirement in design time.

Then, \LongVersion{during the runtime of a composite service, }we use the runtime information to weaken the \dLTC{}, which becomes the refined local time constraint (\rLTC{}).
In particular, two pieces of runtime information have been leveraged---the execution path that has been taken by the composite service, and the elapsed time of the composite service.
The \rLTC{} is then used to validate whether the composite service can still satisfy the global time requirement at runtime.

As a proof of concept, we have implemented our approach into a tool \ToolBPEL{}, and applied it to four service composition examples.
Our experiments show that \LongVersion{the computation time is always smaller than our previous approach~\cite{TanICSE13}, and that }the runtime refinement leads to an improvement of the global time requirement, with limited overhead.

\subsection{Future work}

We plan to further improve and develop the technique presented in this paper.

\paragraph{General and dedicated optimizations}
First, %
the goal of our work is to propose a full framework for analyzing composition of Web services using parametric timings; therefore, integrating existing state space reduction techniques is perhaps a more practical work, orthogonal to the original goal of our approach.
Nevertheless,
in order to address huge sets of services, one could use efficient state-of-the-art techniques developed for timed systems or parametric timed systems such as (parametric) data difference bound matrices~\cite{HRSV02,QSW17}, efficient L/U-zone abstractions~\cite{HSW16,NPP18}, convex state merging~\cite{AFS13}, integer-hull abstractions~\cite{JLR15,ALR15}, or abstraction-refinement algorithms~\cite{RSM19}.

\paragraph{Soft deadlines}
Second, we will investigate the usage of soft deadlines that allow to run a service with a delay, possibly with an acceptable penalty.

\paragraph{Constraints satisfiability}
Regarding our implementation, the bottleneck seems to be the satisfiability test using Z3; from our experience, switching to a polyhedra library (such as the Parma Polyhedra Library~\cite{BHZ08}) may give better results, and could help render our work scalable.

\paragraph{Uncertain response times}
Our work so far deals with exact response times.
A different approach would be to consider that the response time should be fulfilled with some probability.
In that setting, the goal would be to synthesize the values for the timing parameters such that the response time is indeed below the threshold with a given probability.
To achieve this, we could reuse recent works involving probabilities and timing parameters (\eg{} \cite{JK14,CDKP14}).
An even more challenging problem would be to combine both kinds of parameters (timing parameters and probabilistic parameters), so as to infer the probability under which the response time can be fulfilled.

\paragraph{Statistical model checking}
Finally, when concurrent systems with or without timing constraints are too huge to be analyzed in an exact manner, a recent trend is to propose non-exact techniques, and notably statistical model-checking.
	This technique could be used for compositions of Web services arguably too large to be analyzed in an exact manner.
	Recent techniques developed in the timed setting (\eg{} \cite{DLLMW11,LL16,MNBDLB18}) could be applied to our formalism.

\newcommand{\CCIS}{Communications in Computer and Information Science}
\newcommand{\ENTCS}{Electronic Notes in Theoretical Computer Science}
\newcommand{\FAC}{Formal Aspects of Computing}
\newcommand{\FI}{Fundamenta Informaticae}
\newcommand{\FMSD}{Formal Methods in System Design}
\newcommand{\IJFCS}{International Journal of Foundations of Computer Science}
\newcommand{\IJSSE}{International Journal of Secure Software Engineering}
\newcommand{\IPL}{Information Processing Letters}
\newcommand{\JLAP}{Journal of Logic and Algebraic Programming}
\newcommand{\JLAMP}{Journal of Logical and Algebraic Methods in Programming} %
\newcommand{\JLC}{Journal of Logic and Computation}
\newcommand{\LMCS}{Logical Methods in Computer Science}
\newcommand{\LNCS}{Lecture Notes in Computer Science}
\newcommand{\RESS}{Reliability Engineering \& System Safety}
\newcommand{\STTT}{International Journal on Software Tools for Technology Transfer}
\newcommand{\ToPNoC}{Transactions on Petri Nets and Other Models of Concurrency}
\newcommand{\TSE}{{IEEE} Transactions on Software Engineering}

\ifdefined\VersionAuthor
	\renewcommand*{\bibfont}{\small}
	\printbibliography[title={References}]
\else
	\bibliographystyle{spmpsci}      %
	\bibliography{icse13bpel}
\fi

\clearpage
\appendix
\section{Operational semantics}\label{appendix:semantics}

Set of rules for the transition relation $\hookrightarrow$

Let $mpick=\mpicknew{\Service_i}{P_i}{a_j}{Q_j}$

		\begin{infrule}
			\derive[\rSInv]
			(\Valuation, \lsInvoke{(\Service)}_x, \Constraint, \Delay) \stackrel{\sequence{\rSInv}}{\hookrightarrow} (\Valuation',\lstop, (x=\param_{\Service})\wedge \timelaps{\Constraint}, D+\param_{\Service})
		\end{infrule}

		\begin{infrule}
		\derive[\rRec]
			(\Valuation, \lreceive{(\Service)}_x, \Constraint, \Delay) \stackrel{\sequence{\rRec}}{\hookrightarrow} (\Valuation', \lstop, (x=\param_{\Service})\wedge \timelaps{\Constraint}, D+\param_{\Service})
		\end{infrule}

		\begin{infrule}
			\derive[\rReply]
			(\Valuation, \lreply{(\Service)}_x, \Constraint, \Delay) \stackrel{\sequence{\rReply}}{\hookrightarrow} (\Valuation', \lstop, (x=0)\wedge \timelaps{\Constraint}, D)
		\end{infrule}

		\begin{infrule}
			\derive[\rAInv]
			(\Valuation, \laInvoke{(\Service)}_x, \Constraint, \Delay) \stackrel{\sequence{\rAInv}}{\hookrightarrow} (\Valuation', \lstop,(x=0)\wedge \timelaps{\Constraint},D)
		\end{infrule}

		\begin{infrule}
			\derive[\rPickOne]
			(\Valuation, mpick_x, \Constraint, \Delay) \stackrel{\sequence{(\rPickOne,i)}}{\hookrightarrow} (\Valuation', P_i, (x=\param_i) \wedge~\funIdle(mpick_x) \wedge  \timelaps{\Constraint}, D+\param_i)
		\end{infrule}
		
		\begin{infrule}
			\derive[\rPickTwo]
			(\Valuation, mpick_x, \Constraint, \Delay) \stackrel{\sequence{(\rPickTwo,j)}}{\hookrightarrow} (\Valuation', Q_j,  (x=a_j) \wedge~\funIdle(mpick_x) \wedge  \timelaps{\Constraint}, D+a_j)
		\end{infrule}

		\begin{infrule}
            \Valuation(b)= \varUnitialized
			\derive[\rCondOne]
			(\Valuation, \mconditional{A}{b}{B}, \Constraint, \Delay) \stackrel{\sequence{\rCondOne}}{\hookrightarrow} (\Valuation', A, \Constraint, \Delay)
		\end{infrule}		
        \begin{infrule}
            \Valuation(b)= \varUnitialized
			\derive[\rCondTwo]
			(\Valuation, \mconditional{A}{b}{B}, \Constraint, \Delay) \stackrel{\sequence{\rCondTwo}}{\hookrightarrow} (\Valuation',B, \Constraint, \Delay)
		\end{infrule}
		\begin{infrule}
             \Valuation(b)= true
			\derive[\rCondThree]
			(\Valuation, \mconditional{A}{b}{B}, \Constraint, \Delay) \stackrel{\sequence{\rCondThree}}{\hookrightarrow} (\Valuation', A, \Constraint, \Delay)
		\end{infrule}

		\begin{infrule}
            \Valuation(b)= false
			\derive[\rCondFour]
			(\Valuation, \mconditional{A}{b}{B}, \Constraint, \Delay) \stackrel{\sequence{\rCondFour}}{\hookrightarrow} (\Valuation',B, \Constraint, \Delay)
		\end{infrule}

		\begin{infrule}
			(\Valuation, A, \Constraint, \Delay) \stackrel{\ruleseq}{\hookrightarrow} (\Valuation', A', \Constraint', \Delay'), A'\neq \lstop
			\derive[\rSeqOne]
			(\Valuation, \msequence{A}{B}, \Constraint, \Delay) \stackrel{\ruleseq + \sequence{\rSeqOne}}{\hookrightarrow} (\Valuation', \msequence{A'}{B}, \Constraint', \Delay')
		\end{infrule}

		\begin{infrule}
			(\Valuation, A, \Constraint, \Delay) \stackrel{\ruleseq}{\hookrightarrow} (\Valuation', \lstop, \Constraint', \Delay')
			\derive[\rSeqTwo]
			(\Valuation, \msequence{A}{B}, \Constraint, \Delay) \stackrel{\ruleseq + \sequence{\rSeqTwo}}{\hookrightarrow} (\Valuation', B, \Constraint', \Delay')
		\end{infrule}

\begin{infrule}
			(\Valuation, A, \Constraint, \Delay) \stackrel{\ruleseq}{\hookrightarrow} (\Valuation', A', \Constraint', \Delay') %
			\derive[\rFlowOne]
			(\Valuation, \minterleave{A}{B}, \Constraint, \Delay) \stackrel{\ruleseq + \sequence{\rFlowOne}}{\hookrightarrow} (\Valuation', \minterleave{A'}{B}, \Constraint' \wedge \funIdle(B), \Delay')
		\end{infrule}

		\begin{infrule}
			(\Valuation, B, \Constraint, \Delay) \stackrel{\ruleseq}{\hookrightarrow} (\Valuation', B', \Constraint', \Delay')%
			\derive[\rFlowTwo]
			(\Valuation, \minterleave{A}{B}, \Constraint, \Delay) \stackrel{\ruleseq + \sequence{\rFlowTwo}}{\hookrightarrow} (\Valuation',\minterleave{A}{B'}, \Constraint' \wedge \funIdle(A), \Delay')
		\end{infrule}

{\tiny
\hfill{}Last modification to this document: \today{}.
}

\ifdefined\WithReply
	\input{letter4.tex}
\fi

\end{document}